\documentclass[reqno,11pt,a4paper]{amsart}

\usepackage{amsmath}
\usepackage{amsthm}
\usepackage{mathrsfs}
\usepackage{amsfonts}
\usepackage{graphicx}
\usepackage{epsfig}

\usepackage{amsmath,amssymb,amsfonts,amsthm,amstext,verbatim}
\usepackage{enumerate,color,graphicx,stmaryrd}
\usepackage{psfrag}  
\usepackage{url}
\usepackage{mdwlist}   
\usepackage{mathtools} 
\usepackage{wasysym}   
\usepackage{mathrsfs}
\usepackage{epsfig}
\usepackage{subcaption}
\usepackage{hyperref}
\definecolor{colorlinks}{RGB}{0, 24, 168}
\definecolor{colorcites}{RGB}{124, 10, 2}
\hypersetup{
	colorlinks=true,
	linkcolor=colorlinks,
	citecolor=colorcites,
	urlcolor=colorlinks,
	pdfborder={0 0 0}
}

\usepackage[a4paper,left=3cm,right=3cm,top=3cm,bottom=3cm]{geometry}

\newcommand\tempskiped[1]{#1}

\newtheorem{theorem}{Theorem}
\newtheorem{proposition}{Proposition}[section]
\newtheorem{lemma}[proposition]{Lemma}
\newtheorem{corollary}[proposition]{Corollary}
\theoremstyle{remark}
\newtheorem{definition}[proposition]{Definition}
\newtheorem{remark}[proposition]{Remark}

\newcommand{\calE}{\mathcal{E}}
\newcommand{\calF}{\mathcal{F}}

\newcommand{\calT}{\mathcal{T}}

\newcommand{\calV}{\mathcal{V}}

\newcommand{\bbC}{\mathbb{C}}

\newcommand{\bbE}{\mathbb{E}}

\newcommand{\bbP}{\mathbb{P}}

\newcommand{\bbR}{\mathbb{R}}

\newcommand{\bbZ}{\mathbb{Z}}

\newcommand{\Vertices}{\mathcal{V}}
\newcommand{\Edges}{\mathcal{E}}
\newcommand{\Faces}{\mathcal{F}}

\newcommand{\sle}{\mathrm{SLE}}

\newcommand{\sgn}{\mathrm{sgn}}

\newcommand{\loops}{\mathrm{loops}}
\newcommand{\loopsU}{\mathrm{loops}_{[u]}}
\newcommand{\schw}{\mathcal{S}}
\newcommand{\edges}{\mathrm{edges}}

\newcommand{\R}{\mathbb{R}}
\newcommand{\C}{\mathbb{C}}
\newcommand{\T}{\mathbb{T}\mathrm{ri}}
\newcommand{\Hex}{\mathrm{Hex}}

\renewcommand{\H}{\mathbb{H}}

\newcommand{\Te}{T_{\mathrm{edge}}}       \newcommand{\Ce}{c_{\mathrm{edge}}}
\newcommand{\Tm}{T_{\mathrm{mid}}}        
\newcommand{\CmT}{c_T}
\newcommand{\CmR}{c_R}

\newcommand{\Tgeom}{\mathcal{T}}

\newcommand{\Teta}{T^{[\eta]}}
\newcommand{\Reta}{R^{[\eta]}}

\newcommand{\Trho}{T^{[\rho]}}
\newcommand{\Rrho}{R^{[\rho]}}
\newcommand{\Tmrho}{T^{\left[-\rho \right]}}
\newcommand{\Rmrho}{R^{\left[-\rho \right]}}

\newcommand{\Tmu}{T^{[\mu]}}

\newcommand{\Ttau}{T^{[\tau]}}
\newcommand{\Tmtau}{T^{[-\tau]}}
\newcommand{\Rtau}{R^{[\tau]}}  

\renewcommand{\b}{\mathfrak{b}}
\newcommand{\dob}{{b,b'}}

\newcommand{\E}{\mathbb{E}_\Omega}

\newcommand{\Edelta}{\mathbb{E}_{\Omega_\delta}}

\newcommand{\OmegaDual}{\Omega^{\mathrm{dual}}}
\newcommand{\OmegaE}{\OmegaDual_{\mathrm{edge}}}

\newcommand{\OmegaU}{{[\Omega,u]}}
\newcommand{\OmegaUprime}{{[\Omega',u']}}

\newcommand{\weOne}{w_{\rho}}
\newcommand{\weTwo}{w_{\varrho\rho}}
\newcommand{\weThree}{w_{\varrho^2\rho}}
\newcommand{\weFour}{w_{\!-\rho}}
\newcommand{\weFive}{w_{\!-\varrho\rho}}
\newcommand{\weSix}{w_{\!-\varrho^2\rho}}

\newcommand{\seOne}{\sigma_{\rho}}
\newcommand{\seTwo}{\sigma_{\varrho\rho}}
\newcommand{\seThree}{\sigma_{\varrho^2\rho}}
\newcommand{\seFour}{\sigma_{\!-\rho}}
\newcommand{\seFive}{\sigma_{\!-\varrho\rho}}
\newcommand{\seSix}{\sigma_{\!-\varrho^2\rho}}

\newcommand{\weta}{w^{\left[\eta \right]}}

\newcommand{\wrho}{w^{\left[\rho \right]}}
\newcommand{\wmrho}{w^{\left[-\rho \right]}}
\newcommand{\wup}{\weta_{\vphantom{\mathrm{d}}\mathrm{up}}}
\renewcommand{\wup}{\wrho_{\vphantom{\mathrm{d}}\mathrm{up}}}
\newcommand{\wdown}{\weta_{\mathrm{down}}}
\renewcommand{\wdown}{\wrho_{\mathrm{down}}}

\newcommand{\wpdown}{\wpmu_{\mathrm{down}}}
\newcommand{\wpup}{\wpmu_{\mathrm{up}}}

\newcommand{\wptau}{w^{\prime[\tau]}}                   
\newcommand{\wpmtau}{w^{\prime[-\tau]}}                 
\renewcommand{\wpdown}{\wptau_{\mathrm{down}}}            
\renewcommand{\wpup}{\wptau_{\mathrm{\vphantom{d}up}}}    

\newcommand{\wind}{\mathrm{wind}}
\newcommand{\windU}{\mathrm{wind}_{[u]}}
\newcommand{\phase}{\phi}
\newcommand{\phaseU}{\phi_{[u]}}
\newcommand{\weight}{\mathrm{w}}

\newcommand{\F}{F_\Omega}       \newcommand{\Fdelta}{F_{\Omega_\delta}}

\newcommand{\FU}{F_\OmegaU}     \newcommand{\FdeltaU}{F_{[\Omega_\delta,u_\delta]}}
                                
\newcommand{\feta}{f_{\Omega}^{[\eta]}}
\newcommand{\fetaprime}{f_{\Omega'}^{[\eta']}}

\newcommand{\f}{f_{\Omega}}
\newcommand{\fstar}{f_{\Omega}^\star}
\newcommand{\freg}{f_{\Omega}^\dag}
\newcommand{\fprime}{f_{\Omega'}}
\newcommand{\fstarprime}{f_{\Omega'}^\star}

\newcommand{\fstarH}{f^\star_{\H}}
\newcommand{\fmu}{f_{\Omega}^{[\mu]}}
\newcommand{\fetaU}{f_\OmegaU^{[\eta]}}

\newcommand{\fU}{f_\OmegaU}
\newcommand{\fstarU}{f^\star_\OmegaU}
\newcommand{\fregU}{\fU^\dag}
\newcommand{\spinor}{g_\OmegaU}
\newcommand{\spinorReg}{\spinor^\dag}
\newcommand{\spinorprime}{g_\OmegaUprime}

\newcommand\opsi{\psi^\star}

\newcommand{\conf}{\mathrm{Conf}}

\newcommand{\dOne}{\partial_1}
\newcommand{\dTwo}{\partial_2}
\newcommand{\sOne}{\varsigma_1}
\newcommand{\sTwo}{\varsigma_2}

\renewcommand\Re{\operatorname{Re}}
\renewcommand\Im{\operatorname{Im}}

\newcommand\Corr[1]{\langle #1 \rangle}
\newcommand\CorrO[2][+]{\Corr{#2}_\Omega^{#1}}

\newcommand\CorrFerm[2]{\prec\! #2 \!\succ_{#1}{}}

\newcommand\ConfO[1]{\conf_\Omega(#1)}
\newcommand\ConfOe{\conf^{[e]}_\Omega}

\newcommand\ConfOs[1]{\conf^{\mathrm{spin}}_\Omega(#1)}
\newcommand\ConfOl[1]{\conf^{\mathrm{loop}}_\Omega(#1)}

\newcommand{\tedge}{t^{[e]}}

\def\cZ{\mathcal Z}

\def\Weight{\mathrm{W}}

\def\LambdaU{U}

\begin{document}

\title{Discrete stress-energy tensor in the loop~$O(n)$ model}
\date{\today}

\author[Dmitry Chelkak]{Dmitry Chelkak$^\mathrm{A}$}
\email{ddchelkak@umich.edu}
\author[Alexander Glazman]
{Alexander Glazman$^\mathrm{B}$}
\email{alexander.glazman@uibk.ac.at}
\author[Stanislav Smirnov]
{Stanislav Smirnov$^\mathrm{C,D}$}
\email{stanislav.smirnov@unige.ch}

\thanks{\textsc{${}^\mathrm{A}$ University of Michigan, USA}}

\thanks{\textsc{${}^\mathrm{B}$ Universität Innsbruck, Austria.}}

\thanks{\textsc{${}^\mathrm{C}$
D\'epartement de Math\'ematiques, Universit\'e de Gen\`eve, Switzerland.}}
\thanks{\textsc{${}^\mathrm{D}$ Chebyshev Laboratory, Saint-Petersburg State University, Russia}}


\begin{abstract}
We study the loop $O(n)$ model on the honeycomb lattice. By means of local non-planar deformations of the lattice, we construct a discrete stress-energy tensor. 
For $n\in [0,2]$, it gives a new observable satisfying a part of Cauchy-Riemann equations. We conjecture that it is approximately discrete-holomorphic and converges to the stress-energy tensor in the continuum, which is known to be a holomorphic function with the Schwarzian conformal covariance. 
In support of this conjecture, we prove it for the case of~$n=1$ which corresponds to the Ising model. Moreover, in this case, we show that the correlations of the discrete stress-energy tensor with primary fields converge to their continuous counterparts, which satisfy the OPEs given by the CFT with central charge~$c=1/2$.

Proving the conjecture for other values of~$n$ remains a challenge. In particular, this would open a road to establishing the convergence of the interface to the corresponding~$\mathrm{SLE}_\kappa$ in the scaling limit. 
\end{abstract}

\maketitle

\tableofcontents

\section{Introduction}
Ising model was introduced by Lenz in 1921 to describe the phenomenon of ferromagnetism~\cite{Len20}.
It is arguably the most studied model of statistical mechanics, see~\cite[Chapter~3]{FriVel17} for a brief introduction. 
The interest in the Ising model is due to its relatively simple definition allowing to a very rich behaviour.
In any dimension~$d\geq 2$, the model undergoes an order/disorder phase transition, due to a classical argument of Peierls~\cite{Pei36}.
We focus on the planar case~$d=2$, where the Conformal Field Theory (CFT) was introduced in the seminal works by Belavin, Polyakov, Zamolodchikov~\cite{BelPolZam84,BelPolZam84a}. 
Based on the ideas of the Renormalization Group Theory asserting scale invariance, the CFT puts a bridge between two-dimensional minimal models of statistical mechanics and a family of Virasoro algebras (e.g., see~\cite{DiFMatSen97} or~\cite{Mus10}). 
The latter are Li algebras satisfying certain relations. 
These Virasoro algebras are expected to describe the whole hierarchy of limits of discrete correlation functions as the mesh-size of the lattice tends to zero.
The CFT corresponding to the critical Ising model has the central charge~$c=1/2$.

Before 2000s, this passage was usually considered in the infinite-volume (or half-infinite-volume) setup. 
During the two last decades new techniques appeared, which led to a number of rigorous convergence results (confirming the existing CFT predictions) for various correlators in the critical Ising model considered on discrete approximations to a \emph{general} planar domain.
The first results of this kind were obtained by the third author in~\cite{Smi06,Smi10} for holomorphic observables in the so-called FK-Ising model (also known as the random cluster model with~$q=2$). 
In a similar manner, one can view these observables as solutions to some \emph{well-posed} discrete Riemann-type boundary value problems (see Remark~\ref{rem:bvp_in_discrete}). This opens a way for the treatment of their limits as the mesh-size tends to zero in general planar domains.
Further developments of the techniques proposed in~\cite{Smi06,Smi10} include the convergence results for the basic fermionic observables~\cite{CheSmi11}, energy density correlations~\cite{HonSmi11,Hon10}, spin correlations~\cite{CheIzy13,CheHonIzy15} and mixed correlations~\cite{CheHonIzy21}; the latter two works rely on the convergence of discrete spinor observables, see Section~\ref{sub:T-spin-via-spinors} below for their definition.

Still, it remains a challenge to find a direct connection between the hierarchy of fields provided by the CFT and discrete correlation functions in the Ising model. The stress-energy tensor is an operator of an infinitesimal change of the metric. It plays a crucial role in the hierarchy of fields in the CFT as it allows relate different fields to each other. The goal of this work is to provide a geometric construction of a discrete stress-energy tensor on the hexagonal lattice that is partially discrete holomorphic and converges to the continuous counterpart that has Schwarzian conformal covariance. We also show that the correlations of the discrete stress-energy tensor and other fields satisfy the Operator Product Expansions (OPEs) predicted by the CFT.

Our construction of the stress-energy tensor extends to the critical loop~$O(n)$ model on the hexagonal lattice.
This is the model on collections of non-intersecting cycles introduced in~1981 by Domani et al~\cite{DomMukNie81}.
The model has two real parameters~$n,x>0$ and has a rich conjectured phase diagram 
In particular, Nienhuis~\cite{Nie82} non-rigorously derived that a phase transition in terms of loop lengths occurs when~$n\leq 2$ at~$x_c(n) =1/\sqrt{2+\sqrt{2-n}}$: there should be macroscopic when~$x\geq x_c(n)$ (with the Russo--Seymour--Welsh estimates) and exhibit exponential tails when~$x<x_c(n)$.
This has been established rigorously only at~$n=1$ (Ising model)~\cite{AizBarFer87,Tas16}.
There is also a way to extend the model to~$n=0$; this corresponds to the self-avoiding walk and the critical point (inverse connective constant) was confirmed to be~$1/\sqrt{2+\sqrt{2}}$~\cite{DumSmi12}.
The macroscopic behaviour has been recently established at~$x_c(n)$ when~$n\in [1,2]$~\cite{DumGlaPel21}, at~$n=2$ when~$x\in [x_c(2),1]$~\cite{GlaMan21,GlaLam23} and in the vicinity of~$n=x=1$ (critical site percolation on the triangular lattice)~\cite{CraGlaHarPel20}.
When~$n>2$, the loops are expected to exhibit exponential decay at all values of~$x>0$.
The work~\cite{DumPelSamSpi17} has proved this when~$n$ is large enough and established a different type of a transition: from a dilute (when~$nx^6$ is small) to a dense regime (when~$nx^6$ is large).
At integer values of~$n$, there are heuristic connections to the spin~$O(n)$ model.
The latter extends the Ising model ($n=1$) to spin systems with values on the~$n$-dimensional sphere: $n=2$ is the XY model with the seminal Berezinskii--Kosterlitz--Thouless transition~\cite{Ber71,KosTho72} established by Fröhlich and Spencer~\cite{FroSpe81} (see also~\cite{LisEng23} for an elementary argument that also proves sharpness, it relies on~\cite{Lam21});
$n=3$ is the Heisenberg model, for which the Polyakov's conjecture asserts exponential decay at all temperatures~\cite{Pol75}.
We refer to~\cite{PelSpi17} for a survey on the loop~$O(n)$ and the spin~$O(n)$ models.

Nienhuis~\cite{Nie84} non-rigorously related the loop~$O(n)$ model to the Coulomb gas, see~\cite{Mus10} for more details on the subject.
This allowed him to derive the exact values of the universal critical exponents.
Expressing the central charge corresponding to the loop~$O(n)$ model as a function of~$n$, one obtains
\begin{align}
\label{eq:intro-central-charge}
	c = (6-\kappa)(3\kappa - 8) /2\kappa,
	\quad
	\text{where}
	\quad
	\kappa = 
	\begin{cases}
		\tfrac{4\pi}{2\pi-\arccos(-n/2)} &\text{ if } x=x_c(n),\\
		\,\,\,\,\tfrac{4\pi}{\arccos(-n/2)} &\text{ if } x>x_c(n).\\
	\end{cases}	
\end{align}
In particular, this description would mean that the scaling limit of the model is described by the Conformal Loop Ensemble (CLE) of parameter~$\kappa$ (see eg. \cite[Section~5.6]{KagNie04}).

The same conjecture can be obtained by considering the parafermionic observables introduced by the third author in the case of
the random cluster model in~\cite{Smi06,Smi10}:
\begin{align}
\label{eq:intro-fermion}
	F_\Omega(b,z;\sigma) = \sum_{\gamma\,:\,b\to z}{n^{\#\loops(\gamma)}x^{\#\edges(\gamma)}e^{-i\sigma\wind(\gamma)}}\,,
\end{align}
where the sum is taken over all non-intersecting collections of loops and a path linking two fixed bound points~$b$ and~$z$ on a given domain~$\Omega$ on the hexagonal lattice,
$\wind(\gamma)$ denotes the full rotation of this path as one goes from~$b$ to~$z$ and~$\sigma$ is determined by~$n$ and~$x$:
\[
	\begin{cases}
		\sigma = 1 - \tfrac{3}{4\pi}\arccos(-n/2) & \text{ if } x = \tfrac{1}{\sqrt{2+\sqrt{2-n}}} = x_c(n),\\
		\sigma = 1 - \tfrac{3}{4\pi} (2\pi-\arccos(-n/2)) & \text{ if } x = \tfrac{1}{\sqrt{2-\sqrt{2-n}}} = \tilde{x}_c(n).\\
	\end{cases}
\]
This observable satisfies a part of the discrete Cauchy-Riemann equations and its complex phase at each boundary point can be determined.
This led the third author~\cite{Smi06} to conjecture that the parafermionic observable, when properly normalised, converges to the unique holomorphic solution of the corresponding boundary value problem in the continuous domain:
\begin{align}
\label{eq:intro-parafermion-convergence}
\delta^{-\sigma}\cdot F_\Omega(b,z)/\cZ_\Omega^{b,b'} \to \mathcal{C}\cdot (\varphi'/\varphi)^\sigma\,,
\end{align}
where~$\mathcal{C}$ is a lattice dependent constant and~$\varphi$ is a conformal map from~$\Omega$ onto the upper half-plane.
Note that the lefthand side in~\eqref{eq:intro-parafermion-convergence} is a martingale with respect to the growing interface and~$ (\varphi'/\varphi)^\sigma$ is a martingale with respect to~$\sle_\kappa$ for~$\kappa = 3/(\sigma + 1/2)$ (one needs to consider slit domains, see~\cite{Smi06}).
Replacing~$\sigma$ by its expression in terms of~$n$, one gets the same value of~$\kappa$ as the one derived via the Coulomb gas formalism.

The conjecture~\eqref{eq:intro-parafermion-convergence} was proven for the critical Ising model ($n=1,x=1/\sqrt{3}$).
The work in this direction was originated by the third author~\cite{Smi10} who proved the same statement for the FK Ising model.
A series of papers~\cite{CheSmi11,CheSmi12,DumHonNol11,KemSmi12} by a group of authors led to~\cite{CheDumHonKemSmi14}
establishing convergence of the interface to the~$\sle_3$.
Conformal invariance was also proven for the critical site percolation on the triangular lattice ($n=x=1$) by establishing convergence of a similar discrete observable~\cite{Smi01,CamNew06} (see also~\cite{KhrSmi21}).
While conformal invariance and convergence of the parafermionic observable remain open for~$n\neq 1$,
the observable has proven to be useful in the study of the phase diagram~\cite{DumSmi12,DumGlaPel21}.
This motivates the interest in other combinatorial observables satisfying local relations
and having (conjectural) conformally covariant limits, especially in the case of the self-avoiding walks where
the conformal invariance remains a big challenge for the probabilistic community.

In this work, we introduce a novel observable for the loop~$O(n)$ model that satisfies a part of Cauchy--Riemann equations and that we expect to behave as a discrete stress-energy tensor; we prove this rigorously for~$n=1,x=1/\sqrt{3}$ that corresponds to the critical Ising model on the triangular lattice.
A discrete analogue of the deformation of a metric is a modification of the underlying graph and of the coupling constants. 
In the loop~$O(n)$ model, Nienhuis discovered~\cite{Nie90} a family of integrable weights satisfying the Yang-Baxter equation that extends the model to general rhombic tilings in an integrable way.
This allows to define a discrete stress-energy tensor as an operator describing a particular local infinitesimal deformation of the tiling (see Section~\ref{sec:combi} for the details).
In a particular case of the Ising model, the resulting object can be written as a polynomial of several spins.
Compared to an independent concurrent work of Hongler, K{\"y}t{\"o}l{\"a} and Viklund~\cite{HonKytVik22}, we construct the discrete stress-energy tensor as an explicit local field.

\subsection{Loop~$O(n)$ model: observables and relations}
\label{sub:results-o-n}

Consider the hexagonal lattice of mesh-size~$\delta$ whose faces are centred at~$\{\delta(k \cdot e^{\pi i/6}+\ell \cdot e^{5\pi i/6})\colon k,\ell\in\bbZ\}$; we denote the lattice by~$\Hex_\delta$ and identify its faces with their centres.
For a simply-connected domain~$\Omega$, we denote its boundary by~$\partial\Omega$.
Among all cycles on~$\Hex_\delta$ contained in~$\Omega$ consider the one that surrounds the largest area (i.e. number of faces).
We call the subgraph of~$\Hex_\delta$ consisting of vertices and edges surrounded by this cycle (or contained in it) a $\delta$-approximation of~$\Omega$ and denote it by~$\Omega_\delta$; see Fig.~\ref{fig:intro-notations}.
Denote the sets of vertices, edges and faces in~$\Omega_\delta$ by $\Vertices_\delta(\Omega)$, $\Edges_\delta(\Omega)$ and~$\Faces_\delta(\Omega)$ respectively.
Denote the set of edges and vertices on boundary cycle of~$\Omega_\delta$ by~$\partial\calE_\delta(\Omega)$ and~$\partial\Vertices_\delta(\Omega)$ respectively.
For brevity we often omit~$\delta$ in the notation.

Let~$b,b'\in \partial\calV(\Omega)$ be two such vertices that have degree three in~$\Omega$.
Define two sets of configurations:
\begin{itemize}
	\item $\ConfOl{\varnothing}$ is the set of all spanning subgraphs of~$\Omega\setminus \partial\calE(\Omega)$, in which each vertex has an even degree ($0$ or~$2$).
	\item $\ConfOl{b,b'}$ is the set of all spanning subgraphs of~$\Omega\setminus \partial\calE(\Omega)$, in which~$b$ and~$b'$ have degree~$1$ and all other vertices have even degree.
\end{itemize}
Take~$\b\in \{\varnothing , \{b,b'\}\}$.
The loop~$O(n)$ model under~$\b$ boundary conditions is a probability measure on~$\ConfOl{\varnothing}$ defined by
\begin{equation}\label{eq:def-loop-o-n}
	{\sf Loop}^\b_{\Omega_\delta,n,x}(\omega) = \tfrac{1}{\cZ(\Omega,n,x)}n^{\#\loops(\omega)}x^{\#\edges(\omega)},
\end{equation}
where~$\#{\rm loops}(\omega)$ and~$\#{\rm edges}(\omega)$ count the number of loops and edges in~$\omega$ respectively and~$\cZ^\b(\Omega,n,x)$ is the unique normalising constant (called {\em partition function}) that renders~${\sf Loop}_{\Omega,n,x}$ a probability measure.

We restrict to the case~$n\in[0,2]$ and suggest a new combinatorial observable at~$x\in \{x_c(n),\tilde{x}_c(n)\}$.
It is obtained by evaluating the partition function after a local deformation of the lattice along the integrable line.
To the best of our knowledge there is no known integrable way to define the loop~$O(n)$ model on a deformed hexagonal lattice
by adding non-homogeneous weights on the edges.
On the other hand, one can switch to the dual lattice and view the loop~$O(n)$ model as a model on the plane tiled with equilateral triangles.
Gluing a pair of adjacent triangles, we get a model on the plane tiled with equilateral triangles and a rhombus with angle~$\pi/3$:
each non-empty triangle contains exactly one arc of a loop and has weight~$x$;
the rhombus is either empty or contains one or two arcs and its weight is a function of the total length of the arcs (see Fig.~\ref{figConfE}).
Allowing for one more local configuration, Nienhuis discovered~\cite{Nie90} a family of integrable weights of rhombuses
satisfying the Yang-Baxter equation, the weights~$(u_1, u_2, v, w_1, w_2)_\theta$ being parametrized by the angle~$\theta$ of a rhombus\footnote{The works~\cite{Nie90} and~\cite{IkhCar09} contain minor typos in the weights, see~\cite{Gla13} for the correction.}:

\begin{figure}
	\includegraphics[width=0.8\textwidth]{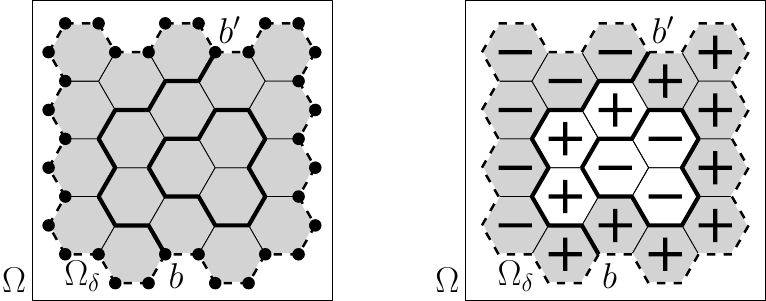}
	\caption{
	\textsc{Left:} 
	Simply-connected domain~$\Omega$ (square) with faces of~$\Omega_\delta$ ($\delta$-approximation of~$\Omega$) in gray.
	The sets of boundary vertices~$\partial\Vertices_\delta(\Omega)$ (bold dots) and edges~$\partial\calE_\delta(\Omega)$ (dashed).
	The Dobrushin conditions are defined by two boundary vertices~$b,b'$ (identified with edges of~$\calE_\delta(\Omega)\setminus \partial\calE_\delta(\Omega)$ incident to them).
	In bold is a sample of a loop configuration in~$\ConfOl{b,b'}$ with a path (called {\em interface}) linking~$b$ and~$b'$.
	\textsc{Right:} The (boundary) faces in~$\partial\calF_\delta(\Omega)$ are in gray.
	The spins~$+/-$ define a spin configuration in~$\ConfOs{b,b'}$. 
	The domain walls between~$+$ and~$-$ form a loop configuration from the left figure; this defines a bijection between~$\ConfOl{b,b'}$ and~$\ConfOs{b,b'}$.}
	\label{fig:intro-notations}
\end{figure}

\begin{align}
\label{eq_u1}
u_1&=\tfrac{1}{t}\cdot\sin[(\pi-\theta)(\sigma-1)]\cdot\sin[\tfrac{2\pi}{3}(\sigma - 1)]\,,\\
\label{eq_u2}
u_2&=\tfrac{1}{t}\cdot\sin[\theta (\sigma - 1)]\cdot\sin[\tfrac{2\pi}{3}(\sigma - 1)]\,,\\
\label{eq_v}
v&=\tfrac{1}{t}\cdot\sin[\theta (\sigma - 1)]\cdot\sin[(\pi-\theta)(\sigma - 1)]\,,\\
\label{eq_w1}
w_1&=\tfrac{1}{t}\cdot\sin[(\tfrac{2\pi}{3}-\theta)(\sigma-1)]\cdot\sin[(\pi-\theta)(\sigma - 1)]\,,\\
\label{eq_w2}
w_2&=\tfrac{1}{t}\cdot\sin[(\theta-\tfrac{\pi}{3})(\sigma - 1)]\cdot\sin[\theta (\sigma - 1)]\,,\\
\text{where} \quad
t&=\frac{\sin^3[\tfrac{2\pi}{3}(\sigma - 1)]}{\sin[\tfrac{\pi}{3}(\sigma - 1)]}+
\sin[(\theta-\tfrac{\pi}{3})(\sigma - 1)]\cdot \sin[(\tfrac{2\pi}{3}-\theta)(\sigma - 1)]. \nonumber
\end{align}
It was pointed out by Ikhlef and Cardy~\cite{IkhCar09} that for these weights the parafermionic observable still satisfies (as it does on the hexagonal lattice) a part of the Cauchy-Riemann equations; see~\cite{AlaBat14} for a connection between the two approaches.
This brings us to a definition of an integrable loop~$O(n)$ model on a Riemann surface tiled by rhombi and equilateral triangles
and allows to deform each rhombus (see Section~\ref{sub:stress-tensor-as-deformation} for details).

The observable~$\Te$ on edges of the hexagonal lattice is then defined as derivative of the partition function under the following perturbation of the original triangular tiling:
replacing two adjacent triangles by a rhombus of angle~$\theta=\pi/3 +\varepsilon$.
This deformation corresponds to the edges of the hexagonal lattice.
Alternatively, $\Te$ can be written as an expectation over the loop~$O(n)$ measure of a certain operator that depends only on how the configuration looks like near a particular edge and on the connectivity pattern on the loops passing through it.

\begin{definition}\label{def:t-e-loop-o-n}
	Take any~$n\in [0,2]$ and~$x\in \{x_c(n),\tilde{x}_c(n)\}$.
	Let~$e=uv$ be a non-boundary edge of~$\Omega$.
	Denote the edges incident to~$u$ and~$v$ by~$uw_1,uw_2,vw_3,vw_4$ in the counter-clockwise order.
	Define
	\[
		\Te^\b (e) =  \Ce + \bbE_{\Omega,n,x}[\tedge(\omega)]\,,  \text{ where }
	\]
	\[
		\tedge(\omega) =
		\begin{cases}
			u_1'(\tfrac{\pi}{3})/u_1(\tfrac{\pi}{3}) & \text{if} \quad \mathrm{deg}_\gamma(u) \neq \mathrm{deg}_\gamma(v),\\
			u_2'(\tfrac{\pi}{3})/u_2(\tfrac{\pi}{3}) & \text{if} \quad uw_1, uv, vw_4 \in \gamma \text{ or }  uw_2, uv, vw_3 \in \gamma,\\
			v'(\tfrac{\pi}{3})/v(\tfrac{\pi}{3}) & \text{if} \quad uw_1, uv, vw_3 \in \gamma \text{ or }  uw_2, uv, vw_4 \in \gamma,\\
			(w_1'(\tfrac{\pi}{3})+w_2'(\tfrac{\pi}{3})n)/w_1(\tfrac{\pi}{3}) & \text{if} \quad uw_1, uw_2, vw_3, vw_4 \in \gamma \text{ and } 
			u\xleftrightarrow{\gamma} v,\\
			(w_1'(\tfrac{\pi}{3})+w_2'(\tfrac{\pi}{3})n^{-1})/w_1(\tfrac{\pi}{3}) & \text{if} \quad uw_1, uw_2, vw_3, vw_4 \in \gamma \text{ and } 
			u\not\xleftrightarrow{\gamma} v
		\end{cases}
	\]
	and~$\Ce=\Ce(n,x)$ is the constant that ensures that~$\Te^\b \to 0$ when~$\Omega_k\nearrow \Hex$.
\end{definition}

\tempskiped{\begin{figure}
    \begin{center}
      \includegraphics[width=\textwidth]{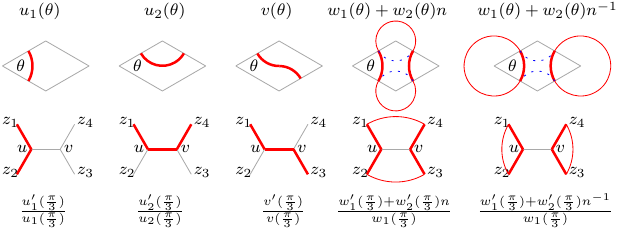}
    \end{center}
    \caption{
    Configurations in a rhombus corresponding to the edge~$e=(u,v)$
    together with their images in~$\ConfOe(\b)$, their weights (at the top) and the coefficients from Definition~\ref{def:t-e-loop-o-n} (at the bottom).}
    \label{figConfE}
\end{figure}}

\begin{remark}
	Though~$\Ce$ should not depend on the sequence of domains~$(\Omega_k)_{k\geq 1}$, this has not been shown in the full generality so far.
	This will follow once all thermodynamic limits are shown to be the same.
	This is known only in the region~$n\geq 1, x \leq 1/\sqrt{n}$~\cite{GlaMan23} that covers a part~$n\in [1,2],x=x_c(n)$ of the critical curve (here also the Russo--Seymour--Welsh estimates have been established~\cite{DumGlaPel21}).
	In the current work, we compute~$\Ce$ exactly for the critical Ising model ($n=1,x=1/\sqrt{3}$).
\end{remark}

\begin{remark}
	In our definition of the loop~$O(n)$ model on a deformed lattice we keep the loop-weight equal to~$n$
	and change only the edge-weight~$x$.
	Note that this is consistent with standard considerations of the loop~$O(n)$ model on the cylinder, e.g. see~\cite{Car13, Car96}.
	It is probably necessary to deform the loop-weight to understand better the relation with the Coulomb gas
	following the idea of Nienhuis~\cite{Nie82} of describing the full weight of a configuration as product of local weights.
	In this case the derivative of the partition function will get an additional term corresponding to varying the loop-weight.
	We do not pursue this direction here.
\end{remark}

The real-valued observable~$\Te$ describes the respond of (the partition function of) the loop $O(n)$ model in a given discrete domain to an infinitesimal change of the discrete complex structure.
We construct a complex-valued observables~$\Tgeom^\b$ which projections are given by~$\Te^\b$.
In field theories the corresponding operator is called the stress-energy tensor.
Hence, we refer to the observables~$\Te^\b$ as a discrete stress-energy tensor.
Although one could add any null-field without affecting the OPEs in the limit stated below, the definition given above has two benefits: clear geometrical interpretation and local relations stated in the following proposition. 
These relations can be viewed as a part of s-holomorphicity (or propagation) equations that imply the Cauchy-Riemann relations, see Section~\ref{sub:fermion}.

\begin{proposition}\label{prop:relation-Te}
	Let~$a$, $b$, $c$, $d$, $e$, $f$ be the edges of a hexagon of~$\Omega$ (listed in the counter-clockwise order). Then the following holds:
	\begin{align}\label{eq_relation_Te}
		\Te^{\b}(d)-\Te^{\b}(a)=\Te^{\b}(b)-\Te^{\b}(e)=\Te^{\b}(f)-\Te^{\b}(c)\,.
	\end{align}
\end{proposition}

According to the conformal field theory, the stress-energy tensor has the Schwarzian conformal covariance:
\begin{align*}
\langle T(w)O_1(\varphi(z_1))\dots O_k(\varphi(z_k))\rangle_\Omega
&=\left[ (\varphi'(w))^2\langle T(\varphi(w)) O_1(\varphi(z_1))\dots O_k(\varphi(z_k)) \rangle_{\Omega'} \right.\\
&\left.+ \tfrac{c}{12}[\schw\varphi](w) \langle O_1(\varphi(z_1))\dots O_k(\varphi(z_k))\rangle_{\Omega'}\right]
\prod_{j=1}^k |\varphi'(z_j)|^{\Delta_j}\,,
\end{align*}
where $[\schw\varphi](w)\,:=\,\frac{\varphi'''(w)}{\varphi'(w)}-\frac{3}{2}\left[\frac{\varphi''(w)}{\varphi'(w)}\right]^2$
is the Schwarzian derivative of a conformal mapping $\varphi \, : \, \Omega \to \Omega'$,
and~$O_j(z_j)$ stands for a real-valued primary field at~$z_j\in\Omega\setminus\{w\}$ with a scaling exponent~$\Delta_j$,
and~$c$ is the central charge of the theory.

As we already mentioned above, the loop~$O(n)$ model is conjectured to have a conformally invariant scaling limit which can be described by the CFT with a central charge given by eq.~\eqref{eq:intro-central-charge}.
Hence, we conjecture that~$\Tgeom^\dob$, when properly normalized, converges to the unique martingale of~$\sle_\kappa$ which has the Schwarzian conformal covariance.

{\bf Conjecture T.} There exists a lattice dependent constant~$\mathcal{C}$ such that for~$\sigma = 3/\kappa - 1/2$
\begin{align*}
\delta^{-2}\Tgeom &\xrightarrow[\delta\to 0]{} \mathcal{C}\cdot\frac{c}{12}[\schw\varphi](w)\,,\\
\delta^{-2}\Tgeom^{b,b'} &\xrightarrow[\delta\to 0]{} \mathcal{C}\cdot\left(\sigma\left[\frac{\varphi'(w)}{\varphi(w)}\right]^{2}+\frac{c}{12}[\schw\varphi](w)\right)\,.
\end{align*}

It is easy to see that~$\Te^{b,b'}$ is a martingale with respect to the growing interface from~$b$ to~$b'$.
Hence, the proof of the above conjecture would imply convergence of the interface to the~$\sle_\kappa$.
This is similar to the approach of the parafermionic observables proposed in~\cite{Smi06}.
As a strong support for this conjecture, we prove it for~$n=1$ (Ising model).

We note that on the boundary the conjecture for~$n=0$ agrees with the restriction property of the~$\sle_{8/3}$.
Indeed, in this case one has~$c=0$ and Conjecture~T means that~$\Tgeom^{b,b'}$, after a proper rescaling,
converges to~$\frac{5}{8}|\varphi'/\varphi|^2\cdot\delta^{2}$.
On the boundary the latter is the probability to intersect the $\delta$-neighbourhood of a boundary point.
The boundary values of~$\Tgeom^{b,b'}$ can be interpreted in a similar way.
Indeed, only paths touching a boundary edge contribute to the value of~$\Tgeom^{b,b'}$ at it.

\subsection{Ising model: convergence statements}

Let~$\Omega$ be a simply-connected domain on~$\bbC$.
Fix~$\delta>0$ and recall that~$\Omega_\delta$ is a~$\Hex_\delta$-approximation of~$\Omega$.
Denote the set faces in~$\Omega_\delta$ adjacent to this cycle by~$\partial\Faces_\delta(\Omega)$.
Define~$\partial\Edges_\delta(\Omega)$ as the set of half-edges separating two faces in~$\partial\Faces_\delta(\Omega)$ and incident to a vertex in~$\Vertices_\delta(\Omega) \setminus \partial\Vertices_\delta(\Omega)$.
As above, we often omit~$\delta$ in the notation for brevity.

We identify each vertex in~$\partial\calV(\Omega)$ that has degree three in~$\Omega$ with the unique edge in~$\Edges_\delta(\Omega)\setminus \partial\calE_\delta(\Omega)$ incident to it.
Let~$b,b'\in \Edges_\delta(\Omega)\setminus \partial\calE_\delta(\Omega)$.
Define two sets of configurations:
\begin{itemize}
	\item $\ConfOs{\varnothing}$ is the set of all possible ways to assign~$\pm 1$ to the faces in $\Faces_\delta(\Omega)$ that assign~$+1$ to all faces in~$\partial\Faces_\delta(\Omega)$.
	\item $\ConfOs{b,b'}$ is the set of all possible ways to assign~$\pm 1$ to the faces in $\Faces_\delta(\Omega)$ such that all faces in $\partial\Faces_\delta(\Omega)$ between~$b$ and~$b'$ in the counter-clockwise order get~$+1$, and all faces between~$b'$ and~$b$ get~$-1$ (see Fig.~\ref{fig:intro-notations}).
\end{itemize}

The Ising model on the faces of~$\Omega_\delta$ at an inverse temperature~$\beta>0$ and under boundary conditions~$\b\in \{\varnothing,\{b,b'\}\}$ is supported on~$\ConfOs{\b}$ and is defined by 
\begin{equation}
	\label{eq:def-ising-hex}
	\bbP_{\Omega_\delta,\beta}^{\b}(\sigma) = \tfrac{1}{\cZ_{\Omega_\delta,\beta}^{\b}}\cdot 
	\exp[\beta\sum_{u\sim v}{\sigma(u)\sigma(v)}],
\end{equation}
where the sum is taken over all pairs of adjacent faces~$u,v\in \Faces_\delta(\Omega)$ and~$\cZ_{\Omega_\delta,\beta}^{\b}$ is the partition function.

We chose to define the Ising model on the faces of the hexagonal lattice instead of vertices of the triangular lattice, since this is more convenient for a geometric construction discussed in Section~\ref{sec:combi}.
The value~$\beta_c=\log \sqrt{3}/2$ is known to be critical for this Ising model, see~\cite{Wan50,AizBarFer87}: the correlations decay exponentially in the distance when~$\beta < \beta_c$ and remain uniformly positive when~$\beta>\beta_c$. 

We consider the six possible orientations of edges of the dual lattice:
\[
	\wp:= \{1,\varrho,\varrho^2,\varrho^3,\varrho^4,\varrho^5\}, 
	\quad
	\text{where } \varrho:= e^{\pi i/3}.
\]

\begin{definition}\label{def:T-via-spins}
	Fix a simply-connected domain~$\Omega$ and a mesh-size~$\delta>0$.
	Let a face~$w\in\Faces_\delta(\Omega)$ be such that the six neighbours of~$w$ are also in~$\Faces_\delta(\Omega)$.
	We index these faces by~$\rho\in\wp$: define~$\weOne:=w - \delta i \rho$; see Fig.~\ref{fig:t-edge-mid-rel}.
	Now fix some~$\rho\in\wp$ and define
	\begin{align*}
		\Trho(w):=   ~\CmT  & + {\textstyle\frac{1}{24\sqrt{3}}}(1\!+\!\sigma_0\seOne)(1\!+\!\sigma_0\seFour)(2\!-\!\sigma_0\seTwo)(2\!-\!\sigma_0\seThree),
		\\
		\Rrho(w) := ~\CmR & - {\textstyle\frac{1}{2\sqrt{3}}}[(1\!-\!\sigma_0\seOne)+(1\!-\!\sigma_0\seFour)]\\
		& + {\textstyle\frac{1}{12\sqrt{3}}}(1\!-\!\seOne\seFour)[(2\!-\!\sigma_0\seTwo)(2\!-\!\sigma_0\seThree)+(2\!-\!\sigma_0\seFive)(2\!-\!\sigma_0\seSix)]\,,
	\end{align*}
	where~$\sigma_0, \seOne, \seTwo, \seThree, \seFour, \seFive, \seSix$ are the spins at~$w, \weOne, \weTwo, \weThree, \weFour, \weFive, \weSix$ respectively and the constants~$\CmT$ and~$\CmR$ are given by
	\begin{equation}
		\label{eq:CmTCmRCmid=}
		\CmT:=\tfrac{3}{2\pi}\! - \!\tfrac{1}{\sqrt{3}}\,,\qquad\CmR:=\tfrac{7}{3\sqrt{3}}\! - \!\tfrac{4}{\pi}.
	\end{equation}
	The projections~$\mathcal{T}^{[\rho]}$ of the discrete stress-energy tensor are defined by
	\[
		\mathcal{T}^{[\rho]}(w):=\Trho(w)+\Tmrho(w)+\Rrho(w)
	\]
	Finally, we define the discrete stress-energy tensor~$T$:
	\[
		T(w) :=~ -\tfrac{2}{3}\cdot \sum_{\rho=1,\varrho^2,\varrho^4}\,\bar{\rho}^2 \,\mathcal{T}^{[\rho]}(w) \,.
	\]
\end{definition}

\begin{remark} \label{rem:T-and-R}
	1) Note that~$\Rrho(w)=\Rmrho(w)$, whence~$\mathcal{T}^{[\rho]}(w) = \mathcal{T}^{[-\rho]}(w)$.
	
	2) We prefer to separate the terms~$\Trho(w)$ and~$\Rrho(w)$ for the following reason: its expectations have a higher order of decay than~$\delta^2$ and hence disappear in the scaling limits considered in Theorem~\ref{thm:Ising1}. Moreover, it has the same effect in Theorems~\ref{thm:Ising2} and~\ref{thm:Ising3}: the terms containing~$\Rrho(w)$ have higher orders of decay comparing to those containing~$\Trho(w)+\Tmrho(w)$ and hence do not contribute to the limits, see Proposition~\ref{prop_via_fermion_t} and Section~\ref{sec:Ising-convergence} for more details.
\end{remark}

We now formulate the first of our convergence results.

\begin{theorem} \label{thm:Ising1}
	Let~$\Omega\subset\C$ be a bounded simply connected domain, $w\in\Omega$ and~$b,b'\in\partial\Omega$.
	For a mesh-size~$\delta>0$, consider~$w_\delta\in\Faces(\Omega_\delta)$ that contains~$w$ and~$b^{\vphantom{\prime}}_\delta,b'_\delta \in\partial\Edges_\delta(\Omega)\setminus \partial\calE_\delta(\Omega)$ that are~$3\delta$-close to~$b,b'$.
	Assume that~$b,b'$ are on straight parts of~$\partial \Omega$ that are orthogonal to the edges~$b^{\vphantom{\prime}}_\delta,b'_\delta$ respectively.
	Then, for any~$\rho\in\{1, e^{\frac{2\pi i}{3} },e^{\frac{4\pi i}{3}}\}$, one has, as~$\delta\to 0$,
\begin{align*}
	\delta^{-2}\Edelta^+[\mathcal{T}^{[\rho]}(w_\delta)]&~\rightrightarrows~\tfrac{3}{\pi}\Re[(i\rho)^2\CorrO{T(w)}], \\
	\delta^{-2}\Edelta^{b_\delta,b'_\delta}[\mathcal{T}^{[\rho]}(w_\delta)]&~\rightrightarrows~\tfrac{3}{\pi}\Re[(i\rho)^2\Corr{T(w)}_\Omega^{b,b'}],
\end{align*}
where the expectation is according to~$\bbP_{\Omega_{\delta},\beta_c}^{\emptyset}$ and~$\bbP_{\Omega_{\delta},\beta_c}^{b_\delta,b_\delta'}$ respectively and the convergence is uniform on compact subsets of~$\Omega$. This readily implies
\[
	\delta^{-2}\Edelta^+[T(w_\delta)]\ \rightrightarrows\ \tfrac{3}{\pi}\CorrO{T(w)}
	\quad\text{and}\quad \delta^{-2}\Edelta^{b^{\vphantom{\prime}}_\delta,b'_\delta}[T(w_\delta)]\ \rightrightarrows\ \tfrac{3}{\pi}\Corr{T(w)}_\Omega^\dob\,.
\]
\end{theorem}

\begin{remark} \label{rem:boundary_thm1}
The assumption about the orthogonality of the half-edges~$b^{\vphantom{\prime}}_\delta$ and~$b'_\delta$ to the straight parts of the boundary~$\partial\Omega$ can be relaxed using the results from~\cite{CheSmi11}.
Moreover, the techniques developed in~\cite{HonKyt13} allow to relax these technical assumptions even further and prove
a similar statement in the case of rough boundaries.
\end{remark}

\begin{remark} \label{rem:strategy_thm1}
	To prove Theorem~\ref{thm:Ising1}, we express~$\mathcal{T}^{[\rho]}$ via the values of discrete fermionic observables defined by~\eqref{eq_F_real_def}, see Section~\ref{sub:T-via-fermionic} for more details. 
	More generally, for boundary conditions~$\b$ with~$|\b|=2m$ one can exploit the fermionic nature of these observables and show that~$\Tm^\b(m)$ admits an expression as a Pfaffian of some~$2(m\!+\!1)\times 2(m\!+\!1)$ matrix with entries given by the values of the fermionic observables; see Section~\ref{sub:TT-via-fermionic}, where the case of Dobrushin boundary conditions~$\b=\{b,b'\}$ is treated. 
	Therefore, one can use existing discrete complex analysis techniques to derive Theorem~\ref{thm:Ising1} from the convergence results for the fermionic observables. 
	In continuum, the fact that the stress-energy tensor~$T=-\frac{1}{2}:\!\psi\partial\psi\!:$ can be expressed via fermions is a particular feature of the corresponding CFT, and there is no surprise that the similar discrete quantities converge to their putative limits as~$\delta\to 0$. 
\end{remark}

The techniques discussed in Remark~\ref{rem:strategy_thm1} allow one to generalize these convergence results to all the multi-point \emph{correlation functions}~$\Edelta^{\b_\delta}[T(w_\delta)T(w'_\delta)T(w''_\delta)\dots]$. 
For shortness, we do not address such a general question in this paper and restrict ourselves to the convergence of two-point expectations with~`$+$' boundary conditions. 
Together with the proof of Theorem~\ref{thm:Ising1}, this already contains all necessary ingredients to treat the general situation.

\begin{theorem} \label{thm:Ising2}
	Let~$\Omega\subset\C$ be a bounded simply connected domain and $w,w'$ be two distinct points of~$\Omega$. 
	For a mesh-size~$\delta>0$, let~$w_\delta$ and~$w'_\delta$ be the faces of~$\Omega_\delta$ that contain~$w$ and~$w'$, respectively.
	Then, as~$\delta \to 0$, one has
	\[
		\delta^{-4}\Edelta^+[T(w_\delta)T(w'_\delta)]~\rightrightarrows~\tfrac{9\,}{\pi^2}\CorrO{T(w)T(w')}
	\]
	where~$\CorrO{T(w)T(w')}$ are holomorphic functions of~$w,w'\in\Omega$ defined in Section~\ref{sub:T-and-energy-via-fermions}.
	The convergence is uniform provided~$w$ and~$w'$ are separated from each other and from~$\partial\Omega$.
\end{theorem}

Another natural question to ask is the convergence of~\emph{mixed} correlations of~$T(w_\delta)$ and other local fields. The two most important examples are the spin~$\sigma(u_\delta)$ and the energy density
\[
\varepsilon(a_\delta)~:=~\sigma(a^+_\delta)\sigma(a^-_\delta)-\textstyle{\frac{2}{3}}\,,
\]
where~$a_\delta$ is an edge of~$\Omega_\delta$ and~$a^+_\delta$ and~$a^-_\delta$ denote the two faces adjacent to an edge~$a_\delta$ of~$\Omega_\delta$. 
Note that the constant~$\tfrac{2}{3}$ in the last definition is lattice dependent and corresponds to the infinite-volume limit of the Ising model (e.g., it should be replaced by~$\sqrt{2}/2$ when working on the square lattice, see~\cite{HonSmi13}). 

\begin{theorem}\label{thm:Ising3}
	Let~$\Omega\subset\C$ be a bounded simply connected domain,~$w\in\Omega$ and $u,a\in\Omega\setminus\{w\}$. 
	For a mesh-size~$\delta>0$, let~$w_\delta,u_\delta$ be the faces of~$\Omega_\delta$ that contain~$w$ and~$u$, respectively, while~$a_\delta$ denotes the closest to the point~$a$ edge of~$\Omega_\delta$. 
	Then, as~$\delta \to 0$, one has
	\[
		\delta^{-2}\frac{\Edelta^+[T(w_\delta)\varepsilon(a_\delta)]}{\Edelta^+[\varepsilon(a_\delta)]}~\rightrightarrows~\frac{3}{\pi}\cdot \frac{\CorrO{T(w)\varepsilon(a)}}{\CorrO{\varepsilon(a)}}\,,\quad
		\delta^{-2}\frac{\Edelta^+[T(w_\delta)\sigma(u_\delta)]}{\Edelta^+[\sigma(u_\delta)]}~\rightrightarrows~\frac{3}{\pi}\cdot \frac{\CorrO{T(w)\sigma(u)}}{\CorrO{\sigma(u)}},
	\]
where the functions~$\CorrO{T(w)\varepsilon(a)}$, $\CorrO{\varepsilon(a)}$, $\CorrO{T(w)\sigma(u)}$ and~$\CorrO{\sigma(u)}$ are defined in Sections~\ref{sub:T-and-energy-via-fermions} and~\ref{sub:T-spin-correlations}. The convergence is uniform provided~$a,u$ and~$w$ are separated from each other and from~$\partial\Omega$.
\end{theorem}

\begin{remark} (i) In the square lattice setup, the convergence of the energy density expectations~$\delta^{-1}\Edelta[\varepsilon(a_\delta)]$ to their CFT limits~$\mathcal{C}\cdot \CorrO{\varepsilon(a)}$ was proved in~\cite{HonSmi13}. Along the way, we also prove a version of this result with the lattice-dependent constant~$\mathcal{C}=\sqrt{3}/\pi$, see Corollary~\ref{cor:energy-density-convergence}.

\noindent (ii) In the square lattice setup, the convergence of the spin expectations~$\delta^{-\frac{1}{8}}\Edelta[\sigma(u_\delta)]$ to their CFT limits~$\mathcal{C}\cdot\CorrO{\sigma(u)}$ was proved in~\cite{CheHonIzy15}. We do not discuss the possible generalization of this result to the triangular lattice in our paper.
Nevertheless, let us mention that the proof given in~\cite{CheHonIzy15} can be adapted to our case in a rather straightforward manner.
\end{remark}

The methods developed in~\cite{CheSmi11,HonSmi13,Hon10,Izy11,CheHonIzy15} and this paper can be used to generalize Theorem~\ref{thm:Ising3} and to treat the convergence of \emph{all} mixed correlations of the local fields~$\sigma$,~$\varepsilon$ and~$T$, as well as its anti-holomorphic counterpart~$\overline{T}$, to their scaling limits. Note that one can \emph{define} all these limits without any reference to the CFT means, using the language of solutions to Riemann-type boundary value problems for holomorphic functions in~$\Omega$ instead. To illustrate this route, we devote Section~\ref{sec:CFT-correlations} to a self-contained exposition of a number of such boundary value problems, which allow one to define all the CFT correlation functions~$\CorrO{\dots}$ involved into Theorems~\mbox{\ref{thm:Ising1}--\ref{thm:Ising3}}. Moreover, in Section~\ref{sub:Schwarzian+OPE} we show that the Schwarzian conformal covariance of the stress-energy tensor expectations~$\Corr{T(w)}_\Omega^\b$ and the particular structure of singularities of the correlation functions~$\CorrO{T(w)T(w')}$,~$\CorrO{T(w)\varepsilon(a)}$ and~$\CorrO{T(w)\sigma(u)}$ (which comes from the OPEs in the standard CFT approach) can be easily deduced from simple properties of solutions to these boundary value problems.

It is worth noting that one can construct other local fields that (conjecturally) have the CFT stress-energy tensor as a scaling limit. For the critical Ising model on the square lattice, a simple linear combination of spin-spin expectations was suggested for this purpose by Kadanoff and Ceva~\cite{KadCev71}, see also the work of Koo and Saleur~\cite{KooSal94} for a generalization of this construction to other critical integrable lattice models.
A discrete stress-energy tensor for the three-state Potts model was analyzed by Mong, Clarke, Alicea, Lindner, and Fendley~\cite{MonClaAliLinFen14},
see also~\cite{AasMonFen16} where an algebraic approach to topological defects on the lattice was developed.

\paragraph{{\bf Organisation of the paper.}}
In Section~\ref{sec:combi}, we discuss the geometric construction of the discrete stress-energy tensor~$T$ and show that it satisfies a part of s-holomorphicity relations.
In Section~\ref{sec:via-fermions-disc}, we express~$T$ and its correlations with other discrete fields via discrete fermions.
In Section~\ref{sec:CFT-correlations}, we construct continuous correlation functions using continuous fermions and spinors.
We show that the correlations satisfy the same conformal covariance and OPEs as the corresponding correlations from the CFT.
In particular the stress-energy tensor has a Schwarzian conformal covariance.
In Section~\ref{sec:Ising-convergence}, we prove convergence results.
Section~\ref{sec:o-n} contains the proof of Proposition~\ref{prop:relation-Te}.
Appendix contains some background on s-holomorhpicity and a construction of the full-plane fermionic observable on the hexagonal lattice.

\medskip

\noindent{\bf Acknowledgements.}
We thank John Cardy, Cl\'ement Hongler and Konstantin Izyurov for many valuable discussions.
The contribution of the first author was made during his stay at the ITS-ETH Zürich in 2014/15. The first author gratefully acknowledges the support of Dr. Max Rössler, the Walter Haefner Foundation and the ETH Foundation.
This research was funded in part by the Austrian Science Fund (FWF) 10.55776/P34713.
The second and the third authors were partially supported by the Russian Science Foundation grant 14-21-00035.
All three authors were partially supported by SWISSMAP, FNS and ERC AG COMPASP.

\setcounter{equation}{0}

\section{Combinatorial section}
\label{sec:combi}

We start this section by extending the definition of the Ising model to non-planar graphs obtained by gluing rhombi and equilateral triangles. This definition includes two- and four-spin interactions and allows to deform the triangular lattice locally. By means of such deformations, we give a geometric construction of a discrete stress-energy tensor that satisfies a part of s-holomorphicity relations.
The latter are described at the end of the section, together with the fermionic observables.

\subsection{Ising model on non-planar rhombic-triangular grids}
\label{sub:ising-on-rhombic-tiling}

In order to obtain a geometrical construction of the stress-energy tensor, we need to extend the definition of the critical Ising model to the non-planar rhombic-triangular grids. Such grids appear naturally as a result of infinitesimal deformations of the triangular lattice.

Let $(G, \Phi)$ be a grid glued from rhombi and equilateral triangles and defined by the underlying graph~$G$ of their vertices and angles~$\Phi := (\varphi_1,\dots, \varphi_k)$ of the rhombuses. 
The Ising model on~$(G, \Phi)$ is supported on all possible ways to assign~$\pm 1$ to the vertices of~$G$ and is defined by
 \begin{equation}\label{eq:o-1-def}
	\bbP_{G,\Phi,\weight}(\sigma) = \frac{1}{\cZ_{G,\Phi}}\prod_{f\in F(G)}{\weight_{f}(\sigma)},
\end{equation}
where~$Z_{G,\Phi}$ is the normalising constant (partition function) that renders~$\bbP_{G,\Phi}$ a probability measure, the product is taken over all faces (triangles and rhombi) of~$G$, we write~$\weight(f)$ for the weight of a face~$f$.
For a triangle, this weight is~$1$ if all spins are the same and~$x$ otherwise;
for a rhombus, this weight is~$1$ if all spins are the same, 
$u(\varphi)$ if a spin at a vertex of angle~$\varphi$ differs from the spins at the three other vertices, 
$v(\varphi)$ if the spins disagree along both diagonals,
$w(\varphi)$ if the spins are plus on one diagonal and minus on the other (see Fig.~\ref{figWeights}).

\begin{figure}
    \begin{center}
      \tempskiped{\includegraphics[scale=0.85]{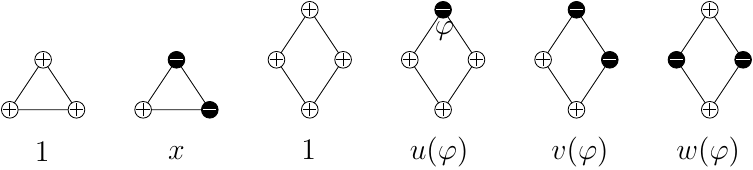}}
    \end{center}
    \caption{Possible local configurations inside the equilateral triangle and
    a rhombus with angles~$\varphi$ and~$\pi-\varphi$ and their weights.
    }
    \label{figWeights}
\end{figure}

The critical Ising model on the faces of the hexagonal lattice defined by~\eqref{eq:def-ising-hex} corresponds to~\eqref{eq:o-1-def} when~$G$ consists of equilateral triangles and~$x=1/\sqrt3$.
When rhombuses are present, the Hamiltonian contains four-point interactions.
Still, the definition seems natural since the model is invariant under the Yang--Baxter transformation for the following weights found by Nienhuis~\cite{Nie90}\footnote{The works~\cite{Nie90,IkhCar09} consider the loop~$O(n)$ model, which at~$n=1$ corresponds to the domain wall representation of the Ising model.
In the notation of~\eqref{eq_u1}-\eqref{eq_w2}, we have~$u(\varphi) = u_1(\varphi)$ and~$w(\varphi) = w_1(\varphi) + w_2(\varphi)$ (since a rhombus with alternating spins has two loop realisations).}:
\begin{equation} \label{eq:weights}
	u(\varphi) = \tfrac{\sqrt{3} \cos \tfrac{\varphi}2}{\sqrt{3}+\sin\varphi}; \quad
	v(\varphi) = \tfrac{\sin \varphi}{\sqrt{3}+\sin\varphi}; \quad
	w(\varphi) = \tfrac{\sqrt{3} - \sin \varphi}{\sqrt{3}+\sin\varphi}.
\end{equation}
We refer to these weights as {\em critical} though we will not check this explicitly.
In Proposition~\ref{prop:inv-to-transform} below, we show that the model is invariant also to a tiling-change of a pentagon.

We first note that the weights are well-defined, in a sense that they have the following symmetries:
\[
	v(\varphi) = v(\pi-\varphi) \quad \text{and} \quad w(\varphi) = w(\pi-\varphi).
\]
In a rhombus has angles~$\frac{\pi}{3}$ and~$\frac{2\pi}{3}$, the weights are~$u(\frac{\pi}{3}) = \tfrac{1}{\sqrt{3}}$ and $u(\frac{2\pi}{3}) = v(\frac{\pi}{3}) = w(\frac{\pi}{3}) = \tfrac13$.
This corresponds to the critical Ising model on the triangular lattice~\cite{Wan50}.

\begin{definition}\label{def:invariance-to-transform}
	Let~$P$ be a bounded polygon in the plane, whose interior can be tiled by rhombuses and equilateral triangles in (at least) two different ways that we denote by~$(G_{P},\Phi_P)$ and~$(\tilde{G}_{P},\tilde\Phi_P)$. 
	We say that~$\bbP_{G,\Phi,\weight}$ is {\em invariant to the transformation} turning tiling~$(G_{P},\Phi_P)$ into~$(\tilde{G}_{P},\tilde\Phi_P)$ if, for any configuration~$\tau$ at the vertices of~$P$,
	\[
		\bbP_{G,\Phi,\weight}(\cdot \, | \, \sigma_{| P} = \tau) = \bbP_{\tilde{G},\tilde{\Phi},\weight}(\cdot \, | \, \sigma_{| P} = \tau).
	\]
\end{definition}

\begin{figure}
	\includegraphics[width=0.7\textwidth]{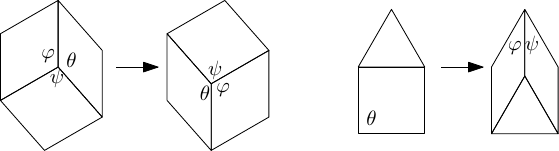}
	\caption{{\it Left:} Yang--Baxter (star-triangle) transformation; $\varphi+\psi+\theta=2\pi$.
	 {\it Right:} Star-segment transformation; $\theta=\frac{\pi}{3}+\varphi$ and $\varphi +\psi = \frac{\pi}{3}$.}
	 \label{fig:yb-pent-transform}
\end{figure}

\begin{proposition}\label{prop:inv-to-transform}
	The Ising model defined by \eqref{eq:o-1-def} with Nienhuis' critical weights \eqref{eq:weights} is invariant to the following transormations (see fig.~\ref{fig:yb-pent-transform}):
	\begin{itemize}
		\item[(T1)] Splitting a rhombus of angle~$\tfrac{\pi}{3}$ into two equilateral triangles.
		\item[(T2)] {\em Yang--Baxter transformation.} Replacing three rhombi tiling a flat hexagon with their symmetric images with respect to the center of the hexagon.
		\item[(T3)] {\em Star--segment transformation.} Replacing two rhombi and an equilateral triangle tiling a flat pentagon with a rhombus and an equilateral triangle tiling the same pentagon.
	\end{itemize}
\end{proposition}
The proof for~(T1) is straightforward; (T2) was established by Nienhuis \cite{Nie90} by checking all possible boundary conditions; (T3) can be proven in the same way, see Appendix~\ref{sec:star-segment}.

\begin{remark}\label{rem:pent-transform}
	1) We are not aware of any prior works mentioning transformation (T3) and we could not establish any connection to the seminal pentagon identity~\cite{MooSei89}. Yang--Baxter equation has proven to be very powerful in solving the models of statistical mechanics, including the recent proof of rotational invariance of the random-cluster model~\cite{DumKozKraManOul20}, and it is known to be related to discrete holomorphic observables~\cite{IkhCar09}. 
	Finding possible applications of the star--segment transformation would be of great interest.
	
	2) Proposition \ref{prop:inv-to-transform} should imply that this definition of the Ising model on a surface tiled with rhombi and triangles does not depend on the tiling but we do not address this general question here.
	
	3) A tiling of triangles and rhombuses is not an isoradial graph and thus does not fit in the framework of~\cite{CheSmi11}.
\end{remark}

This Ising model can also be defined using a Hamiltonian. Indeed, take
	\[
		- H_{G,\Phi}(\sigma) =
		\sum_{\substack{i\sim j \\ \text{angles } \varphi, \psi }} \tfrac12 (J_{1}(\varphi) + J_{1}(\psi))\sigma_{i}\sigma_{j} +
		\sum_{\substack{i \text{ diag. }  j \\  \text{angle } \varphi }} J_{2}(\varphi) \sigma_{i}\sigma_{j} +
		\sum_{\substack{i\sim j\sim k\sim \ell \\  \text{angle } \varphi }} U(\varphi) \sigma_{i}\sigma_{j}\sigma_{k}\sigma_{\ell},
	\]
	where the first sum runs over all pairs of adjacent vertices and~$\varphi,\psi$ are angles in the two rhombuses containing them (if one of them is a triangle, take the angle~$\tfrac{\pi}{3}$); 
the second sum runs over all pairs of non-adjacent spins belonging to the same rhombus and~$\varphi$ denotes its angle at~$i$ (or~$j$);
the third sum runs over all quadruples of vertices belonging to the same rhombus and~$\varphi$ is its angle.

\begin{lemma}
	Let $(G, \Phi)$ be a grid glued from rhombuses and equilateral triangles.
	Take $J_1,J_2,U>0$ that satisfy
	\begin{equation}
		\label{eq:hamiltonian-coefficients}
		e^{8J_1(\varphi)} = \tfrac{1}{w^{2} (\varphi)}, \quad
		e^{8J_2(\varphi)} = \tfrac{u (\pi-\varphi)w (\varphi)}{v^2 (\varphi)u^2 (\varphi)},   \quad
		e^{8U(\varphi)}=\tfrac{v^2 (\varphi)w (\varphi)}{u^2 (\varphi)u^2(\pi-\varphi)}.
	\end{equation}
	Then, the measure~$\bbP_{G,\Phi}(\sigma)$ defined by~\eqref{eq:o-1-def} can be written as follows:
	\begin{equation}\label{eq:def-via-hamiltonian}
		\bbP_{G,\Phi}(\sigma) = \frac{1}{\cZ_{G,\Phi}}\cdot e^{-H_{G,\Phi}(\sigma)}.
	\end{equation}
\end{lemma}

\begin{proof}
	The Hamiltonian~$H_{G,\Phi}(\sigma)$ includes two- and four-spin interactions, and the weight of a rhombus~1234 of angle~$\varphi$ at the first vertex can be written in the following form:
	\begin{align}
		\label{eq:hamiltonian-one-rhombus}
		\weight_{1234} (\sigma) 
		= \exp &\left [\tfrac12 J_1(\varphi)(\sigma_1\sigma_2 + \sigma_2\sigma_3 + \sigma_3\sigma_4 + \sigma_4\sigma_1)\right. \\
		&\left.+ J_2(\varphi)\sigma_1\sigma_3 + J_2(\pi -\varphi)\sigma_2\sigma_4 
		+ U(\varphi)\sigma_1\sigma_2\sigma_3\sigma_4 \right], \nonumber
	\end{align}
	while the weight of a triangular face~123 equals
	\[
		\weight_{123} (\sigma) = \exp \left [\tfrac12 J_1(\tfrac{\pi}{3})(\sigma_1\sigma_2 + \sigma_2\sigma_3 + \sigma_3\sigma_1)\right].
	\]
	Equating these weights to those appearing in definition~\eqref{eq:o-1-def}, we arrive at the following:
	\begin{align*}
		\exp[{2J_1(\varphi) + J_2(\varphi) + J_2(\pi -\varphi) - U(\varphi)}] &= 1,\\
		\exp[{-J_2(\varphi) + J_2(\pi -\varphi) - U(\varphi)}] &= u(\varphi),\\
		\exp[{-J_2(\varphi) - J_2(\pi -\varphi) + U(\varphi)}] &= v(\varphi),\\
		\exp[{-2J_1(\varphi) + J_2(\varphi) + J_2(\pi -\varphi) + U(\varphi)}] &= w(\varphi).
	\end{align*}
	Solving this system, we arrive at~\eqref{eq:hamiltonian-coefficients}.
\end{proof}

Note the first and the last sums in the Hamiltonian~\eqref{eq:def-via-hamiltonian} are well-defined due to the symmetries of~$J_1$ and~$U$ given by~\eqref{eq:hamiltonian-coefficients}:
\[
	J_{1}(\varphi) = J_{1}(\pi- \varphi), \quad U(\varphi) = U(\pi- \varphi).
\]

Let~$\b = (b_1, \dots, b_{2k})$ be a subset of the midpoints of edges on the boundary of a grid~$(G,\Phi)$. The Ising model on~$(G,\Phi)$ with boundary conditions~$\b$ is supported on spin configurations that are fixed to be~$+1$ on all arcs~$b_{2i-1}b_{2i}$ and~$-1$ on all arcs~$b_{2i}b_{2i+1}$, and is defined by the same Hamiltonian \eqref{eq:def-via-hamiltonian}. Denote by~$\cZ^\b(G,\Phi)$ the corresponding partition function.

\subsection{Stress-energy tensor as an infinitesimal non-planar deformation}
\label{sub:stress-tensor-as-deformation}

In this subsection we describe two sets of infinitesimal deformations of the hexagonal lattice: for each edge and for each midline of a face (i.e. an interval between the midpoints of the opposite edges). Real-valued observables~$\Te$ and~$\Tm$ are defined as the logarithmic derivatives of the partition function under these deformations. 
These observables can be regarded as projections of a discrete stress-energy tensor.

Recall that~$\Hex_\delta$ is a hexagonal lattice of mesh-size~$\delta>0$ and that~$\Omega_\delta$ is a $\Hex_\delta$-approximation of a given simply connected planar domain~$\Omega$.
Define~$\T_\delta$ to be a triangular lattice dual to~$\Hex_\delta$.
Define~$\OmegaDual_\delta$ as the graph obtained via the following glueing procedure:
\begin{itemize}
	\item consider all faces of~$\T_\delta$ that correspond to the vertices of~$\Omega_\delta$;
	\item view all vertices of these faces as distinct;
	\item now identify the vertices that correspond to the same face of~$\Omega_\delta$.
\end{itemize}
Note that the vertices of~$\OmegaDual_\delta$ corresponding to a face outside of~$\Omega_\delta$ are not identified and thus appear with multiplicity.
Below in this subsection we will often omit~$\delta$ and write just~$\Omega$ and~$\OmegaDual$ instead of~$\Omega_\delta$ and~$\OmegaDual_\delta$, if no confusion arises.

Consider an edge~$e$ of the graph~$\Omega$. Replace in~$\OmegaDual$ two triangles corresponding to the endpoints of~$e$
by a rhombus of angle~$\varphi$ and denote the resulting tiling by~$\OmegaE(e;\varphi)$. 
The distribution of the Ising model on~$\Omega$ and~$\OmegaE(e;\tfrac\pi{3})$ are the same by the invariance under transformation (T1);
moreover, the partition function is also the same.
Varying the angle~$\varphi$ modifies the Hamiltonian at one rhombus and corresponds to the insertion of a local field. 

Define~$\Te^{\b} (e)$ as the derivative of this field with respect to~$\varphi$ at~$\pi/3$ corrected by a constant~$\Ce = -\tfrac{1}{2\pi} + \tfrac{1}{4\sqrt{3}}$. More precisely, in view of~\eqref{eq:hamiltonian-one-rhombus}, we define
\begin{align}\label{eq:t-edge}
	\Te^{\b}(e) := \Ce &+ \bbE_{\Omega}^{\b}\left[\tfrac12 J_1'(\tfrac\pi{3}) (\sigma_1\sigma_2 + \sigma_2\sigma_3 + \sigma_3\sigma_4 + \sigma_4\sigma_1)\right.\\
		&\left.+ J_2'(\tfrac\pi{3}) \sigma_1\sigma_3
		- J_2'(\tfrac{2\pi}{3}) \sigma_2\sigma_4
		+ U'(\tfrac\pi{3}) \sigma_1\sigma_2\sigma_3\sigma_4\right], \nonumber
\end{align}
where the coefficients can be easily computed substituting~\eqref{eq:weights} into~\eqref{eq:hamiltonian-coefficients},
\[
	J_1'(\tfrac\pi{3}) = \tfrac{1}{3\sqrt{3}}, \quad 
	J_2'(\tfrac\pi{3}) = \tfrac{1}{6\sqrt{3}}, \quad 
	J_2'(\tfrac{2\pi}{3}) = \tfrac{-5}{6\sqrt{3}}, \quad 
	U'(\tfrac\pi{3}) = \tfrac{-1}{12\sqrt{3}}.
\]
The constant~$\Ce$, as well as constants~$\CmT$ and~$\CmR$, eventually comes from the local values~\eqref{eq_local_values_zero}, \eqref{eq_local_values} of some full-plane observable and are not important for the combinatorial considerations. The observable $\Te^{\b}$ is real-valued and we refer to it as a projection of the discrete stress-energy tensor introduced in Definition~\ref{def:T-via-spins} (see Proposition~\ref{prop:T-via-edges}).

Although one could add any null-field without affecting the OPEs in the limit stated in the introduction, the definition given above has two benefits: clear geometrical interpretation and local relations stated in Proposition~\ref{prop:relation-Te}.
These relations can be viewed as a part of s-holomorphicity (or propagation) equations, see Section~\ref{sub:fermion}.
In Proposition~\ref{prop:relation-TeTm} below we relate~$\Te$ to the operators introduced in Definition~\ref{def:T-via-spins}.
It is straightforward to see that Proposition~\ref{prop:relation-Te} (for~$n=1,\sigma=1/2$) is its direct consequence.

\begin{remark}
	The most direct way to prove Proposition~\ref{prop:relation-Te} is by pairing up configurations differing only in the spin at the face $abcdef$~--- we do this in Section~\ref{sec:o-n} for a general loop~$O(n)$ model.
	Another possibility is to take derivative in the Yang--Baxter equation (i.e. (T2) in Proposition~\ref{prop:inv-to-transform}). 
\end{remark}

\begin{definition}\label{def:T-mid}
	For a non-boundary vertex~$w$ of~$\Omega$ and any~$\rho\in\wp$, let~$\wrho$ be the midline of~$w$ oriented in the direction~$i\rho$.
	We define
	\[
		\Tm^{\b}(\wrho) := \bbE_{\Omega}^{\b}\ \left[\calT^{[\rho]}(w)\right]\,.
	\]
\end{definition}

Similarly to~$\Te$, the observable~$\Tm$ also has an interpretation in terms of infinitesimal deformations~--- one can insert a rhombus of angle~$\varepsilon$ along the midline~$\wrho$ and take a derivative in~$\varepsilon$ at~$0$. The expressions in Definition~\ref{def:T-via-spins} can be obtained from the Hamiltonian after adding a vertex to it. This geometric interpretation, together with the invariance under star--segment transformation (T3) described above, leads to the following relation.

\begin{proposition}
	\label{prop:relation-TeTm}
	Let~$a,b,c,d,e,f$ be the edges of a a face in~$\Omega$ (listed in a cyclic order). Let~$m$ be its midline linking midpoints of~$c$ and~$f$ (see fig.~\ref{fig:t-edge-mid-rel}). Then,
	\begin{align}
		\label{eq_relation_Te_Tm}
		\Tm^{\b}(m) = \Te^{\b}(a) + \Te^{\b}(b)\,.
	\end{align}
\end{proposition}

\begin{figure}
    \begin{center}
      \tempskiped{\includegraphics[width=0.9\textwidth]{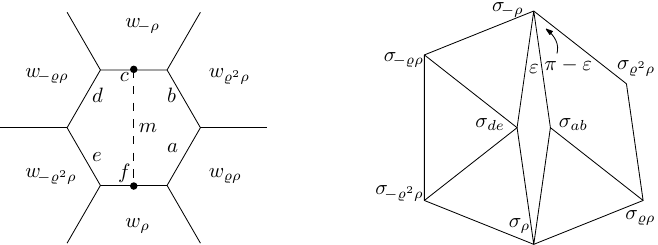}}
    \end{center}
    \caption{
    {\em Left:} Hexagon~$w:=abcdef$ with a midline~$m = \wrho$ (for~$\rho=1$) and the six faces around~$w$.
    {\em Right:} Deformation of the lattice at the midline~$m$.
    }
    \label{fig:t-edge-mid-rel}
\end{figure}

It is immediate that~\eqref{eq_relation_Te_Tm} implies Proposition~\ref{prop:relation-Te}.
We now focus on proving~\eqref{eq_relation_Te_Tm}.

\begin{proof}
	Instead of pairing up configurations differing at $abcdef$, we choose to derive \eqref{eq_relation_Te_Tm} from the invariance to the star-segment transformation stated in Proposition \ref{prop:inv-to-transform}. In addition to circumventing a lengthy system of linear equations, this provides a direct link between local relations on observables and the Yang--Baxter-type equation and puts forward the idea of taking a derivative with respect to a local deformation of the lattice. We start by defining this deformation explicitly.
	
	Recall that $\OmegaDual$ denotes the graph glued from the faces of the triangular lattice $\T$ corresponding to the vertices of $\Omega$. Denote the face $abcdef$ and the faces around it by $w$, $\weOne$, $\weTwo$, $\weThree$, $\weFour$, $\weFive$, $\weSix$, as shown on Fig~\ref{fig:t-edge-mid-rel}. We now define a tiling $(G, \Phi[\varepsilon])$ obtained from $\OmegaDual$ by gluing triangles $w \weTwo \weThree$ and $w \weThree \weFour$ into a rhombus of angle $\tfrac{\pi}{3} -\varepsilon$ and inserting a rhombus of angle $\varepsilon$ at the place of two edges $w \weOne$ and $w \weFour$ (see Fig.~\ref{fig:t-edge-mid-rel}):
	\begin{itemize}
		\item replace vertex $w$ in $\OmegaDual$ by two vertices $\wrho$ and $\wmrho$;
		\item $\wrho$ is connected by edges to $\weOne$, $\weTwo$, $\weFour$ and vertex $\wmrho$ is connected to $\weFour$, $\weFive$, $\weSix$,  $\weOne$;
		\item the triangles $\wrho \weOne \weTwo$, $\wmrho \weFour \weFive$,   $\wmrho \weFive \weSix$, $\wmrho \weSix \weOne$ are equilateral;
		\item the quadrilaterals $ r_{m}:=\weFour \wrho  \weOne \wmrho$ and $ r_{b}:=\weFour \wrho \weTwo \weThree$ are rhombuses whose angles at the vertex $\weFour$ are $\varepsilon$ and $\tfrac{\pi}{3} - \varepsilon$, respectively.
	\end{itemize}
	Use the notation~$\sigma_0^{[\rho]}$ and~$\sigma_0^{[-\rho]}$ for the spins at~$\wrho$ and~$\wmrho$ respectively.
	At~$\varepsilon = 0$, by~\eqref{eq:weights}, $\sigma_0^{[\rho]}=\sigma_0^{[-\rho]}$.
	In addition, Proposition \ref{prop:inv-to-transform} asserts equivalence between the two equilateral triangles $w \weTwo \weThree$, $w \weThree \weFour$ and a rhombus  $w \weTwo \weThree \weFour$ of angle $\tfrac{\pi}{3}$. 
	Altogether, this gives the following equality of partition functions (recall~\eqref{eq:o-1-def}):
	\begin{equation} \label{eq:0-angle-case}
		\cZ_{G,\Phi[0]} = \cZ_{\OmegaDual}.
	\end{equation}
	Define $(\tilde{G},\tilde{\Phi}[\varepsilon])$ as the tiling obtained from $(G,\Phi[\varepsilon])$ by removing the vertex $\rho$ and adding the edge $\wmrho\weThree$, so that $w \weThree \weFour$ is an equilateral triangle and the rhombus $r_{a}:=w\weOne\weTwo\weThree$ has angle $\tfrac{\pi}{3} + \varepsilon$. By Proposition \ref{prop:inv-to-transform}, the spins  at $w$, $\weOne$, $\weTwo$, $\weThree$, $\weFive$ have the same distribution under $\bbP_{\tilde{G},\tilde{\Phi}[\varepsilon]}$ and $\bbP_{G,\Phi[\varepsilon]}$. Setting all spins to be plus, we obtain the coefficient of proportionality of the partition functions:
	\begin{equation} \label{eq:part-func-ratio}
		\cZ_{G,\Phi[\varepsilon]}  = (1+ u(\pi-\varepsilon)u(\tfrac{2\pi}{3} + \varepsilon)x_{c}) \cZ_{\tilde{G},\tilde{\Phi}[\varepsilon]},
	\end{equation}
	where $ u(\pi-\varepsilon)u(\tfrac{2\pi}{3} + \varepsilon)x_{c}$ is the weight of the configuration when the spin at $\wrho$ is minus, while all other spins are plus. 
	We divide both sides of~\eqref{eq:part-func-ratio} by~$\cZ_{\OmegaDual}=\cZ_{G,\Phi[0]}=\cZ_{\tilde{G},\tilde{\Phi}[0]}$, and take a derivative at $\varepsilon = 0$ using that $u(\pi) = 0$:
	\begin{align} \label{eq:part-func-derivative}
		\frac{1}{\cZ_{G,\Phi[0]}} 
		&\cdot  \sum_{\sigma\in \{\pm 1\}^{V(G)}}\weight'_{r_{m}}(\sigma) \prod_{\substack{f\in F(G)\\ f\neq r_{m}}} \weight_{f}(\sigma)  -   
		\frac{1}{\cZ_{\OmegaDual}} \cdot 
		\sum_{\sigma\in \{\pm 1\}^{V(G)}}\weight'_{r_{b}}(\sigma)\prod_{\substack{f\in F(G)\\ f\neq r_{b}}} \weight_{f}(\sigma) \nonumber \\
		&= \frac{1}{\cZ_{\OmegaDual}} \cdot \sum_{\sigma\in \{\pm 1\}^{V(G)}}\weight'_{r_{a}}(\sigma) \prod_{\substack{f\in F(G)\\ f\neq r_{a}}} \weight_{f}(\sigma) + u'(\pi)u(\tfrac{2\pi}{3})x_{c}.
	\end{align}
	The sums on the LHS are denoted by $S_{m}$ and $S_{b}$ respectively, and the sum on the RHS is denoted by $S_{a}$. 
	In~$S_b$, we have~$\sigma_0^{[\rho]}=\sigma_0^{[-\rho]}$, since $u(\pi) = v(0) = 0$.
	Identifying both~$\sigma_0^{[\rho]}$ and~$\sigma_0^{[-\rho]}$ with~$\sigma_0$, we get a bijection between spin configurations on $G$ and spin configurations on the faces of $\Omega$, whence
	\begin{equation}\label{eq:rel-proof-Tb}
		S_{b} = \sum_{\sigma\in \{\pm 1\}^{F(\Omega)}}\tfrac{\weight'_{r_{b}}(\sigma)}{\weight_{r_{b}}(\sigma)} \cdot \bbP_{\Omega}(\sigma) 
		= \bbE_{\Omega}\left( \tfrac{\weight'_{r_{b}}(\sigma)}{\weight_{r_{b}}(\sigma)} \right) = \Te^{\b}(b) - \Ce,
	\end{equation}
	where we used the definition of $\Te$ given by \eqref{eq:t-edge}. Similarly,
	\begin{equation}\label{eq:rel-proof-Ta}
		S_{a} = \Te^{\b}(a) - \Ce.
	\end{equation}
	We will now relate $S_{m}$ to $\Tm$.
	We split all spin configurations on~$V(G)$ into four sets: 
	\begin{itemize}
		\item $\ConfOs{\b,\Trho}:=\left\{ \sigma\in\{\pm 1\}^{V(G)} \colon \sigma^{[-\rho]} \neq \sigma^{[\rho]} = \seOne = \seFour\right\}$;
		\item $\ConfOs{\b,\Tmrho}:= \left\{\sigma\in\{\pm 1\}^{V(G)} \colon \sigma^{[\rho]} \neq \sigma^{[-\rho]} = \seOne = \seFour\right\}$;
		\item $\ConfOs{\b,\Rrho}:= \left\{\sigma\in\{\pm 1\}^{V(G)} \colon \seOne \neq \seFour \,\, \text{or} \,\, \sigma^{[-\rho]} = \sigma^{[\rho]} \neq \seOne = \seFour\right\}$;
		\item $\ConfOs{\b, 0}:= \left\{\sigma\in\{\pm 1\}^{V(G)} \colon \sigma^{[\rho]} = \sigma^{[-\rho]} =\seOne = \seFour\right\}$.
	\end{itemize}
	In the last case, $\sigma$ contributes~$0$ to~$S_m$, because~$\weight_{r_m}'=0$.
	We now show that
	\begin{align}
		\bbE_{\Omega}[\Trho(w)] &= \CmT + \sum_{\sigma\in \ConfOs{\b,\Trho}}\weight'_{r_{m}}(\sigma) \prod_{\substack{f\in F(G), f\neq r_{m}}} \weight_{f}(\sigma), \label{eq:trho-via-double-spin}\\
		\bbE[\Tmrho(w)] &= \CmT + \sum_{\sigma\in \ConfOs{\b,\Tmrho}}\weight'_{r_{m}}(\sigma) \prod_{\substack{f\in F(G), f\neq r_{m}}} \weight_{f}(\sigma),\label{eq:tmrho-via-double-spin}\\
		\bbE[\Rrho(w)] &= \CmR + \sum_{\sigma\in \ConfOs{\b,\Rrho}}\weight'_{r_{m}}(\sigma) \prod_{\substack{f\in F(G), f\neq r_{m}}} \weight_{f}(\sigma).\label{eq:rrho-via-double-spin}
	\end{align}
	Indeed, for every~$\sigma\in \ConfOs{\b,\Trho}$, we view it as a spin configuration on the faces of~$\Omega$ given by~$\sigma_0:= \sigma_0^{[\rho]}$ at~$w$ and by~$\sigma(u)$, for~$u\neq w$.
	By considering~$\seFive,\seSix\in \{\pm1\}$, we get
	\begin{equation}\label{eq:double-spin-to-spin}
		\prod_{\substack{f\in F(G), f\neq r_{m}}} \weight_{f}(\sigma) = \weight_\Omega(\sigma) \cdot \tfrac{1}{3\sqrt{3}} \cdot (2-\sigma_0\seFive)(2-\sigma_0\seSix).
	\end{equation}
	Since~$\sigma^{[-\rho]} \neq \sigma^{[\rho]} = \seOne = \seFour$, we have~$\sigma_0 = \seOne = \seFour$,
	Rewrite the indicator function:
	\begin{equation}\label{eq:indicator-to-spin}
		1_{\sigma_0 = \seOne = \seFour} = \tfrac14(1+\sigma_0\seOne)(1+\sigma_0\seOne).
	\end{equation}
	We also have~$\weight_{r_m}' = u'(\pi) = \frac12$.
	Multiplying~\eqref{eq:double-spin-to-spin} and~\eqref{eq:indicator-to-spin} and summing over all~$\sigma\in \ConfOs{\b}$, we arrive at~\eqref{eq:trho-via-double-spin}. We prove~\eqref{eq:tmrho-via-double-spin} in the same way. 

	For~\eqref{eq:rrho-via-double-spin}, we only need to adapt the definition of~$\sigma_0$:
	\begin{itemize}
		\item If~$\seOne \neq \seFour$, then $\sigma^{[\rho]} = \sigma^{[-\rho]}$ and we take~$\sigma_0 := \sigma^{[-\rho]} = \sigma^{[\rho]}$;
		\item If~$\seOne = \seFour$, then~$\sigma^{[\rho]} \neq \sigma^{[-\rho]}$ and we take~$\sigma_0 =\sigma^{[\rho]}$ or~$\sigma_0 =\sigma^{[-\rho]}$ with probability~$1/2$.
	\end{itemize}
	Since~$w'(0) = -\frac{2}{\sqrt{3}} = 2u'(0)$, the first case contributes
	\[
		-\tfrac1{\sqrt{3}} \sum_{\sigma\in\ConfOs{\b}} [(1-\sigma_0\seOne) + (1-\sigma_0\seFour)].
	\]
	In the second case, $\weight_{r_m}' = v'(0) = \frac1{\sqrt{3}}$ and we need to include a correction for the weight of~$\sigma$ (as in the proof for~$\Trho$ above). Overall, this case contributes
	\[
		\tfrac1{12\sqrt{3}}\sum_{\sigma\in\ConfOs{\b}} (1-\seOne\seFour)\left[(2-\sigma_0\seTwo)(2-\sigma_0\seThree) + (2-\sigma_0\seFive)(2-\sigma_0\seSix) \right].
	\]	
	Substituting~\eqref{eq:trho-via-double-spin}-\eqref{eq:rrho-via-double-spin}, we get that~$S_{m} = \Tm^{\b}(m) - 2\CmT - \CmR$.
	Combining this with~\eqref{eq:part-func-derivative}, \eqref{eq:rel-proof-Tb} and \eqref{eq:rel-proof-Ta}, we obtain
	\[
		\Tm^{\b}(m) - 2\CmT - \CmR = \Te^{\b}(a) +  \Te^{\b}(a) - 2\Ce + \tfrac{1}{6\sqrt{3}},
	\]
	where we used that $u'(\pi)u(\tfrac{2\pi}{3})x_{c} = \tfrac{1}{6\sqrt{3}}$. Finally, substituting the constants from~\eqref{eq:t-edge} and~\eqref{eq:CmTCmRCmid=}, we see that they cancel out and the required relation follows.
\end{proof}

Putting together Definitions~\ref{def:T-via-spins} and~\ref{def:T-mid} and Proposition~\ref{prop:relation-TeTm}, we obtain that~$\Te$ can be viewed as a projection of~$T$:

\begin{proposition}\label{prop:T-via-edges}
	Let~$w$ be face of~$\Omega$ bounded by edges~$a,b,c,d,e,f$ in the counter-clockwise order, starting by a top horizontal edge~$a$.
	Then,
	\[
		\bbE_{\Omega}^{\b}[T(w)] = \tfrac23 \cdot \left[\Te^{\b}(a) + e^{2\pi i/3}  \cdot \Te^{\b}(c) + e^{4\pi i/3}\cdot  \Te^{\b}(e)\right]. 
	\]
\end{proposition}

\subsection{Graphical representation of the stress-energy tensor}
\label{sub:graphical}

Let~$\Omega$ be discrete domain on~$\Hex$ and~$\b$ be an even subset of its boundary half-edges~$\partial\Edges(\Omega)$.
We recall that we identify the boundary edges with their half-edges incident to a non-boundary vertex of~$\Omega$.
For a spin configuration~$\sigma\in \ConfOs{\b}$, its domain wall representation is a collection of edges in~$\Edges(\Omega)$ separating adjacent faces with opposite spins.
In this subgraph, each vertex in~$\Vertices(\Omega)$ has degree~$0$ or~$2$, all half-edges of~$\b$ belong to this subgraph, no other half-edge of~$\partial\Edges(\Omega)$ are present.
This describes a bijection between~$\ConfOs{\b}$ and subgraphs consists of cycles (called {\em loops}) and paths ending at half-edges in~$\b$.
The pushforward of the critical Ising measure under this bijection assigns to every such subgraph~$\gamma$ a probability proportional to~$\weight(\gamma)$ given by
\[
	\weight(\gamma)= (1/\sqrt{3})^{|\gamma|},
\]
where~$|\gamma|$ denotes the number of edges in~$\gamma$ (half-edges being counted as~$\tfrac{1}{2}$).
This measure is also known as the loop~$O(1)$ model.

It is natural to consider also paths ending at half-edges inside~$\Omega$.
In general, subgraphs with such paths do not correspond to any spin configuration but this extension will allow us to express~$\Tm$ and the fermionic observable that we define below.

\begin{definition}\label{def:loop-conf-on-tiling}
	For two (inner) half-edges~$a,e$ of~$\Omega$ and an even subset~$\b$ of~$\partial\Edges(\Omega)$, we denote by~$\ConfO{\b\cup\{a,e\}}$ the set
	of subgraphs of~$\Omega$ consisting of several non-intersecting cycles and paths whose ends are the half-edges
	contained in~$\b\cup\{a,e\}$, so that each half-edge of~$\b\cup\{a,e\}$ is the end of exactly one of the paths.
	The corresponding partition function is given by
	\[
		\Weight^\b_\Omega(a,e)~:=~ \sum\nolimits_{\gamma\in\ConfO{\b\cup\{a,e\}}}\weight(\gamma)\,.
	\]
\end{definition}

For an \emph{oriented} midline~$m$ of an inner face of~$\Omega$, let~$m_{\mathrm{down}}$
and~$m_{\mathrm{up}}$ denote the two half-edges oriented in the direction~$-im$,
the former starting at the beginning of~$m$ and the latter at its endpoint, see Fig.~\ref{fig:winding}.
Further, we denote by~$\overline{m}$ the same midline oriented in the opposite direction.
Slightly abusing the notation, we continue to use the symbol~$\Tm(m)=\Tm(\overline{m})$ for the value of the observable~$\Tm$ discussed above (and originally defined for unoriented midlines).

\begin{lemma}\label{lemma:Tm-Ising-as-sums}
	Let~$w$ be a face in~$\Omega$ and~$\rho\in\wp$. Recall~$\CmT$ and~$\CmR$ defined by~\eqref{eq:CmTCmRCmid=}. Then,
	\begin{align}
		\bbE_{\Omega}^\b[\Trho(w)] 
		&= \CmT + \frac{1}{2\cZ^\b_\Omega} \cdot \Weight^\b_\Omega(m_{\mathrm{down}},m_{\mathrm{up}}),
		\label{eq:Trho-Ising-as-sums} \\
		\bbE_{\Omega}^\b[\Tmrho(w)] 
		&= \CmT + \frac{1}{2\cZ^\b_\Omega} \cdot \Weight^\b_\Omega(\overline{m}_{\mathrm{down}},\overline{m}_{\mathrm{up}}),
		\label{eq:Tmrho-Ising-as-sums} \\
		\bbE_{\Omega}^\b[\Rrho(w)]
		&= \CmR + \frac{1}{\sqrt{3}\cZ^\b_\Omega}\left[\Weight^\b_\Omega(m_{\mathrm{down}},\overline{m}_{\mathrm{down}}) +\Weight^\b_\Omega(m_{\mathrm{up}},\overline{m}_{\mathrm{up}})\right] \notag\\
		& -~\frac{1}{\sqrt{3}\cZ^\b_\Omega}\left[\Weight^\b_\Omega(m_{\mathrm{down}},\overline{m}_{\mathrm{up}}) + \Weight^\b_\Omega(m_{\mathrm{up}},\overline{m}_\mathrm{down})\right].
		\label{eq:Rrho-Ising-as-sums}
	\end{align}
\end{lemma}

\begin{figure}[ht]
    \begin{center}
      \tempskiped{\includegraphics[scale=1.2]{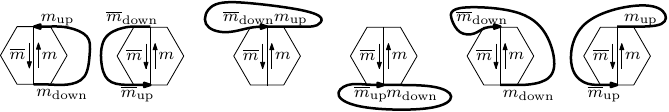}}
    \end{center}
    \caption{A schematic drawing of curves connecting two half-edges adjacent to a midline~$m$, which correspond to the six terms~$\Weight^\b_\Omega(\cdot,\cdot)$ in~\eqref{eq:Trho-Ising-as-sums}-\eqref{eq:Rrho-Ising-as-sums}. Due to topological reasons, the winding (total rotation angle) of an oriented curve equals~$\pi$ or~$-3\pi$ in the first two cases and~$\pm2\pi$ in the last four.}
    \label{fig:winding}
\end{figure}

\begin{proof}
	Recall the setting of Proposition~\ref{prop:relation-TeTm}: $abcdef=w$ is a face of~$\Omega$ and $m=\wrho$ is its midline; $\weOne$, $\weTwo$, $\weThree$, $\weFour$, $\weFive$, $\weSix$ are the six faces around~$w$ (see Fig~\ref{fig:t-edge-mid-rel}); 
	a tiling~$(G,\Phi[0])$ coincides with the triangulation dual to~$\Omega$ everywhere except at~$w$, and this face of~$\Omega$ is represented by two vertices~$\wrho$ and~$\wmrho$ in~$(G,\Phi[0])$; the quadrilateral~$r_m:=\weFour \wrho  \weOne \wmrho$ is a rhombuses whose angle at vertex $\weFour$ is~$0$.
	Construct a domain wall representation~$\gamma$ for spin configuration~$\sigma$ on~$G$:
	\begin{itemize}
		\item (half-)edges~$m_{\mathrm{up}}$, $a$, $b$, or~$m_{\mathrm{down}}$ are in~$\gamma$ if and only if~$\sigma^{[\rho]}$ differs from~$\sigma(\weOne)$, $\sigma(\weTwo)$, $\sigma(\weThree)$, or~$\sigma(\weFour)$ respectively;
		\item (half-)edges~$\overline{m}_{\mathrm{up}}$, $a$, $b$, or~$\overline{m}_{\mathrm{down}}$ are in~$\gamma$ if and only if~$\sigma^{[-\rho]}$ differs from~$\sigma(\weFour)$, $\sigma(\weFive)$, $\sigma(\weSix)$, or~$\sigma(\weOne)$ respectively;
		\item other (half-)edges are in~$\gamma$ if and only if they separate two faces with opposite spins.
	\end{itemize}
	Then, \eqref{eq:Trho-Ising-as-sums} and~\eqref{eq:Tmrho-Ising-as-sums} follow directly from~\eqref{eq:trho-via-double-spin} and~\eqref{eq:tmrho-via-double-spin} respectively; we also directly obtain the first two terms in~\eqref{eq:Rrho-Ising-as-sums} from~\eqref{eq:rrho-via-double-spin}.
	In order to get the terms~$\Weight^\b_\Omega(m_{\mathrm{down}},\overline{m}_{\mathrm{up}})$ and~$\Weight^\b_\Omega(m_{\mathrm{up}},\overline{m}_\mathrm{down})$, we take any~$\gamma$ containing~$c$ or~$f$ and split this edge into two half-edges.
	Clearly, in this way we count twice every~$\gamma\in\ConfO{\b}$ that contains both~$c$ and~$f$~--- and this is exactly what we need, because
	\[
		w'(0) = -\tfrac2{\sqrt{3}} = 2 u'(0).\qedhere
	\]
\end{proof}

\subsection{Fermions in the planar Ising model}
\label{sub:fermion}

The crucial tool that allows us to prove convergence results for correlations of the stress energy-tensor~$T(z)$ is the \emph{discrete fermionic observable}. 
It was proposed by the third author (see~\cite[Section~4]{Smi10a}) as a convenient tool to study the scaling limit of the Ising model (and other models) considered in \emph{general} planar domains; see also~\cite[Section~3]{CheCimKas15} for a discussion of their relations with dimers, spin-disorders and Grassmann variables techniques.

\subsubsection{Definitions of discrete fermionic observables}

For an edge~$e$ of~$\Omega$, let~$z_e$ denote its midpoint. Recall that we identify oriented edges~$e$ of~$\Omega$ with the half-edges emanating from~$z_e$ in the same directions.
Given an oriented edge~$e$, we denote the oppositely oriented edge by~$\overline{e}$.
We use the same notation for oppositely oriented edges emanating from~$z_e$.

Recall that~$\wp$ denotes the set of all six possible directions of the oriented edges in~$\Hex$.
We split these directions into positive and negative ones: 
\[
	\wp=\wp^+\cup\wp^-:=\{1,\varrho,\varrho^2\}\cup \{-1,-\varrho,-\varrho^2\}.
\]
For~$\rho\in \wp$, we define~$\eta(\rho)$ as a particular choice of its square root: 
\begin{equation}
	\label{eq:eta=rho}
	\eta(\rho):=e^{i\frac{\pi}{4}}\overline{\rho}^{\frac{1}{2}}=e^{\pm i\frac{\pi}{4}}\rho,\quad \rho\in\wp^\pm.
\end{equation}
For an oriented edge~$a$, we denote its direction by~$\rho_a$ and define~$\eta_a:=\eta(\rho_a)$.

\begin{definition}
\label{def_fermion}
Let~$a,e$ be two distinct half-edges in a discrete domain~$\Omega$. For a configuration~$\gamma\in\ConfO{a,e}$, we define its \emph{complex phase}~$\phase(\gamma;a,e)$ as
\begin{equation}
\label{eq_phase_def}
\phase(\gamma;a,e) := \exp \left[ -{\textstyle\frac{i}{2}}\wind(\gamma;a,e) \right],
\end{equation}
where by~$\wind(\gamma;a,e)$ we denote the total rotation angle of the (unique) path in~$\gamma$ linking the half-edges~$a$ and~$e$, oriented from~$a$ to~$e$ (so that it starts in the direction of the half-edge~$a$ and ends in the direction opposite to~$e$). Then, the \emph{real-valued fermionic observable}~$\F(a,e)$ is defined as
\begin{equation}
\label{eq_F_real_def}
	\F(a,e)~:=~\frac{{(-i\eta_a{\vphantom{1}}\overline{\eta}_e)}}{\cZ_\Omega^+}\sum_{\ConfO{a,e}}\phase(\gamma;a,e)\weight(\gamma).
\end{equation}
Further, for~$z_e\ne z_a$, the \emph{complex-valued fermionic observable}~$\F(a,z_e)$ is defined as
\begin{align}
\notag
\F(a,z_e):= & ~\eta_e\F(a,e)+\eta_{\overline{e}}\F(a,\overline{e})\\
\label{eq_F_complex_def} = & \frac{{(-i\eta_a)}}{\cZ_\Omega^+}\biggl[\,\sum_{\ConfO{a,e}}\phase(\gamma;a,e)\weight(\gamma)\ \ + \sum_{\ConfO{a,\overline{e}}}\phi(\gamma;a,\overline{e})\weight(\gamma)\biggr].
\end{align}
Above we set~$\ConfO{a,\overline{b}}:=\emptyset$ if~$e=b$ is a boundary half-edge. We also formally define $\F(a,a):=0$ and~$\F(a,z_a):=\eta_{\overline{a}}\F(a,\overline{a})$.
\end{definition}
\begin{remark} \label{rem:fermionic-via-Grassmann}
Since~$\eta_{\overline{e}}=\pm i\eta_e$ for~$\rho_e\in\wp^\pm$, we have
\begin{equation}
\label{eq_F_real_via_complex}
\F(a,e)=\Re[\overline{\eta}_e\F(a,z_e)]\,,\quad \F(a,\overline{e})=\Re[\overline{\eta}_{\overline{e}}\F(a,z_e)]=\pm\Im[\overline{\eta}_e\F(a,z_e)]\,.
\end{equation}

\begin{remark}
\label{rem:bvp_in_discrete}
It is easy to check that~$\F(a,e)$ is real and~$\F(e,a)=-\F(a,e)$ for all~$a,e$, since
$\wind(\gamma;e,a)=-\wind(\gamma;a,e)$, for all $\gamma\in\ConfO{a,e}$.
In particular, the function~$\F(a,\,\cdot\,)$ satisfies the following boundary conditions:
\begin{equation}
\label{eq_F_boundary_conditions}
\Re[i\overline{\eta}_b\F(a,z_b)]=\Im[\F(a,b)]=0
\end{equation}
for all boundary half-edges~$b$.
\end{remark}

Using the Grassmann variables formalism (see~\cite[Sections~3.2,3.6]{CheCimKas15} for the relation of combinatorial fermionic observables with this notation), for~$z_e\ne z_a$, one can write
\[
F(a,e)=\CorrFerm{}{\phi_e\phi_a}=-\CorrFerm{}{\phi_a\phi_e}\quad \text{and}\quad \F(a,z_e)~=~\CorrFerm{}{\psi(z_e)\phi_a},
\]
where $\psi(z_e):=i(\overline{\eta}_e\phi_e+\overline{\eta}_{\overline{e}}\phi_{\overline{e}})$ and the formal correlators~$\CorrFerm{}{\phi_e\phi_a}$ are associated with the classical Pfaffian representation of the Ising model partition function~$\cZ_\Omega^+$, see~\cite{CheCimKas15} for more details. Note that our definitions differ from those discussed in~\cite{CheCimKas15} in three respects. First, we drop the additional normalization~$t_e=(x_{\mathrm{crit}}+x_{\mathrm{crit}}^{-1})^{1/2}$ in the definition of the complex-valued observable: this factor does not depend on~$e$ and thus is irrelevant when working with the homogeneous model. Second,~(\ref{eq_F_complex_def}) does~\emph{not} contain the additional complex factor~$\exp[-i\frac{\pi}{4}]$, thus it differs by this factor from the definition used, e.g., in~\cite{CheSmi12,CheIzy13,CheCimKas15} and coincides with the one used, e.g., in~\cite{HonSmi13,Hon10,CheHonIzy15}. Third, our definition of the set~$\ConfO{a,\overline{a}}$ differs from the set~$\mathcal{C}(a,\overline{a})$ used in~\cite[Theorem~1.2]{CheCimKas15}: one is the complement of the other, thus
\[
F(a,\overline{a})~=~\pm 1~+\CorrFerm{}{\phi_{\overline{a}}\phi_{a}},\quad \rho_a\in\wp^\pm
\]
(recall that, according to our choice of square roots, we have~$\eta_{\overline{a}}=i\eta_a$ if~$\eta_a\in\{1,\rho,\rho^2\}$).
\end{remark}

\subsubsection{S-holomorphicity of discrete fermionic observables} \label{sub:s-holomorphicity}
In this section we briefly discuss the crucial property of complex-valued fermionic observables~(\ref{eq_F_complex_def}): the so-called \emph{s-holomorphicity}, which allows one to think about these observables as about solutions to some discrete boundary value problems and to prove their convergence (as~$\delta\to 0$) to holomorphic functions solving the similar boundary value problems in continuum.
\begin{definition}
\label{def:s-hol}
Let~$\Omega$ be a discrete domain drawn on the regular hexagonal grid and~$F$ be a complex-valued function defined on (a subset of) the set of medial vertices (aka midedges)~$z_e$ of~$\Omega$. We say that~$F$ is an \emph{s-holomorphic function}, if one has
\begin{equation}
\label{eq:s-hol-rel}
	\mathrm{Proj}[F(z_e)\,;(w-u)^{-\frac{1}{2}}\R]=\mathrm{Proj}[F(z_{e'})\,;(w-u)^{-\frac{1}{2}}\R]
\end{equation}
for all pairs of adjacent medial vertices~$z_e,z_{e'}$, where~$u$ denotes the common vertex of the edges~$e$ and~$e'$, a point~$w$ is the center of the face adjacent to both~$e$ and~$e'$, and
\[
\mathrm{Proj}[F\,;\tau\R]:=\tfrac{1}{2}[\,F+(\tau/|\tau|)^2\overline{F}\,]
\]
denotes the orthogonal projection of the complex number~$F$, considered as a point in the plane, onto the line~$\tau\R\subset\C$.
\end{definition}

\begin{remark}
	This property of discrete fermionic observables, which is a stronger version of usual definitions of discrete holomorphic functions, was pointed out by the third author in~\cite{Smi06} and used to prove convergence results for similar observables appearing in the random-cluster representation of the Ising model (considered on the square lattice) in~\cite{Smi10}. The name \mbox{s-holomorphicity} was suggested in~\cite[Section~3]{CheSmi12}, where similar convergence results were proved for the spin representation of the Ising model, considered on arbitrary isoradial graphs; see also~\cite[Section~4]{Smi10a}, \cite[Section~3]{CheCimKas15} and reference therein for a discussion of relations between s-holomorphic functions and other techniques used to study the 2D Ising model (dimers, disorder insertions, etc). Also, it is worth mentioning that there are two versions of this definition appearing in the literature, which differ from each other by the factor~$\exp[-i\frac{\pi}{4}]$. The definition given above coincides with those used, e.g., in~\cite{HonSmi13,Hon10,CheHonIzy15} and differs by this factor from the one used, e.g., in~\cite{Smi10,CheSmi12,CheIzy13,CheCimKas15}, cf. Remark~\ref{rem:fermionic-via-Grassmann}.
\end{remark}

It is well-known that, for any half-edge~$a_\delta$ of~$\Omega_\delta$, the complex-valued fermionic observable~$\Fdelta(a_\delta,\,\cdot\,)$ defined by~(\ref{eq_F_complex_def}) satisfies the s-holomorphicity condition everywhere in~$\Omega_\delta$ except at the midedge~$z_{a_\delta}$, we recall this fact in Proposition~\ref{prop:s-hol}. This allows to define the (real) fermionic observable on the corners of the hexagonal lattice, i.e. on the pairs $(u_\delta,v_\delta)$ as in Definition~\ref{def:s-hol}. We will denote it by~$\Fdelta (a_\delta,u_\delta,v_\delta)$.
According to the direction of~$u_\delta - v_\delta$ there are six type of corner. Each of them forms a triangular lattice, and the fact that~$\Fdelta (a_\delta,\cdot,\cdot)$ are projections of a function defined on the edges implies that on each type of corners the function~$\Fdelta (a_\delta,\cdot,\cdot)$ is harmonic away from~$a_\delta$ (see Appendix and Fig.~\ref{fig:full_plane_projections} for more details).

\begin{remark}
\label{rem:propag-eq-and-mercedes}
	We now describe in which sense the relation~\eqref{eq_relation_Te} on~$\Te$ can be viewed as a part of the s-holomorphicity relations~\eqref{eq:s-hol-rel}.
	Indeed, let~$u$ be a vertex surrounded by the faces of~$\Omega$ centred at~$w_1,w_2,w_3$.
	For~$j=1,2,3$, denote the second common vertex of the faces centered~$w_j$ and~$w_{j+1}$ (in the cyclic order) by~$v_{j+2}$
	It is easy to see that if~$F$ is s-holomorphic, then
	\[
		F(v_1,w_2) - F(v_2,w_1) = F(v_2,w_3) - F(v_3,w_2) = F(v_3,w_1) - F(v_1,w_3),
	\]
	which coincides with~\eqref{eq_relation_Te}.
	In addition, one has
	\[
		F(v_1,w_2) + F(v_2,w_1) + F(v_2,w_3) + F(v_3,w_2) + F(v_3,w_1) + F(v_1,w_3) = 0.
	\]
	One can show that these two display equations are in fact equivalent to the s-holomorphicity.
\end{remark}

Moreover, the function~$\Fdelta(a_\delta,\,\cdot\,)$ satisfies the boundary condition~(\ref{eq_F_boundary_conditions}) for all boundary (half-)edges of~$\Omega_\delta$, which is nothing but the discrete analogue of the boundary condition~(\ref{eq_feta_boundary_conditions}).

\setcounter{equation}{0}
\section{Correlation functions via fermions in discrete}
\label{sec:via-fermions-disc}

The aim of this section is to give expressions similar to Proposition~\ref{prop_via_fermion_t} for the expectation~$\mathbb{E}^\dob[T(w)]$ as well as for the two-point expectations~$\E^+[\Trho(w)\Ttau(w')]$, $\E^+[\Trho(w)\varepsilon(a)]$ and also for the ratio~$\E^+[\Trho(w)\sigma(u)]/\E^+[\sigma(u)]$. Recall that, for an (oriented) edge~$a$ of~$\Omega$, we define the energy density~$\varepsilon(a)$ as
\begin{equation}
\label{eq_e_definition}
\varepsilon(a)~:=~\sigma(a^+)\sigma(a^-)-\textstyle{\frac{2}{3}} ~=~ \varepsilon(\overline{a}),
\end{equation}
where~$a^+$ and~$a^-$ are the two adjacent faces to~$a$ and the lattice-dependent constant~$\frac{2}{3}$ corresponds to the infinite-volume limit of the Ising model on the honeycomb lattice (e.g., it should be replaced by~$\sqrt{2}/2$ when working on the square lattice). Since~$\wind(\gamma;a,\overline{a})=\pm 2\pi$ for any configuration~$\gamma\in\ConfO{a,\overline{a}}$, it easily follows from~(\ref{eq_F_real_def}) that
\begin{equation}
\label{eq_e_via_fermion}
\E^+[\varepsilon(a)] ~=~ 1-2(\cZ_\Omega^+)^{-1}\Weight_\Omega^+(a,\overline{a})-\tfrac{2}{3} ~=~ {\textstyle{\frac{1}{3}}}\mp 2\F(a,\overline{a})\quad \text{if}\ \ \rho_a\in\wp^\pm.
\end{equation}

\subsection{Stress-energy tensor expectations via fermionic observables} \label{sub:T-via-fermionic}

Consider the empty boundary conditions~$\b=\varnothing$. In this case, Lemma~\ref{lemma:Tm-Ising-as-sums} expresses the quantity~$\Tm(m)$ as a linear combination of sums over the sets $\ConfO{m_{\mathrm{down}},m_{\mathrm{up}}}$, $\ConfO{m_{\mathrm{down}},\overline{m}_{\mathrm{up}}}$ etc. It is easy to see that the complex phase~(\ref{eq_phase_def}) of a configuration~$\gamma$ is constant on each of these sets due to topological reasons. Therefore, one can easily rewrite~$\Tm(m)$ using several values $\F(m_{\mathrm{down}},m_{\mathrm{up}})$, $\F(m_{\mathrm{down}},\overline{m}_{\mathrm{up}})$ of the real fermionic observable introduced above.
To shorten the expressions, we introduce notation for discrete derivatives and mean values.

\begin{definition}
\label{def:discrete-derivatives}
For a function~$F(a,e)$ defined on half-edges of a discrete domain~$\Omega$, an inner face~$w$ of~$\Omega$ and $\rho\in\wp$,
we put
\[
\begin{array}{ll}
[\dOne F](\wrho,e):=\frac{1}{\sqrt{3}} [ F(\wup,e) - F(\wdown,e) ]\,,& [\sOne F](\wrho,e) := \frac{1}{2} [ F(\wup,e) + F(\wdown,e) ]\,, \\{}
[\dTwo F](a,\wrho):=\frac{1}{\sqrt{3}} [ F(a,\wup) - F(a,\wdown) ]\,,& [\sTwo F](a,\wrho) := \frac{1}{2} [ F(a,\wup) + F(a,\wdown) ]\,,
\end{array}
\]
where~$\wrho$ denotes the midline of the face~$w$ oriented in the direction~$i\rho$ and $\wdown,\wup$ are the two boundary half-edges of~$w$ oriented in the direction~$\rho$, see Fig.~\ref{fig:t-edge-mid-rel} and~\ref{fig:winding}.
\end{definition}

\begin{proposition} \label{prop_via_fermion_t}
Let~$w$ be an inner face of~$\Omega$ { and~$\rho\in\wp^\pm$}.
Then one has
\begin{align*}
\E^+[\Trho(w)] &~=~ \CmT~{-}~{\textstyle\frac{\sqrt{3}}{2}}[\sOne\dTwo\F](\wrho,\wrho)\,,\\
\E^+[\Rrho(w)] &~=~ \CmR~{\pm}~{\sqrt{3}}\,[\dOne\dTwo\F](\wrho,\wmrho)\,,
\end{align*}
where
$\Trho(w)$ and $\Rrho(w)$ are given by Definition~\ref{def:T-via-spins}, and~$\F$ is defined by~(\ref{eq_F_real_def}).
\end{proposition}

\begin{proof}
Recall that $F(\wdown,\wdown)=F(\wup,\wup)=0$ and~$F(\wdown,\wup)=-F(\wup,\wdown)$.
Thus,
\begin{align*}
{\textstyle\frac{\sqrt{3}}{2}}[\sOne\dTwo\F](\weta,\weta) =
{\textstyle\frac{\sqrt{3}}{2}}\cdot\tfrac{1}{2}\cdot {\textstyle\frac{1}{\sqrt{3}}}
\left[
\F(\wdown, \wup) - \F(\wup, \wdown)
\right] = \tfrac{1}{2}\F(\wdown, \wup)\,.
\end{align*}
Since~$\phase(\gamma;\wdown,\wup) = -i$ for any~$\gamma\in \ConfO{\wdown, \wup}$, it is easy to see that the right-hand side coincides with the the right-hand side of~\eqref{eq:Trho-Ising-as-sums}, see also Fig.~\ref{fig:winding}.
The expectation~$\E^+[\Reta(w)]$ can be treated in the same way.
\end{proof}

\subsection{Four-point fermionic observables} In order to handle the correlations mentioned above, we need a combinatorial definition of \emph{four-point} fermionic observables. Given four distinct half-edges~$a_1,a_2,a_3,a_4$ of~$\Omega$, let~$\ConfO{a_1,a_2|a_3,a_4}$ denote the set of configurations~$\gamma\in\ConfO{a_1,a_2,a_3,a_4}$ containing several loops and two paths linking~$a_1$ with~$a_2$ and~$a_3$ with~$a_4$. Note that
\[
\ConfO{a_1,a_2,a_3,a_4}=\ConfO{a_1,a_2|a_3,a_4}\sqcup \ConfO{a_1,a_3|a_2,a_4}\sqcup \ConfO{a_1,a_4|a_2,a_3}
\]
since we work with subsets of a trivalent lattice. The next combinatorial definition of four-point fermionic observables can be thought of as a generalization of the definition~(\ref{eq_F_real_def}). The~$2n$-point version of this construction was used in~\cite{Hon10} to treat the scaling limits of expectations~$\E^\b[\varepsilon(a_1)\dots\varepsilon(a_n)]$ on the square lattice, for any boundary conditions~$\b=\{b_1,\dots,b_{2m}\}$.

\begin{definition}
\label{def_fermion_4pt}
Let~$a_1,a_2,a_3,a_4$ be four distinct half-edges in a discrete domain~$\Omega$. For a configuration~$\gamma\in\ConfO{a_1,a_2|a_3,a_4}$, we define its \emph{complex phase}~$\phase(\gamma;a_1,a_2|a_3,a_4)$ as
\[
\phase(\gamma;a_1,a_2|a_3,a_4) := \exp \left[ -{\textstyle\frac{i}{2}}(\wind(\gamma;a_1,a_2)+\wind(\gamma;a_3,a_4))\right],
\]
where~$\wind(\gamma;a_1,a_2)$ and~$\wind(\gamma;a_3,a_4)$ stand for the total rotation angles of the (unique) paths in~$\gamma$ linking~$a_1$ with~$a_2$ and~$a_3$ with~$a_4$, oriented from~$a_1$ to~$a_2$ and from~$a_3$ to~$a_4$, respectively. Then, we put
\[
\F(a_1,a_2|a_3,a_4)~:=~\frac{(-i\eta_{a_1}^{\vphantom{1}}\overline{\eta}_{a_2})(-i\eta_{a_3}^{\vphantom{1}}\overline{\eta}_{a_4})} {\cZ_\Omega^+}\!\!\!\!\! \sum_{\gamma\in\ConfO{a_1,a_2|a_3,a_4}} {\phase(\gamma;a_1,a_2|a_3,a_4)\weight(\gamma)}
\]
and define the real-valued \emph{four-point fermionic observable}~$\F(a_1,a_2,a_3,a_4)$ by
\[
\F(a_1,a_2,a_3,a_4)~:=~ \F(a_1,a_2|a_3,a_4)-\F(a_1,a_3|a_2,a_4)+\F(a_1,a_4|a_2,a_3)\,.
\]
\end{definition}

The next lemma reflects the free fermionic structure of these combinatorial observables.

\begin{figure}[ht]
    \begin{center}
      \tempskiped{\includegraphics[scale=1.2]{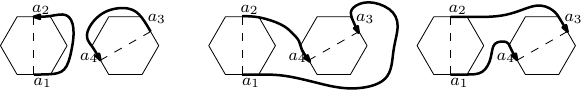}}
    \end{center}
    \caption{
    Three possible ways to connect four half-edges~$a_1$, $a_2$, $a_3$ and~$a_4$
    (the orientation of curves are chosen as in Lemma~\ref{lem_pfaffian}).
    The complex phases are
    $\exp \left[ -\tfrac{i}{2} (\tfrac{\pi}{2}+\tfrac{\pi}{2}) \right]$,
    $\exp \left[ -\tfrac{i}{2} (\tfrac{5\pi}{3} - \tfrac{\pi}{3}) \right]$ and
    $\exp \left[ -\tfrac{i}{2} (-\tfrac{\pi}{3} - \tfrac{\pi}{3}) \right]$, respectively.
    }
    \label{figTwoPaths}
\end{figure}

\begin{lemma}
\label{lem_pfaffian}
Let~$a_1,a_2,a_3,a_4$ be four distinct half-edges of~$\Omega$. Then, the following Pfaffian formula is fulfilled:
\[
\F(a_1,a_2,a_3,a_4)~=~ \F(a_1,a_2)\F(a_3,a_4)-\F(a_1,a_3)\F(a_2,a_4)+\F(a_1,a_4)\F(a_2,a_3)\,.
\]
In particular,~$\F(a_1,a_2,a_3,a_4)$ is an antisymmetric function of its arguments: for any permutation~$\varpi\in S_4$, one has~$\F(a_{\varpi(1)},a_{\varpi(2)},a_{\varpi(3)},a_{\varpi(4)})=(-1)^{\sgn(\varpi)}\F(a_1,a_2,a_3,a_4)$.
\end{lemma}

\begin{proof}
A similar result was proved in~\cite[Proposition~84]{Hon10} in the square lattice setup. This proof is based on a uniqueness theorem for the complex-valued fermionic observables viewed as the so-called s-holomorphic functions defined in a discrete domain~$\Omega$; the same argument can be also found in~\cite[Section~3.5]{CheIzy13}. In principle, it can be easily repeated in the honeycomb lattice setup, but we prefer to refer the reader to a general version of the same result provided, e.g., by~\cite[Theorem~1.2]{CheCimKas15}. In~\cite{CheCimKas15}, the definition of the four-point observable given above appears as a combinatorial expansion of the formal correlator
\begin{equation}
\label{eq_F_Pfaff}
\F(a_1,a_2,a_3,a_4)~=~\CorrFerm{}{\phi_{a_4}\phi_{a_3}\phi_{a_2}\phi_{a_1}}~=~ \mathrm{Pf}\big[\!\!\CorrFerm{}{\phi_{a_q}\phi_{a_p}}\!\!\big]_{p,q=1}^4 =~\mathrm{Pf}\big[F(a_p,a_q)\big]_{p,q=1}^4
\end{equation}
of Grassmann variables~$\phi_a$, which are naturally associated with a classical Pfaffian formula for the Ising model partition function. The only special case which should be considered separately is when the set~$\{a_1,a_2,a_3,a_4\}$ contains a pair of opposite half-edges, since in this case our definition of the set~$\ConfO{a_1,a_2,a_3,a_4}$ differs from the set~$\mathcal{C}(a_1,a_2,a_3,a_4)$ appearing in the expansion provided by~\cite[Theorem~1.2]{CheCimKas15}, cf.~Remark~\ref{rem:fermionic-via-Grassmann}. If, say,~$a_4=\overline{a}_3$ but $a_2\ne\overline{a}_1$, then $\ConfO{a_1,a_2,a_3,\overline{a}_3}=\mathcal{C}(a_1,a_2)\setminus\mathcal{C}(a_1,a_2,a_3,\overline{a}_3)$
and it is easy to check that
\[
\F(a_1,a_2,a_3,\overline{a}_3)~=~\CorrFerm{}{\phi_{a_2}\phi_{a_1}} + \CorrFerm{}{\phi_{\overline{a}_3}\phi_{a_3}\phi_{a_2}\phi_{a_1}}\quad \text{if}\ \ \rho_{a_3}\in\wp^+.
\]
Thus, the Pfaffian identity~(\ref{eq_F_Pfaff}) still holds true in this situation due to~(\ref{eq_F(aa)=1+<aa>}). If both~\mbox{$a_1=\overline{a}_2$} and~$a_4=\overline{a}_3$, then the inclusion-exlusion formula applied to the sets~$\mathcal{C}(\varnothing)$, $\mathcal{C}(a_1,\overline{a}_1)$, $\mathcal{C}(a_3,\overline{a}_3)$ and~$\mathcal{C}(a_1,a_2,a_3,a_4)$ to get the set~$\ConfO{a_1,a_2,a_3,a_4}$ leads to
\[
\F(a_1,\overline{a}_1,a_3,\overline{a}_3)~=~1~+\CorrFerm{}{\phi_{\overline{a}_1}\phi_{a_1}} + \CorrFerm{}{\phi_{\overline{a}_3}\phi_{a_3}} + \CorrFerm{}{\phi_{\overline{a}_3}\phi_{a_3}\phi_{\overline{a}_1}\phi_{a_1}}\quad \text{if}\ \  \rho_{a_1},\rho_{a_3}\in\wp^+,
\]
which eventually gives the same Pfaffian formula for~$\F(a_1,\overline{a}_1,a_3,\overline{a}_3)$.
\end{proof}

\subsection{Two-point expectations via four-point observables} \label{sub:TT-via-fermionic}
We begin with an analogue of Proposition~\ref{prop_via_fermion_t} for Dobrushin boundary conditions.

\begin{proposition}
\label{prop_via_fermion_tdob}
Let~$w$ be an inner face of~$\Omega$ and~$\rho\in\wp^\pm$. Then, for any two boundary edges~$b_{},b'_{}$ of~$\Omega$, one has
\begin{align*}
\E^{\dob}[\Trho(w)] ~ =& \ \E^+[\Trho(w)] ~{-}~{\textstyle\frac{\sqrt{3}}{2}}\left[\F(b_{},b'_{})\right]^{-1} \\ & \quad \times\left[[\dTwo\F](b_{},\wrho)\cdot[\sTwo\F](b'_{},\wrho)-[\sTwo\F](b_{},\wrho)\cdot[\dTwo\F](b'_{},\wrho)\right], \\
\E^{\dob}[\Rrho(w)]~=& \ \E^+[\Rrho(w)]~{\pm}~{\sqrt{3}}\left[\F(b_{},b'_{})\right]^{-1}\\ & \quad \times\left[[\dTwo\F](b_{},\wrho)\cdot[\dTwo\F](b'_{},\wmrho)-[\dTwo\F](b_{},\wmrho)\cdot[\dTwo\F](b'_{},\wrho)\right],
\end{align*}
where $\Trho(w)$ and $\Rrho(w)$ are given by Definition~\ref{def:T-via-spins} and~$\F$ is defined by~(\ref{eq_F_real_def}).
\end{proposition}
\begin{proof}
Let~$\phase^{b,b'}$ denote the common value of~$\phase(\gamma;b,b')$ for~$\gamma\in\ConfO\b$, which is fixed by topological reasons and does not depend on~$\gamma$. Then, one can easily check that
\begin{align*}
\F(b,\wdown \, | \,  \wup,b') & = \frac{\left(-i\eta_{b}\overline{\eta}\right)\left(-i\eta\overline{\eta}_{b'}\right)}{\cZ_\Omega^{+}}\cdot
(-i\phase^{b,b'}) \cdot \Weight_\Omega(b,\wup\, | \, \wdown,b')\,,\\
\F(b,\wup\, | \, \wdown, b') & = \frac{\left(-i\eta_{b}\overline{\eta}\right)\left(-i\eta\overline{\eta}_{b'}\right)}{\cZ_\Omega^{+}}\cdot
i\phase^{b,b'} \cdot \Weight_\Omega(b,\wup\, | \, \wdown,b')\,,\\
\F(b,b' \, | \,  \wdown, \wup) & = \frac{\left(-i\eta_{b}\overline{\eta}_{b'}\right)\left(-i\eta\overline{\eta}\right)}{\cZ_\Omega^{+}}\cdot
(-i\phase^{b,b'}) \cdot \Weight_\Omega(b,b'\, | \, \wdown,\wup) \,,
\end{align*}
where we use the notation~$ \Weight_\Omega(a_1,a_2\, | \, a_3,a_4) :=
\sum_{\gamma\in \ConfO{a_1,a_2\, | \, a_3,a_4}}{\weight(\gamma)}$. Therefore,
\begin{align*}
\F(b,\wdown, \wup, b') = \frac{i\eta_{b}\overline{\eta}_{b'}}{\cZ_\Omega^{+}} \cdot\phase^{b,b'}\cdot \Weight^{\dob}_\Omega(\wdown,\wup) =
-\frac{\F(b,b')}{\cZ_\Omega^{\dob}}\cdot \Weight^{\dob}_\Omega(\wdown,\wup)\,,
\end{align*}
where we used the definition of~$\F(b,b')$ in the last equality. Together with~\eqref{eq:Trho-Ising-as-sums} and the definition of~$\F(\wdown,\wup)$ this implies
\begin{align*}
\E^{\dob}[\Trho(w)] - \E^+[\Trho(w)] &=
\frac{1}{2}\biggl[ \frac{1}{\cZ_\Omega^{\dob}}\cdot\Weight^{\dob}_\Omega(\wdown,\wup)
-  \frac{1}{\cZ_\Omega^{+}}\cdot\Weight^+_\Omega(\wdown,\wup)  \biggr]\\
&= \frac{1}{2}\biggl[ -\frac{\F(b,\wdown,\wup,b')}{\F(b,b')} + \F(\wdown,\wup)  \biggr] \\
&= \frac{1}{2\F(b,b')}\left[ \F(b,\wdown)\F(b',\wup) - \F(b,\wup)\F(b',\wdown) \right],
\end{align*}
where the last equality is a consequence of the Pfaffian representation of the four-point observable (see Lemma~\ref{lem_pfaffian}) and the anti-symmetry of two-point observables.
The claim for~$\E^{\dob}[\Trho(w)]$ follows immediately. The proof of the formula for~$\E^{\dob}[\Rrho(w)]$ is similar.
\end{proof}

Exactly the same techniques can be used to derive expressions of the two-point expectations~$\E^+[\Trho(w)\Ttau(w')]$ and~$\E^+[\Trho(w)\varepsilon(a)]$ in terms of two-point fermionic observables~$\F$. As in Proposition~\ref{prop_via_fermion_t} and Proposition~\ref{prop_via_fermion_tdob}, one can also write similar expressions for the two-point expectations~$\E^+[\Rrho(w)\Ttau(w')]$,~$\E^+[\Rrho(w)\Rtau(w')]$ and~$\E^+[\Rrho(w)\varepsilon(a)]$. They all contain one more discrete derivative operator~$\partial$ instead of~$\varsigma$ comparing to the same expectations with the local field~$\Rrho(w)$ replaced by~$\Trho(w)$, and thus disappear in the scaling limit. Below we do not include these explicit expressions in the presentation for shortness.

\begin{proposition}
\label{prop_via_ferm_tt}
(i) Let~$w,w'$ be two non-adjacent inner faces of~$\Omega$ and~$\rho,\tau\in\wp$. Then,
\begin{align*}
\E^+ &[\Trho(w)\Ttau(w')]
~ = \ \E^+[\Trho(w)]\cdot \E^+[ \Ttau(w')]\\
& +{\textstyle\frac{3}{4}}\left[[\sOne\dTwo\F](\wrho,\wptau)\cdot[\dOne\sTwo \F](\wrho,\wptau)- [\sOne\sTwo\F](\wrho,\wptau)\cdot[\dOne\dTwo\F](\wrho,\wptau) \right].
\end{align*}

\noindent (ii) Let~$w$ be an inner face of~$\Omega$, $\rho\in\wp$, and~$a$ be an inner edge of~$\Omega$ with~$\rho_a\in\wp^\pm$ and such that it is not adjacent to~$w$. Then, one has
\begin{align*}
\E^+ [\Trho(w)\varepsilon(a)]~ = \ &\E^+[\Trho(w)] \cdot \E^+\left[ \varepsilon(a)\right]\\
& {\mp} \sqrt{3}\left[[\sTwo \F](a,\wrho)\cdot[\dTwo \F](\overline{a},\wrho)-[\dTwo \F](a,\wrho)\cdot[\sTwo\F](\overline{a},\wrho)\right],
\end{align*}
where the local field (energy density)~$\varepsilon(a)$ is defined by~(\ref{eq_e_definition}).
\end{proposition}

\begin{proof}
In the same way as in the proofs of Propositions~\ref{prop_via_fermion_t} and~\ref{prop_via_fermion_tdob} one gets
\begin{align*}
\E^+[\Trho(w)\Ttau(w')] &- \E^+[\Trho(w)]\E^+[\Ttau(w')] \\
&= \tfrac{1}{4}\left[\F(\wdown,\wpdown,\wpup,\wup) - \F(\wdown,\wup)\F(\wpdown,\wpup) \right]\\
&= \tfrac{1}{4}\left[\F(\wdown,\wpdown)\F(\wpup,\wup) - \F(\wdown,\wpup)\F(\wpdown,\wup) \right].
\end{align*}
Then, the expression for~$\E^+[\Teta(w)\Tmu(w')]$ follows from the identity
\begin{align*}
\tfrac{1}{4}&\left[\F(\wdown,\wpdown)\F(\wpup,\wup)  - \F(\wdown,\wpup)\F(\wpdown,\wup) \right]\\
&=\tfrac{3}{4}\left[ [\sOne\dTwo\F](\wrho,\wptau)\cdot[\dOne\sTwo \F](\wrho,\wptau)- [\sOne\sTwo\F](\wrho,\wptau)\cdot[\dOne\dTwo\F](\wrho,\wptau) \right].
\end{align*}
The formula for~$\E^+ [\Teta(w)\varepsilon(a)]$ can be derived in the same way.
One just needs to substitute~$\wpdown$,~$\wpup$ with~$a$,~$\overline{a}$
and to use the identity~$\E^+[\varepsilon(a)-\tfrac{1}{3}] = \mp2\F(a,\overline{a})$, see~(\ref{eq_e_via_fermion}).
\end{proof}

\subsection{Spinor observables and correlations with the spin field} \label{sub:T-spin-via-spinors}

In order to handle the two-point expectations involving the spin field~$\sigma(u)$, we need a generalization of combinatorial observables discussed in Section~\ref{sub:T-via-fermionic}, the so-called spinor observables.
For a simply connected discrete domain~$\Omega$ and its inner face~$u\in\Omega$, we denote by~$\OmegaU$ the double-cover of~$\Omega$ branching around~$u$ (and not branching around any other face) and by~$\pi:\OmegaU\to\Omega$ the corresponding projection mapping. For an (oriented) edge ~$e$ of~$\OmegaU$, we denote by~$e^*\ne e$ the other (oriented) edge of~$\OmegaU$ such that~$\pi(e^*)=\pi(e)$.
One can view the next definition as a generalization of Definition~\ref{def_fermion} in the presence of the branching.

\begin{definition}
\label{def_spinor}
Let~$\OmegaU$ be the double-cover of a simply connected discrete domain~$\Omega$ branching around its inner face~$u$. For two half-edges~$a,e$ of~$\OmegaU$ such that~$\pi(a)\ne\pi(e)$ and a configuration~$\gamma\in\ConfO{\pi(a),\pi(e)}$, let $\loopsU(\gamma)$ denote the number of loops in~$\gamma$ that surround~$u$ and
\begin{equation}
\label{eq_phaseU_def}
\phaseU(\gamma;a,e) := \exp \left[ -{\textstyle\frac{i}{2}}\windU(\gamma;a,e) \right]\cdot (-1)^{\loopsU(\gamma)},
\end{equation}
where~$\windU(\gamma;a,e):=\wind(\gamma;\pi(a),\pi(e))$ if the path linking the projections~$\pi(a)$ and~$\pi(e)$ in~$\gamma$ lifts to a path linking~$a$ and~$e$ on the double-cover~$\OmegaU$, and~$\wind(\gamma;\pi(a),\pi(e))+2\pi$ if this path lifts to a path linking~$a$ with~$e^*$ on~$\OmegaU$. Then, the \emph{real-valued spinor observable}~$\FU(a,e)$ is defined as
\begin{equation}
\label{eq_F_spinor_def}
\FU(a,e)~:=~\frac{(-i\eta_{\pi(a)}\overline{\eta}_{\pi(e)})}{\E^+[\sigma(u)]\cdot\cZ_\Omega^+} \sum_{\ConfO{\pi(a),\pi(e)}}\phaseU(\gamma;a,e)\weight(\gamma)\,.
\end{equation}
We also formally set~$\FU(a,a)=\FU(a,a^*):=0$ and, for a midedge~$z_e$ on~$\OmegaU$, define the \emph{complex-valued spinor observable}~$\FU(a,z_e)$ by
\[
\F(a,z_e):= {\eta}_{\pi(e)}\FU(a,e)+{\eta}_{\pi(\overline{e})}\FU(a,\overline{e}).
\]
\end{definition}

\begin{remark}
\label{rem:spinor-discrete} As in the non-branching case, one can easily check that~$\FU(a,e)$ is real and anti-symmetric:~$\FU(a,e)=-\FU(e,a)$. Moreover, it directly follows form the definition of~$\windU(\gamma;a,e)$ that~$\FU(a,e^*)=-\FU(a,e)$ for all~$a,e$ in~$\OmegaU$.
\end{remark}

The next proposition is a straightforward analogue of Proposition~\ref{prop_via_fermion_t}. Similarly to Proposition~\ref{prop_via_ferm_tt}, we do not include a formula for the correlations~$\E^+[\Rrho(w)\sigma(u)]$ for shortness.

\begin{proposition}
\label{prop_via_fermion_ts}
Let~$w,u$ be two { distinct} inner faces of~$\Omega$ and~$\rho\in\wp$. Then, one has
\begin{align*}
\frac{\E^+[\Trho(w)\sigma(u)]}{\E^+[\sigma(u)]} &~=~ \CmT-{\textstyle\frac{\sqrt{3}}{2}}[\sOne\dTwo\FU](\wrho,\wrho)\,,
\end{align*}
where
the spinor observable~$\FU$ is defined by~(\ref{eq_F_spinor_def}) and the arguments~$\wup,\wdown$ of~$\FU$ are assumed to be lifted on the same sheet of the double-cover~$\OmegaU$.
\end{proposition}

\begin{proof}
For each~$\gamma \in \ConfO{\pi(\wdown),\pi(\wup)}$, denote by~$\gamma'$ the
configuration obtained from~$\gamma$ by making the XOR operation on the right semi-hexagon of~$\pi(\wrho)$.
Then one has
\begin{align*}
\phaseU (\gamma;\wdown,\wup) = -i(-1)^{\loopsU(\gamma')}
\end{align*}
and~$(-1)^{\loopsU(\gamma')}$ is just the value of~$\sigma(u)$ in the spin configuration corresponding to~$\gamma'$.
Using the same arguments as in Lemma~\ref{lemma:Tm-Ising-as-sums},
we get the following analogue of~\eqref{eq:Trho-Ising-as-sums}:
\begin{align*}
\frac{\E^+[\Trho(w)\sigma(u)]}{\E^+[\sigma(u)]} - \CmT
&~=~ \frac{i\sum_{\gamma\in \ConfO{\pi(\wdown),\pi(\wup)}}{\weight(\gamma)\phaseU (\gamma;\wdown,\wup)}}{2\cZ^+\E^+[\sigma(u)]}\\
&~=~ -\tfrac{1}{2}\FU(\wdown,\wup) ~=~ -{\textstyle\frac{\sqrt{3}}{2}}[\sOne\dTwo\FU](\weta,\weta)\,,
\end{align*}
where the first equality follows from considerations of~$\phaseU (\gamma;\wdown,\wup)$ given above,
the second uses the definition of~$\FU$ and the third repeats the one in Proposition~\ref{prop_via_fermion_t}.
\end{proof}

\setcounter{equation}{0}
\section{Correlation functions via fermions in continuum} \label{sec:CFT-correlations}

In this section we briefly discuss correlation functions~$\CorrO{T(w)}$\,, $\CorrO[\dob]{T(w)}$,~$\CorrO{T(w)T(w')}$ etc coming from the (free fermion, central charge~$\frac{1}{2}$) Conformal Field Theory, which is known to describe the scaling limit of the 2D Ising model at criticality (e.g., see~\cite[Sec. 14.2]{Mus10}). It is worth noting that we do \emph{not} refer to the usual CFT language along this discussion and adopt the approach discussed in~\cite[Section~4]{Che18} and~\cite{CheHonIzy21} instead, though keeping the presentation self-contained for the sake of the reader. Thus, all the correlation functions in continuum that we need to formulate Theorems~\mbox{\ref{thm:Ising1}--\ref{thm:Ising3}} are \emph{defined} via solutions to appropriate Riemann-type boundary value problems for holomorphic functions in~$\Omega$. In particular, this approach fits our proofs of Theorems~\mbox{\ref{thm:Ising1}--\ref{thm:Ising3}}, which are presented in Section~\ref{sec:Ising-convergence} and are based on convergence results for solutions to similar \emph{discrete} boundary value problems.

\subsection{Fermionic correlators as solutions to boundary value problems for holomorphic functions} \label{sub:fermions-via-bvp}

For a while, let us assume that~$\Omega$ is a \emph{bounded} simply connected domain with a \emph{smooth} boundary~$\partial\Omega$.

\begin{definition}\label{def:feta}
Let~$\Omega$ be a bounded simply connected domain with a smooth boundary,~$a\in\Omega$, and $\eta\in \mathbb{C}$.
We define~$\feta(a,\cdot)$ to be the (unique) holomorphic in~$\Omega\setminus\{a\}$ and continuous up to~$\partial\Omega$ function such that
$\feta (a,z) = \overline{\eta}\cdot(z-a)^{-1}+O(1)$ as~$z\to a$ and
\begin{equation}
\label{eq_feta_boundary_conditions}
\Im\big[\feta (a,\zeta)\sqrt{\tau(\zeta)}\,\big]=0~\quad\text{for~all}\ \ \zeta\in\partial \Omega,
\end{equation}
where~$\tau(\zeta)$ denotes the tangent vector to~$\partial\Omega$ at the point~$\zeta$, oriented counterclockwise and considered as a complex number.
\end{definition}

\begin{remark} \label{rem:feta-def}
(i) The uniqueness of solution to this boundary value problem can be easily deduced from the following consideration: the difference~$f(z)$ of any two solutions should be holomorphic \emph{everywhere} in~$\Omega$, which implies~$\oint_{\partial\Omega} \Im[(f(\zeta))^2d\zeta]=0$ and contradicts to the boundary conditions.

\noindent (ii) Let~$\varphi \, : \, \Omega \to \Omega'$ be a conformal mapping and~$\eta' = \eta\cdot (\varphi'(a))^{\frac{1}{2}}$. Then one has
\begin{equation}
\label{eq_covariance_feta}
\feta (a,z) = \fetaprime (\varphi(a), \varphi(z))\cdot (\varphi'(z))^{\frac{1}{2}},\quad a,z\in\Omega\,.
\end{equation}
Indeed, the right-hand side satisfies all the conditions used to define~$\feta (a,z)$, thus this covariance rule follows from the uniqueness property discussed above.

\noindent (iii) The \emph{existence} of~$\feta(a,z)$ can be easily derived by solving the corresponding boundary value problem in the upper half-plane~$\H$ and then applying~(\ref{eq_covariance_feta}) with~$\varphi:\Omega\to\H$, but we prefer to postpone these explicit formulae, taking this existence for granted until Proposition~\ref{prop:fH_explicit}. Note that the explicit formulae~(\ref{eq_f_fdag_explicit}) and the decomposition~(\ref{eq_feta_as_f_plus_fdag}) given below can be also used to define the functions~$\feta(a,z)$ for rough and/or unbounded domains~$\Omega$.
\end{remark}

\begin{proposition} \label{propContFermDecomp}
Let~$\Omega$ be a bounded simply connected domain with a smooth boundary, and let~$a\in\Omega$. Then, there exists unique pair of holomorphic in~$\Omega\setminus\{a\}$ and continuous up to~$\partial\Omega$ functions~$\f(a,\cdot)$ and~$\fstar(a,\cdot)$ such that the functions~$\freg(a,z):=\f(a,z)-2(z\!-\!a)^{-1}$ and~$\fstar(a,z)$ are bounded as~$z\to a$ and
\begin{equation}
\label{eq_f_fdag_boundary_conditions}
\fstar(a,\zeta) = \overline{\tau(\zeta)\cdot \f(a,\zeta)}~\quad\text{for~all}\ \ \zeta\in\partial \Omega.
\end{equation}
Moreover, for any~$\eta\in\C$ and~$z\in\Omega\setminus\{a\}$, one has
\begin{equation}
\label{eq_feta_as_f_plus_fdag}
\feta (a,z)  = \tfrac{1}{2}[\overline{\eta}\f (a,z)  + \eta\fstar(a,z)]\,.
\end{equation}
\end{proposition}

\begin{proof} It is easy to see that, for any~$\alpha,\beta\in\R$ and~$\eta,\mu\in\C$, one has
\[
\f^{[\alpha\eta+\beta\mu]}(a,z)=\alpha\f^{[\eta]}(a,z)+\beta\f^{[\mu]}(a,z),
\]
since the right-hand side satisfies all the conditions from Definition~\ref{def:feta} that determine the function~$\f^{[\alpha\eta+\beta\mu]}(a,\cdot)$. In other words, for any~$z\in\Omega\setminus\{a\}$, the quantity~$\feta(a,z)$ is a real-linear functional of~$\eta\in\C$, and hence identity~(\ref{eq_feta_as_f_plus_fdag}) holds true for some~$\f(a,z)$ and~$\fstar(a,z)$. Since these functions can be easily represented as linear combinations of~$\feta(a,z)$, they both are holomorphic in~$\Omega\setminus\{a\}$ and continuous up to~$\partial\Omega$, and their behavior as~$z\to a$ follows immediately. Moreover, boundary conditions~(\ref{eq_feta_boundary_conditions}) considered for all~$\eta\in\C$ simultaneously are equivalent to boundary conditions~(\ref{eq_f_fdag_boundary_conditions}).
\end{proof}

\begin{remark}
(i) It follows from~(\ref{eq_covariance_feta}) and~(\ref{eq_feta_as_f_plus_fdag}) that
the functions~$\f(a,z)$ and~$\fstar(a,z)$ have the following covariances under conformal mappings $\varphi \, : \, \Omega \to \Omega'$:
\begin{equation}
\label{eq_covariance_f_fdag}
\begin{array}{l}
\f (a,z) = \fprime (\varphi(a), \varphi(z)) \cdot (\varphi^\prime(a)\varphi^\prime(z))^\frac{1}{2}\,, \\
\fstar (a,z) = \fstarprime (\varphi(a), \varphi(z))  \cdot (\overline{\varphi^\prime(a)}\varphi^\prime(z))^\frac{1}{2}\,,
\end{array}
\quad a,z\in\Omega.
\end{equation}

\noindent (ii) Let~$\eta,\mu \in \mathbb{C}$ and~$a,z\in\Omega$.  Applying the Cauchy residue theorem to the product of the functions~$\feta(a,\cdot)$ and $\fmu(z,\cdot)$ and taking into account boundary conditions~(\ref{eq_feta_boundary_conditions}), one obtains
\[
0 = \Im\biggl[\oint_{\partial \Omega} \feta(a,\zeta)\fmu(z,\zeta)d\zeta\biggr] = 2\pi \cdot \Re\left[\overline{\eta}\fmu(z,a)+\overline{\mu}\feta(a,z)\right].
\]
Since~$\eta$ and~$\mu$ can be chosen arbitrary, the decomposition~(\ref{eq_feta_as_f_plus_fdag}) easily implies the identities
\begin{equation}
\label{eq_symmetries_f}
\f(z,a) = - \f(a,z), \qquad \fstar(z,a) = -\overline{\fstar(a,z)}\,.
\end{equation}
In particular, the function~$\f(\cdot,z)$ is holomorphic in~$\Omega\setminus\{z\}$ while the function~$\fstar(\cdot,z)$ is anti-holomorphic in~$\Omega$.
\end{remark}

\begin{proposition} \label{prop:fH_explicit}
Let~$\Omega$ be a bounded simply connected domain with a smooth boundary and~$\varphi:\Omega\to\H$ be a conformal mapping from~$\Omega$ onto the upper half-plane~$\H$. Then, for any~$a,z\in\Omega$, one has
\begin{equation}
\label{eq_f_fdag_explicit}
\f(a,z)=\frac{2(\varphi'(a)\varphi'(z))^{\frac12}}{\varphi(z)-\varphi(a)}\,,\qquad \fstar(a,z)= \frac{2(\overline{\varphi'(a)}\varphi'(z))^{\frac12}}{\varphi(z)-\overline{\varphi(a)}}\,.
\end{equation}
\end{proposition}
\begin{proof} Note that the functions $f_\H(a,z):=2(z-a)^{-1}$ and $\fstarH(a,z):=2(z-\overline{a})^{-1}$
satisfy the conformal covariance rules~(\ref{eq_covariance_f_fdag}) for all M\"obius automorphisms~$\varphi:\H\to\H$ and solve the boundary value problem from Proposition~\ref{propContFermDecomp} in the upper half-plane~$\H$:
\[
\fstarH(a,\zeta)=\overline{f_\H(a,\zeta)}\quad \text{for}\ \ a\in\H\ \ \text{and}\ \ \zeta\in\partial\H=\R.
\]
Therefore, the functions~$\f(a,\cdot)$ and~$\fstar(a,\cdot)$ given by~(\ref{eq_f_fdag_explicit}) do not depend on the choice of the uniformization map~$\varphi:\Omega\to\H$, are holomorphic in~$\Omega\setminus\{a\}$ and continuous up to~$\partial\Omega$, have the correct behavior as~$z\to a$, and satisfy boundary conditions~(\ref{eq_f_fdag_boundary_conditions}) everywhere on~$\partial\Omega$.
\end{proof}

\begin{remark}\label{rem:CFT-terminology}
Adopting the colloquial CFT language (e.g., see \cite[Sec. 4, pp.20-21]{GhoZam94} for a discussion of boundary conditions in the upper half-plane), one can write
\begin{align*}
\f(a,z)~&=~\CorrFerm\Omega{\psi(z)\psi(a)}~=~-\CorrFerm\Omega{\psi(a)\psi(z)}\,,\\
\fstar(a,z)~&=~\CorrFerm\Omega{\psi(z)\opsi(a)}~=~-\CorrFerm\Omega{\opsi(a)\psi(z)}\,,
\end{align*}
where we use the notation~$\CorrFerm\Omega{\dots}$ for correlators of the free (holomorphic~$\psi(z)$ and anti-holomorphic~$\opsi(z)$) fermions in a simply connected domain~$\Omega$, coupled along the boundary so that $\opsi(\zeta)=\overline{\tau(\zeta)\psi(\zeta)}$ for~$\zeta\in\partial\Omega$; note that this coupling respects the conformal covariance~(\ref{eq_covariance_f_fdag}) and reads simply as~$\opsi(\zeta)=\overline{\psi(\zeta)}$ if~$\Omega=\H$ is the upper half-plane and~$\zeta\in\R$. Therefore, the function~$\feta(a,z)$ can be interpreted as the fermionic correlator
\[
\feta(a,z)~=~\CorrFerm\Omega{\psi(z)\psi^{[\eta]}(a)}\,,\quad\text{where}\quad \psi^{[\eta]}(a):=\tfrac{1}{2}[\overline{\eta}\psi(a)+\eta\opsi(a)]\,.
\]
We also refer the interested reader to~\cite[Section~3.6]{CheCimKas15}, where a similar `Grassmann variables' notation for the discrete combinatorial observables from Section~\ref{sec:via-fermions-disc} is discussed.
\end{remark}

\subsection{Stress-energy tensor and energy density correlations} \label{sub:T-and-energy-via-fermions}

In the free fermion CFT framework, it is well-known (e.g., see~\cite[Sec. 12.3 and 14.2.1]{Mus10}, note however that we use slightly different multiplicative normalization~(\ref{eq_f_fdag_explicit}) of fermions) that both the (holomorphic) stress-energy tensor~$T$ and the energy density field~$\varepsilon$ can be written in terms of~$\psi,\opsi$ as
\[
T\,=\,{-\textstyle\frac{1}{4}}:\!\psi\partial\psi\!:\quad\text{and}\quad \varepsilon\,=\,{\textstyle\frac{i}{2}}\psi\opsi
\]
(as usual, the colons mean that some regularization is required in the former case). This motivates the following definition, which uses the language of solutions to Riemann-type boundary value problems developed in Section~\ref{sub:fermions-via-bvp} instead of the CFT terminology.

\begin{definition} \label{def:T-and-eps-one-point}
Given a simply connected domain~$\Omega$ and~$w,a\in\Omega$, we set 
\[
\CorrO{T(w)}\,:=\,{\textstyle\frac{1}{4}}\partial_z \freg(w,z)\big|_{z=w}\quad\text{and}\quad \CorrO{\varepsilon(a)}\,:=\, {\textstyle\frac{i}{2}}\fstar(a,a)\,,
\]
where the functions~$\freg$ and~$\fstar$ are given by Proposition~\ref{propContFermDecomp} (and Definition~\ref{def:feta}).
\end{definition}

Note that one can easily derive explicit formulae for the quantities~$\CorrO{T(w)}$ and~$\CorrO{\varepsilon(a)}$ using~(\ref{eq_f_fdag_explicit}). Namely, for any uniformization mapping~$\varphi:\Omega\to \H$, one has
\begin{equation}
\label{eq_T_explicit}
\CorrO{T(w)}\,=\,\frac{[\schw\varphi](w)}{24}\,,\quad \text{where}\quad [\schw\varphi](w)\,:=\,\frac{\varphi'''(w)}{\varphi'(w)}-\frac{3}{2}\left[\frac{\varphi''(w)}{\varphi'(w)}\right]^2
\end{equation}
denotes the Schwarzian derivative of the conformal mapping~$\varphi$ at the point~$w$, and
\begin{equation}
\label{eq_e_explicit}
\CorrO{\varepsilon(a)}\,=\,|\varphi'(a)|\cdot (2\Im\varphi(a))^{-1}\,.
\end{equation}
In particular,~$\CorrO{\varepsilon(a)}\in\R$ for all~$a\in\Omega$ while~$\CorrO{T(\cdot)}$ is a holomorphic function in $\Omega$, this is why we prefer to use the letter~$w$ (instead of~$a$) for its argument.

The next step is to give a definition of the function~$\CorrO[\dob]{T(w)}$ from Theorem~\ref{thm:Ising1}, as well as two-point functions~$\CorrO{T(w)\varepsilon(a)}$ and~$\CorrO{T(w)T(w')}$ from Theorems~\ref{thm:Ising2} and~\ref{thm:Ising3}. It is worth mentioning that all these quantities can be alternatively defined via explicit formulae in the upper half-plane and appropriate conformal covariance rules. Nevertheless, we prefer to use the functions~$\f$ and~$\fstar$ as the starting point for all our definitions since they provide clear links with both the standard CFT notation on the one hand and with the discrete fermionic observables discussed in Section~\ref{sec:via-fermions-disc} on the other.

\begin{definition} \label{def:T-Dobrushin-expectation} Given a simply connected domain~$\Omega$, an inner point~$w\in\Omega$ and two boundary points (or prime ends) $b_{},b'_{}\in\partial\Omega$, we define the quantity~$\CorrO[\dob]{T(w)}$ by
\[
\CorrO[\dob]{T (w)} ~:=~
\CorrO{T(w)} ~+~ \frac{\f(b',w)\partial_w\f(b,w)-\f(b,w)\partial_w\f(b',w)}{4\f(b_{},b'_{})}\,,
\]
where 
the one-point function~$\CorrO{T(w)}$ is given by Definition~\ref{def:T-and-eps-one-point}.
\end{definition}

Note that for any uniformization mapping~$\varphi:\Omega\to\H$ the formula~(\ref{eq_f_fdag_explicit}) yields
\[
\CorrO[\dob]{T (w)} ~=~ \CorrO{T(w)} ~+~ \frac{1}{2}\left[\frac{\varphi'(w)\cdot(\varphi(b')-\varphi(b))} {(\varphi(b')-\varphi(w))(\varphi(b)-\varphi(w))}\right]^2.
\]
Therefore, if~$\varphi:\Omega\to\H$ is the (unique up to a scaling) uniformization mapping such that~\mbox{$\varphi(b_{})=0$} and~\mbox{$\varphi(b'_{})=\infty$} or vice versa, then
\begin{equation}
\label{eq_T_Dobrushin_explicit}
\CorrO[\dob]{T(w)}= \,\frac{1}{2}\left[\frac{\varphi'(w)}{\varphi(w)}\right]^2+\frac{[\schw\varphi](w)}{24}\,.
\end{equation}

\begin{remark} \label{rem:T-Dobrusin}
(i) Following the discussion given in Remark~\ref{rem:CFT-terminology}, one can interpret the ratio of partition functions of the Ising model in~$\Omega$ with Dobrushin and~`$+$' boundary conditions as the (absolute value) of the fermionic correlator~$\CorrFerm\Omega{\psi(b'_{})\psi(b_{})} = \f(b_{},b'_{})$. Note that this interpretation perfectly agrees with the combinatorial definition of discrete fermionic observables. Therefore, Definition~\ref{def:T-Dobrushin-expectation} can be thought about as a Pfaffian formula for the four-point fermionic correlator~$-\frac{1}{4}\CorrFerm\Omega{\;:\!\psi(w)\partial\psi(w)\psi(b'_{})\psi(b_{})\!:\;}$ normalized by~$\CorrFerm\Omega{\psi(b'_{})\psi(b_{})}$\,.

\noindent (ii) Note that the values of the function~$\f(w,\cdot)$ at boundary points~$b_{},b'_{}$ may be ill-defined if~$\partial\Omega$ is non-smooth. Nevertheless, the ratio involved in the definition of~$\CorrO[\dob]{T(w)}$ is always well-defined: the easiest way to see this is to apply a uniformization map~$\varphi:\Omega\to\H$ and to rewrite all the instances of~$\f(w,\cdot)$ using~(\ref{eq_f_fdag_explicit}). Along the way, both the numerator and the denominator gain the same ill-defined factor~$(\phi'(b_{})\phi'(b'_{}))^{\frac{1}{2}}$, which cancels out.
\end{remark}

\begin{definition}\label{def:TT-and-Te-correlations}
Given a simply connected domain~$\Omega$ and points~$w,w',a\in\Omega$, we define the quantities~$\CorrO{T(w)\varepsilon(a)}$ and~$\CorrO{T(w)T(w')}$ by
\begin{align*}
\langle T (w) T (w')\rangle^+_\Omega := & \langle T (w) \rangle^+_\Omega  \langle T(w') \rangle^+_\Omega\\
& + {\textstyle \frac{1}{16}}( \partial_w \f (w,w')\partial_{w'} \f (w,w') \!-\! \f (w,w')\partial_w \partial_{w'} \f (w,w')), \\
\langle T (w) \varepsilon (a) \rangle^+_\Omega := & \langle T (w) \rangle^+_\Omega \langle \varepsilon (a) \rangle^+_\Omega\\
& + \textstyle{\frac{i}{8}}(\f(a,w)\partial_w \fstar (a,w)
 - \fstar(a,w)\partial_w \f(a,w))\,,
\end{align*}
where 
the one-point functions~$\CorrO{T(w)}$ and~$\CorrO{\varepsilon(a)}$ are given by Definition~\ref{def:T-and-eps-one-point}.
\end{definition}

\begin{remark}\label{rem:TT-and-Te-Pfaff}
Similarly to Remark~\ref{rem:T-Dobrusin}, Definition~\ref{def:TT-and-Te-correlations} can be understood as Pfaffian formulae for the correlators~$\frac{1}{16}\!\CorrFerm{\Omega}{\;:\!\psi(w)\partial\psi(w)\psi(w')\partial\psi(w')\!:\;}$ and~$-\frac{i}{8}\!\CorrFerm{\Omega}{\psi(w)\partial\psi(w)\psi(a)\opsi(a)}$\,.
\end{remark}

\subsection{Correlations with the spin field} \label{sub:T-spin-correlations}
The main aim of this section is to define the quantities~$\CorrO{\sigma(u)}$ and~$\CorrO{T(w)\sigma(u)}$, which appear in the statement of Theorem~\ref{thm:Ising3}. In principle, this can be done by writing explicit formulae for the case~$\Omega=\H$ and then applying appropriate covariance rules. Nevertheless, similarly to Section~\ref{sub:fermions-via-bvp} we prefer to use the language of Riemann-type boundary value problems for holomorphic spinors branching over a point~$u\in\Omega$. Note that such continuous spinors are natural counterparts of discrete branching observables discussed in Section~\ref{sec:via-fermions-disc}, which are the crucial tool for proving the convergence of discrete correlation functions involving the spin field to their limits, see~\cite{CheIzy13,CheHonIzy15}.

Below we need some terminology. Given a simply connected planar domain~$\Omega$ and a point~$u\in\Omega$, we denote by~$\OmegaU$ the (canonical) double-cover of~$\Omega$ branching over the point~$u$, and by~$\pi:\OmegaU\to\Omega$ the corresponding projection. For a point~$z\in\OmegaU\setminus\{u\}$, let~$z^*\in\OmegaU$ be such that~$\pi(z^*)=\pi(z)$ and~$z^*\ne z$. In other words,~$z^*$ is the point lying on the other sheet of~$\OmegaU$ over the same point of~$\Omega$ as~$z$.
We say that a function~$f:\OmegaU\setminus\{u\}\to\C$ is a \emph{holomorphic spinor} if~$f(z^*)=-f(z)$ for all~$z\neq u$ and~$f(\cdot)$ is  holomorphic on~$\OmegaU\setminus\{u\}$.

The following description of the one-point function~$\CorrO{\sigma(u)}$ in continuum was used in~\cite{CheHonIzy15}, we refer the interested reader to this paper for more details.

\begin{definition}
\label{def:spinor} Let~$\Omega$ be a bounded simply connected domain with a smooth boundary and~$u\in\Omega$. We denote by~$\spinor(z)$  the (unique) holomorphic spinor defined on~$\OmegaU\setminus\{u\}$ such that~$\spinorReg(z):=\spinor(z)-e^{-i\frac{\pi}{4}}(z-u)^{-1/2}$ is bounded as~$z\to u$ and
\[
\Im\big[\spinor(\zeta)\sqrt{\tau(\zeta)}\,\big]=0~\quad\text{for~all}\ \ \zeta\in\partial \OmegaU\,.
\]
Then, we define the quantity~$\CorrO{\sigma(u)}$ as the (unique up to a multiplicative normalization) real-valued function in~$\Omega$ satisfying the differential equation
\begin{equation}
\label{eq_sigma_def}
d\log\CorrO{\sigma(u)}~=~\Re[\mathcal{A}_\Omega(u)du]~=~\Re{\mathcal{A}_\Omega(u)}dx-\Im\mathcal{A}_\Omega(u)dy\,,
\end{equation}
where~$u=x+iy$ and~$\mathcal{A}_\Omega(u):=\tfrac{1}{2}e^{i\frac{\pi}{4}}(z-u)^{-1/2}\spinorReg(z)\big|_{z=u}$\,.
\end{definition}
\begin{remark} (i) The uniqueness of solution~$\spinor(\cdot)$ to the boundary value problem described above can be easily deduced similarly to Remark~\ref{rem:feta-def}(i). Moreover, one can use this uniqueness property similarly to Remark~\ref{rem:feta-def}(ii) to prove that
\begin{equation}\label{eq_covariance_spinor}
\spinor(z) = \spinorprime(\varphi(z))\cdot (\varphi^\prime(z))^\frac{1}{2}
\end{equation}
for any conformal mapping~$\varphi:\Omega\to\Omega'$, where~$u'=\varphi(u)$. Then, one can use this covariance rule and the explicit formula
\[
g_{[\H,u]}(z)=(2\Im u)^{\frac{1}{2}}\cdot ((z-u)(z-\overline{u}))^{-\frac{1}{2}}
\]
to prove the existence of~$\spinor(\cdot)$, and to define this function for rough and/or unbounded~$\Omega$.

\noindent (ii) A priori, it is not clear that the differential form~(\ref{eq_sigma_def}) is exact. The easiest (though not the most conceptual) way to check this fact is to start with the explicit formula given above for~$\Omega=\H$ and then to pass to the general case using~(\ref{eq_covariance_spinor}).

\noindent (iii) Equation~(\ref{eq_sigma_def}) defines the function~$\CorrO{\sigma(u)}$ up to a multiplicative normalization which, in principle, could depend on~$\Omega$. There exists a coherent way to choose this normalization for all planar domains~$\Omega$ simultaneously, which leads to the following formula:
\begin{equation}
\label{eq_spin-explicit}
\CorrO{\sigma(u)}= (2|\varphi'(u)|)^{\frac{1}{8}}\cdot (\Im\varphi(u))^{-\frac{1}{8}},\quad u\in\Omega\,,
\end{equation}
where~$\varphi$ stands for a uniformization mapping from~$\Omega$ onto~$\H$, see~\cite{CheHonIzy15} for more details.
\end{remark}

\begin{definition}
\label{def:fetaU}
Let~$\OmegaU$ be a double-cover of a bounded simply connected domain~$\Omega$ with a smooth boundary, branching over the point~$u\in\Omega$. Let~$a\in\OmegaU\setminus\{u\}$ and $\eta\in \mathbb{C}$. We define~$\fetaU(a,\cdot)$ to be the (unique) holomorphic spinor in~$\OmegaU\setminus\{u,a,a^*\}$ such that
\[
\begin{array}{ll}
\Im\big[\fetaU (a,\zeta)\sqrt{\tau(\zeta)}\,\big]=0 & \text{for~all}\ \ \zeta\in\partial \OmegaU\,,\\
\fetaU(a,z) = \overline{\eta}\cdot(z-a)^{-1} + O(1) & \text{as}\ \ z\to a\,,\\
\fetaU(a,z) = e^{i\frac{\pi}{4}} c\cdot(z-u)^{-\frac{1}{2}}+O(|z-u|^{\frac{1}{2}}) & \text{as}\ \ z\to u,
\end{array}
\]
for some (unknown) constant~$c\in\R$.
\end{definition}

Repeating the arguments given in Section~\ref{sub:fermions-via-bvp} for holomorphic functions~$\feta(a,z)$, one can prove that the boundary value problem for holomorphic spinors~$\fetaU(a,z)$ described in~Definition~\ref{def:fetaU} has a unique solution which can be represented as
\[
\fetaU(a,z)  = \tfrac{1}{2}[\overline{\eta} \fU(a,z)  + \eta \fstarU(a,z)],\quad a,z\in\OmegaU\,.
\]
Moreover, one easily sees that the function
\[
\fregU (a,z)~ := ~\fU (a,z) - 2(z\!-\!a)^{-1}
\]
is bounded in a vicinity of the point~$a$ (note that this function is defined only locally). Further, holomorphic spinors~$\fU(a,\cdot)$ and~$\fstarU(a,\cdot)$ satisfy the same symmetries~(\ref{eq_symmetries_f}) and conformal covariance rules~(\ref{eq_covariance_f_fdag}), and one can use the latter to define these spinors in rough and/or unbounded domains starting with the explicit formulae
\[
f_{[\H,u]}(a,z)= \left[\frac{(a-u)(a-\overline{u})}{(z-u)(z-\overline{u})}\right]^{\frac{1}{2}}\cdot \left[\frac{2}{z-a}+\frac{1}{a-u}+\frac{1}{a-\overline{u}}\right],
\]
\[
f^\star_{[\H,u]}(a,z)= \left[\frac{(\overline{a}-u)(\overline{a}-\overline{u})}{(z-u)(z-\overline{u})}\right]^{\frac{1}{2}}\cdot \left[\frac{2}{z-\overline{a}}+\frac{1}{\overline{a}-u}+\frac{1}{\overline{a}-\overline{u}}\right].
\]
The next definition mimics Definition~\ref{def:T-and-eps-one-point} in the presence of the branching point~$u\in\Omega$.
\begin{definition}
\label{def:Tsigma-correlation}
Given a simply connected domain~$\Omega$ and points~$w,u\in\Omega$, we set
\[
\CorrO{T(w)\sigma(u)}~:=~\CorrO{\sigma(u)}\cdot {\textstyle\frac{1}{4}}\partial_w\fregU(w,z)\big|_{z=w}\,,
\]
where~$\CorrO{\sigma(u)}$ is given by~(\ref{eq_spin-explicit}) and in the right-hand side we assume that~$w$ is lifted onto~$\OmegaU$. Note the choice of~$w\in\OmegaU$ is irrelevant since~$\fU(w,z)=\fU(w^*,z^*)$.
\end{definition}

\begin{remark}
\label{rem:f_fdag_to_g}
It is not hard to see that for any fixed~$z\in\Omega\setminus\{u\}$ the following holds as~$a\to u$:
\[
\fU(a,z)\ \sim \ e^{i\frac{\pi}{4}}(a-u)^{-\frac{1}{2}}\cdot \spinor(z)\,,\qquad
\fstarU(a,z)\ \sim \ e^{-i\frac{\pi}{4}}(\overline{a}-\overline{u})^{-\frac{1}{2}}\cdot \spinor(z)\,.
\]
The easiest way to prove these asymptotics in the setup of our paper is to use the explicit formulae given above for~\mbox{$\Omega=\H$} and then to apply the covariance rules~(\ref{eq_covariance_f_fdag}),~(\ref{eq_covariance_spinor}) in order to treat the general case. More conceptually, one could derive these asymptotics by analyzing solutions to the boundary value problems from Definitions~\ref{def:feta} and~\ref{def:spinor}; see~\cite{CheHonIzy21} for details.
\end{remark}

\subsection{Schwarzian conformal covariance and singularities of correlation functions}
\label{sub:Schwarzian+OPE}

In this section we mention several CFT formulae for the two-point correlation functions involving the (holomorphic) stress-energy tensor~$T(w)$, in which the central charge~$c=\frac{1}{2}$ of the theory and the conformal dimensions~$\frac{1}{16}$ and~$\frac{1}{2}$ of the primary fields~$\varepsilon$ and~$\sigma$ appear. It should be said that all these claims are perfectly well-known in the CFT context. Nevertheless, let us emphasize that we do \emph{not} use any CFT concepts in our definitions of correlation functions in continuum and all these quantities are constructed starting with natural continuous counterparts of \emph{discrete} combinatorial observables discussed in Section~\ref{sec:via-fermions-disc}. From this perspective, the Schwarzian nature of the stress-energy tensor (see Proposition~\ref{prop:Schwarzian-covariance}) and the particular coefficients of singular terms in the asymptotics of correlation functions (see Proposition~\ref{prop:OPE}) appear as simple corollaries of the conformal covariance rules for the solutions to boundary value problems from Definitions~\ref{def:feta}~\ref{def:spinor} and~\ref{def:fetaU}.

\begin{proposition}
\label{prop:Schwarzian-covariance} Let~$\Omega$ be a simply connected domain,~$b_{},b'_{}\in\partial\Omega$ and~$a,u\in\Omega$. Then, each of the four functions
\[
\mathcal{F}_\Omega(w\,;\dots)~=~\CorrO{T(w)}\,,~\CorrO[\dob]{T(w)},~\frac{\CorrO{T(w)\varepsilon(a)}}{\CorrO{\varepsilon(a)}}\,,~ \frac{\CorrO{T(w)\sigma(u)}}{\CorrO{\sigma(u)}}
\]
is holomorphic in~$w\in\Omega$ (except at the points~$a$ and~$u$ for the third and the fourth ones, respectively) and obey the Schwarzian covariance
\begin{equation}
\label{eq_Schwarzian_covariance}
\mathcal{F}_\Omega(w\,;\dots)= \mathcal{F}_{\Omega'}(\varphi(w)\,;\dots)\cdot (\varphi'(w))^2+ {\textstyle\frac{1}{24}}[\schw\varphi](w)
\end{equation}
under conformal mappings~$\varphi:\Omega\to\Omega'$.
\end{proposition}

\begin{proof} The holomorphicity of all these functions directly follows from their definitions and the holomorphicity of the functions~$\f(a,\cdot)$~$\fstar(a,\cdot)$ and~$\fU(a,\cdot)$. The Schwarzian covariance of the first function~$\CorrO{T(w)}$ can be easily derived from its definition and the conformal covariance rule~(\ref{eq_covariance_f_fdag}), see also the explicit formula~(\ref{eq_T_explicit}). Also, exactly the same local computation leads to~(\ref{eq_Schwarzian_covariance}) for the fourth function~$\CorrO{T(w)\sigma(u)}/\CorrO{\sigma(u)}$. To prove~(\ref{eq_Schwarzian_covariance}) for the second and the third functions, it is enough to check that the second terms in their definitions (see Definition~\ref{def:T-Dobrushin-expectation} and Definition~\ref{def:T-and-eps-one-point}, respectively) are multiplied by~$(\varphi'(w))^2$ when applying a conformal mapping~$\varphi:\Omega\to\Omega'$. This is straightforward, see also~(\ref{eq_T_Dobrushin_explicit}).
\end{proof}

\begin{proposition} \label{prop:OPE} Let~$\Omega$ be a simply connected domain and~$w',a,u\in\Omega$. Then, one has
\begin{align}
 \label{eq_OPE_TT} \langle T(w)T(w')\rangle_\Omega^+ &=
 \frac{1}{4}\cdot\frac{1}{(w-w')^4}+\frac{2\langle T(w')\rangle_\Omega^+}{(w-w')^2}
 +\frac{\partial_{w'} \langle T(w')\rangle_\Omega^+}{w-w'} + O(1)\quad\text{as}\ \ w\to w',\\
 \label{eq_OPE_Te} \langle T(w) \varepsilon(a)\rangle_\Omega^+ &=
 \frac{1}{2}\cdot\frac{\langle \varepsilon(a)\rangle_\Omega^+}{(w-a)^2}
 +\frac{\partial_a \langle \varepsilon(a)\rangle_\Omega^+}{w-a}+ O(1)\quad \ \text{as}\ \ w\to a,\\
 \label{eq_OPE_Ts} \langle T(w) \sigma(u)\rangle_\Omega^+ &=
  \frac{1}{16}\cdot \frac{\langle \sigma(u)\rangle_\Omega^+}{(w-u)^{2}}
  +\frac{\partial_u\langle \sigma(u)\rangle_\Omega^+}{w-u}+ O(1)\quad \text{as}\ \ w\to u.
\end{align}
\end{proposition}

We first prove a simple lemma, which identifies the second coefficients in the expansions of the functions~$\f(w,z)$ and~$\fstar(a,z)$ near the singularity.

\begin{lemma} \label{lemma:second-coeff}
Let~$\Omega$ be a simply connected domain and~$w,a\in\Omega$. Then, the following asymptotic expansions hold as~$z\to w$ and~$z\to a$, respectively:
\begin{align*}
\tfrac{1}{4}\f (w,z) &= ~\tfrac{1}{2}(z-w)^{-1}\,+\, \langle T (w)\rangle^+_\Omega \cdot (z-w)
+ \tfrac{1}{2}\partial_w\CorrO{T (w)} \cdot (z-w)^2 + O((z-w)^3)\,,\\
\tfrac{i}{2}\fstar (a,z) & = \langle\varepsilon (a) \rangle^+_\Omega + \partial_a\langle\varepsilon (a) \rangle^+_\Omega \cdot (z-a) + O((z-a)^2)\,.
\end{align*}
\end{lemma}
\begin{proof}
Let us start with the second expansion. The definition of~$\CorrO{\varepsilon(a)}$ yields
\[
\tfrac{i}{2}\fstar(a,z)= \langle\varepsilon (a) \rangle^+_\Omega + c_1(a)\cdot (z-a) + O((z-a)^2)\,,\quad z\to a\,,
\]
with some (unknown) function~$c_1(a)$. Since the function~$\fstar(a,z)$ is anti-holomorphic in~$a$, one easily concludes that
\[
0=\tfrac{i}{2}\partial_a\fstar(a,z)\big|_{z=a}=\partial_a\CorrO{\varepsilon(a)} - c_1(a)\,.
\]
To derive the first expansion, note that~$\freg(w,w)=0$ since~$\f(w,z)= -\f(z,w)$ for all~$w,z$. Therefore, our definition of~$\CorrO{T(w)}$ yields
\[
\tfrac{1}{4}\freg(w,z) = \CorrO{T(w)} \cdot (z-w) + c_2(w)\cdot (z-w)^2 + O((z-w)^3)\,,\quad z\to w,
\]
for some~$c_2(w)$. Using the symmetry~$\f(w,z)=-\f(z,w)$ once again, one easily arrives at
\[
0=\tfrac{1}{4}\partial_w\partial_z\freg(w,z)\big|_{z=w} = \partial_w\CorrO{T(w)} - 2c_2(w)\,,
\]
which gives the result.
\end{proof}

\begin{proof}[Proof of Proposition~\ref{prop:OPE}] Recall that the two-point functions~$\CorrO{T(w)T(w')}$ and~$\CorrO{T(w)\varepsilon(a)}$ are given by Definition~\ref{def:TT-and-Te-correlations}. Thus, one can use the expansions provided by Lemma~\ref{lemma:second-coeff} to compute the singular parts of these functions as~$w'\to w$ and~$w\to a$, respectively. In particular, straightforward computations based on the asymptotics
\begin{align*}
&\tfrac{1}{4}\f(w,w')=\, \tfrac{1}{2}(w'\!-\!w)^{-1}\!+\CorrO{T(w)}\cdot(w'\!-\!w)+\tfrac{1}{2}\partial_{w}\CorrO{T(w)}\cdot(w'\!-\!w)^2+O((w'\!-\!w)^3), \\
&\tfrac{1}{4}\partial_{w'}\f(w,w') =\, -\tfrac{1}{2}(w'\!-\!w)^{-2}\!+\CorrO{T(w)}+\partial_{w}\CorrO{T(w)}\cdot(w'\!-\!w)+O((w'\!-\!w)^2), \\
&\tfrac{1}{4}\partial_{w}\f(w,w') =\,\tfrac{1}{2}(w'\!-\!w)^{-2}\! - \CorrO{T(w)} + O((w'\!-\!w)^2), \\
&\tfrac{1}{4}\partial_{w'}\partial_{w}\f(w,w') =\, -(w'\!-\!w)^{-3}\!+O(w'\!-\!w)
\end{align*}
lead to~(\ref{eq_OPE_TT}) and simpler computations give~(\ref{eq_OPE_Te}). To prove the last asymptotics~(\ref{eq_OPE_Ts}), note that Definition~\ref{def:Tsigma-correlation} and the Cauchy residue theorem imply the equality
\[
\frac{1}{2\pi i}\oint \frac{(z-u)^{\frac{1}{2}}\fU(w,z)dz}{(z-w)^2}\ = \ (w-u)^{\frac{1}{2}}\cdot\left[\frac{4\CorrO{T(w)\sigma(u)}}{\CorrO{\sigma(u)}}~-~\frac{1}{4(w-u)^2} \right],
\]
provided the (fixed) contour of integration surrounds both points~$w$ and~$u$. Passing to the limit~$w\to u$ and using Remark~\ref{rem:f_fdag_to_g} for all~$z$ on this contour, one concludes that
\[
\frac{e^{i\frac{\pi}{4}}}{2\pi i}\oint \frac{\spinor(z)dz}{(z-u)^{3/2}}\ =\ (w-u)\cdot\left[\frac{4\CorrO{T(w)\sigma(u)}}{\CorrO{\sigma(u)}}~-~\frac{1}{4(w-u)^2} \right],
\]
The formula~(\ref{eq_OPE_Ts}) easily follows since the left-hand side is equal to
\[
e^{i\frac{\pi}{4}}(z-u)^{-\frac{1}{2}}\spinorReg(z)\big|_{z=u}~=~2\mathcal{A}_\Omega(u)~=~4\partial_u\CorrO{\sigma(u)}\big/\CorrO{\sigma(u)}
\]
due to Definition~\ref{def:spinor}.
\end{proof}

We conclude this section by mentioning that one can use the language of solutions to boundary value problems for holomorphic functions and spinors discussed above to construct \emph{all} the correlations of the fields~$\sigma,\varepsilon$ and~$T$, as well as the anti-holomorphic counterpart~$\overline{T}$ of the stress-energy tensor, for any boundary conditions~$\b=\{b_1,\dots,b_{2m}\}$\,. If the spin field is not involved, this simply amounts to writing down a proper Pfaffian formula in the spirit of Definitions~\ref{def:T-Dobrushin-expectation},~\ref{def:TT-and-Te-correlations}, and Remarks~\ref{rem:T-Dobrusin},~\ref{rem:TT-and-Te-Pfaff}; cf.~\cite{Hon10} where the multi-point expectations~$\CorrO[\b]{\varepsilon(a_1)\dots\varepsilon(a_n)}$ were considered. In the presence of several spins~$\sigma(u_1),\dots,\sigma(u_k)$, the same techniques can be applied for correlations normalized by~$\CorrO{\sigma(u_1)\dots\sigma(u_k)}$ similarly to Definition~\ref{def:Tsigma-correlation}; cf.~\cite[Section~5]{IkhCar09} where the ratios~$\CorrO[\b]{\sigma(u_1)\dots\sigma(u_k))}/\CorrO{\sigma(u_1)\dots\sigma(u_k))}$ were discussed. Finally, the spin correlations~$\CorrO{\sigma(u_1)\dots\sigma(u_k)}$ themselves were treated in~\cite{CheHonIzy15}.

\setcounter{equation}{0}
\section{Convergence theorems}
\label{sec:Ising-convergence}

In this section we prove Theorems~\mbox{\ref{thm:Ising1}--\ref{thm:Ising3}}. Recall that, given a bounded simply connected domain~$\Omega\subset\C$, we denote by~$\Omega_\delta$ its lattice approximations on the regular honeycomb grid~$\C_\delta$ with the mesh size~$\delta$. For simplicity, one can define~$\Omega_\delta$ as the (largest connected component of the) union of all grid cells that are contained in~$\Omega$, together with appropriately chosen boundary half-edges, or just to assume that~$\Omega_\delta$ converges to~$\Omega$ in the Hausdorff sense as~$\delta\to 0$. It is worth mentioning that all results discussed below remain true if one assumes that discrete domains~$\Omega_\delta$ approximate~$\Omega$ in the so-called Carath\'eodory topology, which is a much weaker notion than the Hausdorff convergence, see~\cite[Chapter~1.4]{Pom92} and~\cite[Section~3.2]{CheSmi11} for definitions.

We use the same strategy as in the series of papers~\cite{HonSmi13,Hon10,CheIzy13,CheHonIzy15} devoted to convergence of various correlation functions in the Ising model to their scaling limits. Let~$w_\delta$ denote the face of the discrete domain~$\Omega_\delta$ containing~$w$,~$b_\delta$ denote the nearest to~$b\in\partial\Omega$ boundary edge of~$\Omega_\delta$ etc. As discussed in Section~\ref{sec:via-fermions-disc}, all the expectations~$\Edelta^+[T(w_\delta)]$, $\Edelta^{b_\delta^{\vphantom{\prime}},b'_\delta}[T(w_\delta)]$, $\Edelta^+[T(w_\delta)T(w'_\delta)]$ etc can be written in terms of discrete derivatives of the fermionic observables~$F_{\Omega_\delta}$ (or their spinor analogues). At the same time, their putative scaling limits~$\CorrO{T(w)}$, $\Corr{T(w)}_\Omega^{\dob}$, $\CorrO{T(w)T(w')}$ etc can be \emph{defined} in a similar manner using derivatives of the continuous functions~$\f$ and~$\fstar$ (or their spinor analogues), as discussed in Section~\ref{sec:CFT-correlations}. Therefore, to prove Theorems~\mbox{\ref{thm:Ising1}--\ref{thm:Ising3}} it is essentially enough to prove the $C^1$-convergence of (appropriately normalized) discrete fermionic observables~$F_{\Omega_\delta}$ to their continuous counterparts. The first result of this kind was obtained by the third author in~\cite{Smi10} (in a slightly different context of the random-cluster representation of the Ising model aka FK-Ising model). Later, more involved discrete complex analysis techniques treating such convergence problems were developed in~\cite{CheSmi11,CheSmi12,HonSmi13,Hon10,CheIzy13,CheHonIzy15,CheHonIzy21}. 
Below we do not need any new machinery and just provide relevant references.

\subsection{Convergence of discrete holomorphic observables} \label{sub:convergence-fermions}
In this section we state the needed convergence results (as~$\delta\to 0$) for the discrete holomorphic observables discussed in Sections~\ref{sec:combi} and~\ref{sec:via-fermions-disc} to holomorphic functions defined in Section~\ref{sec:CFT-correlations}.

Recall that, given a half-edge~$a_\delta$ of~$\Omega_\delta$, the complex-valued fermionic observable~$\Fdelta(a_\delta,\,\cdot\,)$ defined by~(\ref{eq_F_complex_def}) can be viewed as a solution to some discrete Riemann-type boundary value problem in~$\Omega^\delta$. In particular, this function is s-holomorphic everywhere in~$\Omega_\delta$ except at~$a_\delta$. Moreover, similarly to~\cite{HonSmi13} one can \emph{explicitly} construct a function~$F_{\C_\delta}(a_\delta,\,\cdot\,)$ such that,
for any discrete domain~$\Omega_\delta$, the difference
\begin{equation}
\label{eq:Fdag=}
\Fdelta^\dag(a_\delta,\,\cdot\,)~:=~\Fdelta(a_\delta,\,\cdot\,) - F_{\C_\delta}(a_\delta,\,\cdot\,)
\end{equation}
is s-holomorphic \emph{everywhere} in~$\Omega_\delta$, including the midedge~$a_\delta$.
The construction of~$F_{\C_\delta}(a_\delta,\,\cdot\,)$ is given in the Appendix, here
we only mention two of its properties leading to the explicit constants in Theorems~\mbox{\ref{thm:Ising1}--\ref{thm:Ising3}}. { Recall that~$\eta_{a_\delta}=\eta(\rho_{a_\delta})$ is defined by~\eqref{eq:eta=rho}.}
\begin{itemize}
\item \emph{Asymptotic behavior.} For midedges~$z_{e_\delta}$ such that~$|z_{e_\delta}-z_{a_\delta}|^{-1}=o(\delta^{-1})$, one has
\begin{equation}
\label{eq_full_plane_asymptotics}
\delta^{-1}F_{\C_\delta}(a_\delta,z_{e_\delta}) ~=~ \tfrac{\sqrt{3}}{2}\,{\overline{\eta}_{a_{\delta}}}{(z_{e_\delta}\!-z_{a_\delta})^{-1}}+o(1)\quad \text{as}\ \ \delta\to 0\,.
\end{equation}
\item \emph{Local values.} One has
\begin{align}
\label{eq_local_values_zero}
F_{\C_\delta}(a_\delta,z_{a_\delta})~& =~({-i\eta_{a_\delta}})\cdot (-{\textstyle\frac{1}{6}})\,, \\
\label{eq_local_values} F_{\C_\delta}(a_\delta,z_{a_\delta}^\pm)~& =~({-i\eta_{a_\delta}})\cdot [{\textstyle(1-\frac{2\sqrt{3}}{\pi})}~\pm~i(\textstyle\frac{3}{\pi}-\frac{2}{\sqrt{3}})]\,,
\end{align}
\end{itemize}
where~$z_{a_\delta}^\pm:=z_{a_\delta}\pm\sqrt{3}\delta\cdot i{\rho_{a_\delta}}$ denote the midedges of the two neighboring to~$a_\delta$ edges of~$\C_\delta$, which are oriented in the same direction~$\rho_{a_\delta}$ as the (half-)edge~$a_\delta$.

The following convergence theorem for discrete fermionic observables~$\Fdelta(a_\delta,\,\cdot\,)$ is a straightforward analogue of~\cite[Theorem~1.8]{HonSmi13}. In particular, to obtain this result, one can simply repeat the proof given in~\cite{HonSmi13} with minor modifications needed when dealing with the honeycomb grid instead of the square one. Alternatively, one can use a more robust scheme of proving convergence results of this kind, which was proposed in~\cite[Section~3.4]{CheHonIzy15}. Note that the latter approach does not require any regularity assumption on the boundary of~$\Omega^\delta$.

\begin{theorem}
\label{thm:convergence-fermions}
Given a bounded simply connected domain~$\Omega\subset\C$, two points~$a,z\in\Omega$ and~$\rho\in\wp^\pm$, let~$\Omega_\delta$ denote a discrete approximation of~$\Omega$ on the regular honeycomb grid~$\C_\delta$, $a_\delta$~be the closest to~$a\in\Omega$ half-edge of~$\Omega_\delta$ oriented in the direction~$\rho$, and~$z_{e_\delta}$ be the closet to~$z$ midedge of~$\Omega_\delta$. Then, the following convergence results hold as~$\delta\to 0$:
\begin{align*}
\delta^{-1}\cdot\Fdelta(a_\delta,z_{e_\delta})&\ \rightrightarrows\ \tfrac{\sqrt{3}}{2\pi}\feta(a,z),\qquad \quad z\ne a, \\
\delta^{-1}\cdot\Fdelta^\dag(a_\delta,z_{e_\delta})&\ \rightrightarrows\ {\tfrac{\sqrt{3}}{4\pi}}(\overline{\eta}\freg(a,z)+\eta\fstar(a,z)),
\end{align*}
where~$\eta=\eta(\rho)$, the function~$\feta(a,z)$ is given by Definition~\ref{def:feta}, and~$\freg(a,z)$, $\fstar(a,z)$ are defined in Proposition~\ref{propContFermDecomp}. Moreover, the first convergence is uniform (with respect to~$(a,z)$) on compact subsets of~$(\Omega\times\Omega)\setminus D$, where~$D:=\{(a,a),a\in\Omega\}$, and the second is uniform on compact subsets of~$\Omega\times\Omega$, in particular it holds near the diagonal~$D\subset\Omega\times\Omega$.
\end{theorem}

\begin{proof}
	See~\cite[Theorem~1.8]{HonSmi13} or~\cite{CheHonIzy21}.
\end{proof}

The following corollary is a straightforward analogue of the main result of~\cite{HonSmi13} for discrete domains drawn on the refining honeycomb grids~$\C_\delta$ (instead of square ones considered in~\cite{HonSmi13}).

\begin{corollary}
\label{cor:energy-density-convergence} Given a bounded simply connected domain~$\Omega\subset\C$ and a point~$a\in\Omega$, let~$\Omega_\delta$ denote a discrete approximation of~$\Omega$ on the regular honeycomb grid~$\C_\delta$, and $a_\delta$~be the closest to~$a\in\Omega$ half-edge of~$\Omega_\delta$\,. Then, the following convergence holds as~$\delta\to 0$:
\[
\delta^{-1}\Edelta^+[\varepsilon(a_\delta)]\ \mathop{\to}_{\delta\to 0}\ \tfrac{\sqrt{3}}{\pi}\CorrO{\varepsilon(a)}\,,
\]
where the energy density~$\varepsilon(a_\delta)$ and the funciton~$\CorrO{\varepsilon(a)}$ are given by~(\ref{eq_e_definition}) and Definition~\ref{def:T-and-eps-one-point}.
\end{corollary}

\begin{proof}
Since~$\varepsilon(\overline{a}_\delta)=\varepsilon(a_\delta)$ does not depend on the orientation of~$a_\delta$, it is enough to consider the case~$\eta_{a_\delta}=\eta\in\wp^+$. Recall that the following identities hold (see~(\ref{eq_e_via_fermion}) and~(\ref{eq_F_real_via_complex})):
\begin{align*}
\Edelta^+[\varepsilon(a_\delta)] ~=~ {\textstyle{\frac{1}{3}}}-2\F(a_\delta,\overline{a}_\delta)~=&~ \tfrac{1}{3}-2\Re[{\overline{\eta}_{\overline{a}_\delta}}\F(a_\delta,z_{a_\delta})] \\
~=&~ \tfrac{1}{3}-2\Im[{\overline{\eta}_{a_\delta}}\F(a_\delta,z_{a_\delta})] ~=~ -2\Im[{\overline{\eta}_{a_\delta}}\F^\dag(a_\delta,z_{a_\delta})]\,,
\end{align*}
where we also used~(\ref{eq_local_values_zero}). Theorem~\ref{thm:convergence-fermions} now implies
\[
\delta^{-1}\Edelta^+[\varepsilon(a_\delta)] ~\mathop{\to}_{\delta\to 0}~ {-\tfrac{\sqrt{3}}{2\pi}\Im}[{\overline{\eta}^2}\freg(a,a)+\fstar(a,a)]~=~\tfrac{\sqrt{3}}{\pi}\CorrO{\varepsilon(a)}\,.
\]
since~$\freg(a,a)=-\freg(a,a)=0$ and~$\fstar(a,a)=-{2i}\CorrO{\varepsilon(a)}$.
\end{proof}

We also need an analogue of Theorem~\ref{thm:convergence-fermions} for the branching observables~$\FdeltaU(a_\delta,\cdot)$ discussed in Section~\ref{sub:T-spin-via-spinors}. Similarly to~(\ref{eq:Fdag=}), let
\begin{equation}
\label{eq:FUdag=}
\FdeltaU^\dag(a_\delta,\cdot)~:=~\FdeltaU(a_\delta,\cdot)-F_{\C_\delta}(a_\delta,\cdot)
\end{equation}
As in the non-branching case, the difference~(\ref{eq:FUdag=}) is s-holomorphic everywhere \emph{near}~$a_\delta$. Note that, contrary to~(\ref{eq:Fdag=}), this function is defined only in a vicinity of~$a_\delta$ since the full-plane observable~$F_{\C_\delta}$ does not respect the spinor property of~$\FdeltaU$\,.

\begin{theorem}
\label{thm:convergence-spinors}
Given a bounded simply connected domain~$\Omega\subset\C$, two distinct points~$a,u\in\Omega$ and~$\rho\in\wp^\pm$, let~$\Omega_\delta$ denote a discrete approximation to~$\Omega$ (converging in the Carath\'eodory sense as~$\delta\to 0$, see~\cite[Section~3.2]{CheSmi11} for definitions), $u_\delta$~be the face of~$\Omega_\delta$ containing~$u$, and~$a_\delta$ be the closest to~$a\in\Omega$ half-edge of~$\Omega_\delta$ oriented in the direction~$\rho$, which we assume to be lifted on the discrete double-cover~$[\Omega_\delta,u_\delta]$. For a point~$z\in\OmegaU$, let~$z_{e_\delta}$ denote the closest to~$z$ midedge of~$[\Omega_\delta,u_\delta]$. Then, the following convergence results hold as~$\delta\to 0$:
\begin{align}
\label{eq_spinor_convergence}
\delta^{-1}\FdeltaU(a_\delta,z_{e_\delta})&\ \rightrightarrows\ \tfrac{\sqrt{3}}{2\pi}\fetaU(a,z),\quad\qquad z\ne a,a^*, \\
\delta^{-1}\FdeltaU^\dag(a_\delta,z_{e_\delta})&\ \rightrightarrows\ {\tfrac{\sqrt{3}}{4\pi}}({\overline{\eta}}\fregU(a,z)+{\eta}\fstarU(a,z)),
\label{eq_spinor_reg_convergence}
\end{align}
where~$\eta=\eta(\rho)$ and~$\feta(a,z)$, $\freg(a,z)$, $\fstar(a,z)$ are defined in Section~\ref{sub:T-spin-correlations}. Moreover, the first convergence is uniform (with respect to~$(a,z)$)
on compact subsets of~$(\OmegaU\times\OmegaU)\setminus D$, where~$D:=\{(a,a),a\in\OmegaU\}\cup\{(a,a^*)\in\OmegaU\}\cup\{u\}$, and the second is uniform for $a\in\OmegaU$ lying in compact subsets of~$\OmegaU\setminus\{u\}$ and~$z$ in a vicinity of~$a$.
\end{theorem}
\begin{proof}
	See \cite[Theorem~3.16]{CheHonIzy21}.
\end{proof}

The next lemma states that the convergence of discrete derivatives of the fermionic observables follows from the convergence of these functions themselves, as usual when working with discrete holomorphic or discrete harmonic functions on regular lattices.

\begin{lemma}
\label{lemma:conv-derivatives}
Let $f:\LambdaU\to\C$ be a holomorphic function defined in a planar domain~$\LambdaU$, $\rho\in\wp$, and~$F_\delta:\LambdaU_\delta\to \C$ be a family of discrete s-holomorphic functions such that~$\delta^{-1}F_\delta (z_{e_\delta}) \rightrightarrows f(z)$ as~$\delta\to 0$ uniformly on compact subsets of~$\LambdaU$, where~$e_\delta$ denotes the closest to~$z\in\LambdaU$ edge of~$\LambdaU_\delta$ { oriented in the direction~$\rho$}. Then, one has
\[
(\sqrt{3}\delta)^{-1}\cdot (\delta^{-1}F_\delta(z_{e_\delta}^\pm)-\delta^{-1}F_\delta(z_{e_\delta}))~ \rightrightarrows~ \pm i{\rho}\cdot \partial f(z)\quad\text{as}\ \ \delta\to 0
\]
uniformly on compact subsets of~$\LambdaU$, where~$z_{e_\delta}^\pm:=z_{e_\delta}\pm\sqrt{3}\delta\cdot i{\rho}$ are the midpoints of the two neighboring to~$e_\delta$ edges oriented in the same direction~$\rho$.
\end{lemma}

\begin{proof} As discussed in {Section~\ref{sub:s-holomorphicity}}, one can view s-holomorphic functions~$F_\delta$ as collections of six discrete harmonic components~$F_\delta^\mu$, each of them being defined on the corners of~$\LambdaU_\delta$ of type~$\mu\in\wp$, which form a piece of a regular triangular lattice with the mesh size~$\sqrt{3}\delta$. It is easy to see that the uniform convergence of the functions~$\delta^{-1}F_\delta$ implies the uniform convergence (to~$\mathrm{Proj}[f(z);\mu\R]$) of each of the components~$\delta^{-1}F_\delta^\mu$. Such a convergence of discrete harmonic functions always yields the uniform convergence of their discrete derivatives (e.g., see~\cite[Proposition~3.3]{CheSmi11}, where a similar result is proved for harmonic functions on general isoradial graphs). Finally, as the discrete derivatives of~$\delta^{-1}F_\delta$ can be easily reconstructed from the discrete derivatives of their components~$\delta^{-1}F_\delta^\mu$, one easily sees that the uniform convergence of the latter implies the uniform convergence of the former.
\end{proof}

\begin{remark}
\label{rem:conv-on-boundary}
To prove the convergence of the stress-energy tensor expectations for Dobrushin boundary conditions, one also needs an extension of Theorem~\ref{thm:convergence-fermions} to the case when one or both points~$a,z$ lie on straight parts of the boundary~$\partial\Omega$, oriented in one of the directions~$ i\rho$, orthogonal to the directions~$\rho\in\wp$ of lattice edges. Note that the functions~$\feta(a,\cdot)$ are well-defined on such parts of~$\partial\Omega$ due to~(\ref{eq_feta_boundary_conditions}) and the Schwarz reflection principle. This passage from the convergence in the bulk of~$\Omega$ to the convergence on (the straight part of) the boundary~$\partial\Omega$ can be found in~\cite[Section~5.2]{CheSmi12}, where a similar result was proved with no restrictions on the local direction of~$\partial\Omega$ (and for general isoradial graphs instead of regular grids). The distilled version of this argument dealing with square lattices and particular orientations of~$\partial\Omega$ (in the directions orthogonal to lattice edges) is given in~\cite[Lemma~4.8]{CheIzy13}. It can be easily adapted to s-holomorphic functions defined on the refining honeycomb grids~$\C_\delta$.
\end{remark}

\subsection{Proof of Theorem~\ref{thm:Ising1}} \label{sub:convergence-theorem1}
We are now in the position to derive Theorem~\ref{thm:Ising1} from the convergence of basic fermionic observables and their derivatives.

\begin{proof}[Proof of Theorem~\ref{thm:Ising1} for~`$+$' boundary conditions] Recall that
\[
\mathcal{T}^{[\rho]}(w_\delta)=\Trho(w_\delta)+\Tmrho(w_\delta)+\Rrho(w_\delta)
\]
and the expectations of each of these three terms are expressed via the fermionic observables in Proposition~\ref{prop_via_fermion_t}. In particular,
\[
\Edelta^+[\Trho(w_\delta)] \ =\ \CmT-{\textstyle\frac{\sqrt{3}}{2}}[\sOne\dTwo \Fdelta](\wrho_\delta,\wrho_\delta),
\]
where the discrete derivative~$\dTwo$ with respect to the second argument and the mean value~$\sOne$ with respect to the first are introduced in Definition~\ref{def:discrete-derivatives}.
Let~$\eta=\eta(\rho)$ be defined by~(\ref{eq:eta=rho}).
Due to~(\ref{eq_F_real_via_complex}), for both~$a_\delta\in\{\wrho_{\delta,\mathrm{up}}\,,\wrho_{\delta,\mathrm{down}}\}$, we have
\[
\tfrac{\sqrt{3}}{2}[\dTwo\Fdelta](a_\delta,\wrho_\delta)~=~ {\tfrac{1}{2}\Re}\left[\overline{\eta}\cdot\big(\Fdelta(a_\delta,z(\wrho_{\delta,\mathrm{up}}))- \Fdelta(a_\delta,z(\wrho_{\delta,\mathrm{down}}))\big)\right],
\]
where we use the notation~$z(e)$ instead of~$z_e$.
Moreover,~(\ref{eq_local_values}) implies
\[
\qquad \qquad \qquad\quad \CmT~=~\tfrac{1}{2}\Re\left[\overline{\eta}\cdot\big( F_{\C_\delta}(a_\delta,z(\wrho_{\delta,\mathrm{up}}))-F_{\C_\delta}(a_\delta,z(\wrho_{\delta,\mathrm{down}}))\big)\right],
\]
thus we arrive at the equation
\[
\textstyle \Edelta^+[\Trho(w_\delta)] \ =\ -{\frac{1}{4}}\sum\nolimits_{a_\delta\in\{\wrho_{\delta,\mathrm{up}}\,,\wrho_{\delta,\mathrm{down}}\}} \Re\left[\overline{\eta}\cdot\big(\Fdelta^\dag(a_\delta,z(\wrho_{\delta,\mathrm{up}}))- \Fdelta^\dag(a_\delta,z(\wrho_{\delta,\mathrm{down}}))\big)\right].
\]
For both~$a_\delta\in\{\wrho_{\delta,\mathrm{up}}\,,\wrho_{\delta,\mathrm{down}}\}$\,, it easily follows from Theorem~\ref{thm:convergence-fermions} and Lemma~\ref{lemma:conv-derivatives} that
\[
\delta^{-2}\left[\Fdelta^\dag(a_\delta,z(\weta_{\delta,\mathrm{up}}))- \Fdelta^\dag(a_\delta,z(\weta_{\delta,\mathrm{down}}))\right] \ \mathop{\to}_{\delta\to 0}\ {\tfrac{3}{4\pi}(i\rho)\cdot}\partial_{z}(\overline{\eta}\freg(w,z)+\eta\fstar(w,z))\big|_{z=w}\,.
\]
This leads to the convergence
\[
\delta^{-2}\Edelta^+[\Trho(w_\delta)]~\mathop{\to}_{\delta\to 0}~ {\tfrac{3}{8\pi}\Im\big[\rho}\cdot(\overline{\eta}^2\partial_z\freg(w,z)+\partial_z\fstar(w,z))\big]\big|_{z=w}\,.
\]
Since~$\overline{\eta}^2=-i\rho$ due to~\eqref{eq:eta=rho} and~$\partial_z\freg(w,z)\big|_{z=w}=4\CorrO{T(w)}$ (see Definition~\ref{def:T-and-eps-one-point}), we obtain
\begin{equation}\label{eq_TetaTieta_conv}
\delta^{-2}\Edelta^+[\Trho(w_\delta)+\Tmrho(w_\delta)]\ \mathop{\to}_{\delta\to 0}\  
-\tfrac{3}{\pi}\Re[\,{\rho^2}\CorrO{T(w)}\,]\,.
\end{equation}

Thus, it remains to prove that the term~$\delta^{-2}\E^+[\Rrho(w)]$ disappears in the limit~$\delta\to 0$. Recall that~$\Rrho(w)=\Rmrho(w)$ and thus we can assume~$\rho\in\wp^+$. Proposition~\ref{prop_via_fermion_t} gives
\[
\E^+[\Rrho(w_\delta)] ~=~ \CmR+{\sqrt{3}}\,[\dOne\dTwo\Fdelta](\wrho_\delta,\wmrho_\delta).
\]
For~$a_\delta\in\{\wrho_{\delta,\mathrm{up}}\,,\wrho_{\delta,\mathrm{down}}\}$\,, it easily follows from~(\ref{eq_local_values}) and~(\ref{eq_F_real_via_complex}) that
\[
[\dTwo F_{\C_\delta}](\wrho_{\delta,\star}\,,\wmrho_\delta)~=~\Im\big[\overline{\eta}\cdot [\dTwo F_{\C_\delta}](\wrho_{\delta,\star}\,,z(\wmrho_\delta))\big]~=~\mp\tfrac{1}{2}\CmR\,,
\]
where~$\star\in\{\text{`up',`down'}\}$ and the~`$\mp$' sign is~`$-$' if~$\star=\text{`up'}$ and `$+$' if $\star=\text{`down'}$. Therefore,
\[
\E^+[\Rrho(w_\delta)] ~=~ \sqrt{3}\,[\dOne\dTwo\Fdelta^\dag](\wrho_\delta,\wmrho_\delta)\,.
\]
Using Theorem~\ref{thm:convergence-fermions} and Lemma~\ref{lemma:conv-derivatives}, it is easy to see that
\begin{align*}
\delta^{-2}[\dTwo\Fdelta^\dag](a_\delta^{\vphantom{\rho}},\wmrho_\delta)~&=~ \Im\big[\overline{\eta}\cdot\delta^{-2}[\dTwo\Fdelta^\dag](a_\delta^{\vphantom{\rho}},z(\wmrho_\delta))\big] \\ &=~
 \tfrac{\sqrt{3}}{4\pi}\Im\big[(-i\rho\overline{\eta})\cdot \partial_w(\overline{\mu}\freg(a,w)+\mu \fstar(a,w))\big]+o_{\delta\to 0}(1)\\
 &=~\tfrac{\sqrt{3}}{4\pi}\Re\big[\overline{\mu}\cdot\rho\overline{\eta}\partial_w\freg(a,w)\big]+\tfrac{\sqrt{3}}{4\pi}\Re\big[\mu\cdot \rho\overline{\eta}\partial_w\fstar(a,w))\big]+o_{\delta\to 0}(1)
\end{align*}
uniformly for~$a_\delta$ in a (macroscopic) vicinity of~$w_\delta$, where~$\mu:=\eta_{a_\delta}$. Let~$a_\delta$ be oriented in one of the directions~$\wp^+$ so that we have~$\eta_{\overline{a}_\delta}=i\mu$. Using the antisymmetry of fermionic observables and~(\ref{eq_F_complex_def}), one can write
\[
-[\dOne\Fdelta^\dag](\wmrho_\delta,z(a_\delta^{\vphantom{\rho}}))~=~\mu\cdot [\dTwo\Fdelta^\dag](a_\delta^{\vphantom{\rho}},\wmrho_\delta)+ i\mu\cdot [\dTwo\Fdelta^\dag](\overline{a}_\delta^{\vphantom{\rho}},\wmrho_\delta)
\]
and then obtain the (uniform for~$a_\delta$ in a macroscopic vicinity of~$w_\delta$) convergence
\[
-\delta^{-2}[\dOne\Fdelta^\dag](\wmrho_\delta,z(a_\delta^{\vphantom{\rho}}))~\mathop{\to}_{\delta\to 0}~\tfrac{\sqrt{3}}{4\pi}\big[\rho\overline{\eta}\partial_w\freg(a,w)+\overline{\rho\overline{\eta}\partial_w\fstar(a,w)}\,\big]
\]
Using~(\ref{eq_F_real_via_complex}) and Lemma~\ref{lemma:conv-derivatives} once again, we get
\begin{align*}
-\delta^{-3}[\dOne\dTwo\Fdelta^\dag](\wmrho_\delta,\wrho_\delta)~&=~ -\delta^{-3}\Re\big[\overline{\eta}\cdot[\dOne\dTwo\Fdelta^\dag](\wmrho_\delta,z(\wrho_\delta))\big]\\
&\mathop{\to}_{\delta\to 0}~\tfrac{\sqrt{3}}{4\pi}\Re\big[(i\rho)\overline{\eta}\cdot\partial_z(\rho\overline{\eta}\partial_w\freg(z,w)+ \overline{\rho}\eta\overline{\partial_w\fstar(z,w)})\big|_{z=w}\,\big]
\end{align*}
Due to the antisymmetry~(\ref{eq_symmetries_f}) we have~$\partial_z\partial_w\freg(z,w)|_{z=w}=0$ and
\begin{align}
\notag
\delta^{-3}\E^+[\Rrho(w_\delta)] ~&=~ -\delta^{-3}\sqrt{3}\,[\dOne\dTwo\Fdelta^\dag](\wmrho_\delta,\wrho_\delta)~\\
&\mathop{\to}_{\delta\to 0}~\tfrac{3}{4\pi}\Im[\,\overline{\partial}_w\partial_z\fstar(w,z)\,]\big|_{z=w}~=~\tfrac{3}{8\pi}\Delta_w\CorrO{\varepsilon(w)}\,,
\label{eq_Reta_conv}
\end{align}
where the last identity follows from Lemma~\ref{lemma:second-coeff}.
In particular, $\delta^{-2}\E^+[\Rrho(w_\delta)]=O(\delta)$.
\end{proof}

\begin{proof}[Proof of Theorem~\ref{thm:Ising1} for Dobrushin boundary conditions] According to Proposition~\ref{prop_via_fermion_tdob}, it is now enough to prove the convergence of quantities
\[
\delta^{-1}[\sTwo\Fdelta](b_\delta^{\vphantom{\rho}},\wrho_\delta),\quad \delta^{-1}\Fdelta(b^{\vphantom{\prime}}_\delta,b'_\delta)\quad\text{and}\quad
\delta^{-2}[\dTwo\Fdelta](b_\delta^{\vphantom{\rho}},\wrho_\delta)
\]
(and similarly for~$b'_\delta$, recall that~$\Fdelta(b'_\delta,b^{\vphantom{\prime}}_\delta)=-\Fdelta(b^{\vphantom{\prime}}_\delta,b'_\delta)$) as~$\delta\to 0$. As mentioned in Remark~\ref{rem:conv-on-boundary}, these results can be derived from an extension of the convergence provided by Theorem~\ref{thm:convergence-fermions} to straight parts of~$\partial\Omega$. Indeed, due to Remark~\ref{rem:bvp_in_discrete} one has
\[
\Fdelta(b_\delta^{\vphantom{\rho}},\wrho_{\delta,\star})~=~ -\Fdelta(\wrho_{\delta,\star},b_\delta^{\vphantom{\rho}})~=~-\overline{\beta}\cdot \Fdelta(\wrho_{\delta,\star},z(b_\delta^{\vphantom{\rho}}))~\in~\R, \quad \text{where}\quad \beta:=\eta_{\overline{b}}
\]
for both~$\star\in\{\text{`up',`down'}\}$ (note that~$\beta=\overline{\tau(b)}{}^{1/2}$ is defined by the direction~$\tau(b)$ of the straight part of~$\partial\Omega$ at the point~$b$ and hence does not depend on~$\delta$). Therefore, the extension of Theorem~\ref{thm:convergence-fermions} to straight boundaries (e.g., see~\cite[Lemma~4.8]{CheIzy13}) gives
\[
\delta^{-1}\Fdelta(b_\delta^{\vphantom{\rho}},\wrho_{\delta,\star})~\mathop{\to}_{\delta\to 0}~-\overline{\beta}\cdot \tfrac{\sqrt{3}}{2\pi}\feta(w,b)
\]
In particular, averaging over~$\star\in\{\text{`up',`down'}\}$ and using~(\ref{eq_f_fdag_boundary_conditions}), one gets
\begin{equation}
\delta^{-1}[\sOne\Fdelta](b_\delta^{\vphantom{\rho}},\wrho_{\delta,\star})~\mathop{\to}_{\delta\to 0}~
  -\overline{\beta}\cdot \tfrac{\sqrt{3}}{2\pi}(\overline{\eta}\f(w,b)+\eta\fstar(w,b))~=~\tfrac{\sqrt{3}}{2\pi}\Re[\,\overline{\eta}\overline{\beta}\f(b,w)].
\label{eq_x_conv_Tdob-one}
\end{equation}
Using a similar convergence result for~$\wmrho_{\delta,\star}$, 
we obtain
\begin{align}
\delta^{-1}\Fdelta(b_\delta,z(\wrho_{\delta,\star})) ~&=~ \delta^{-1}\left[\eta\Fdelta(b_\delta,\wrho_{\delta,\star})\pm i\eta\Fdelta(b_\delta,\wmrho_{\delta,\star})\right] \notag \\
&\mathop{\to}_{\delta\to 0}~ 
-\overline{\beta}\cdot\tfrac{\sqrt{3}}{2\pi}\left[\eta\feta(w,b)\pm i\eta f_\Omega^{[\pm i\eta]}(w,b)\right] ~=~
\overline{\beta}\cdot\tfrac{\sqrt{3}}{2\pi}\f(b,w)\,,
\label{eq_x_conv_Tdob}
\end{align}
uniformly on compact subsets of~$\Omega$. In particular, another extension of this convergence to the boundary point~$b'$ via~\cite[Lemma~4.8]{CheIzy13} yields
\begin{equation}
\label{eq_x_conv_Tdob-zero}
\delta^{-1}\Fdelta(b_\delta^{\vphantom{\prime}},b'_\delta)~=~\overline{\beta}'\cdot \delta^{-1}\Fdelta(b_\delta^{\vphantom{\prime}},z(b'_\delta))
~\mathop{\to}_{\delta\to 0}~ \overline{\beta'\beta}\cdot\tfrac{\sqrt{3}}{2\pi}\f(b,b')~\in~\R\,,
\end{equation}
where~$\beta':=\eta_{\overline{b}'}$ does not depend on~$\delta$. Further, applying Lemma~\ref{lemma:conv-derivatives} to~(\ref{eq_x_conv_Tdob}) we obtain
\[
\delta^{-2}\cdot\tfrac{1}{\sqrt{3}}\left[\Fdelta(b_\delta,z(\wrho_{\delta,\mathrm{up}})) - \Fdelta(b_\delta,z(\wrho_{\delta,\mathrm{down}}))\right]\
\mathop{\to}_{\delta\to 0}\ i\rho\overline{\beta}\cdot\tfrac{\sqrt{3}}{2\pi}\partial_w\f(b,w)\,.
\]
Now one can use~(\ref{eq_F_real_via_complex}) once again to conclude that
\begin{equation}
\label{eq_x_conv_Tdob-two}
\delta^{-2}[\dTwo\Fdelta](b_\delta^{\vphantom{\rho}},\wrho_\delta)\
\mathop{\to}_{\delta\to 0}\ \tfrac{\sqrt{3}}{2\pi}\Re\left[\overline{\eta}\cdot i\rho\overline{\beta}\partial_w\f(b,w)\right]~=~ {-\tfrac{\sqrt{3}}{2\pi}}\Re\left[\overline{\eta}^3\cdot \overline{\beta}\partial_w\f(b,w)\right]
\end{equation}
since~$\eta^2=i\overline{\rho}$ due to~(\ref{eq:eta=rho}). A straightforward substitution of~(\ref{eq_x_conv_Tdob-one}),~(\ref{eq_x_conv_Tdob-zero}) and~(\ref{eq_x_conv_Tdob-two}) into the formula for~$\E^\dob[\Trho(w)]$ provided by Proposition~\ref{prop_via_fermion_tdob} together with the identity
\begin{equation}
\label{eq_x_eta_eta3=eta4}
\Re[\overline{\eta}^3 A]\Re[\overline{\eta} B] + \Re[(\overline{\pm i\eta})^3 A]\Re[(\overline{\pm i\eta})B]~=~\Re[\overline{\eta}^4AB]~=~{-}\Re[\rho^2AB]
\end{equation}
lead to
\begin{align*}
\delta^{-2}\Big[ \Edelta^{b^{\vphantom{\prime}}_\delta,b'_\delta}&\big[\Trho(w_\delta)+\Tmrho(w_\delta)\big] - \Edelta^+\big[\Trho(w_\delta)+\Tmrho(w_\delta)\big]\Big]\\  \mathop{\to}_{\delta\to 0}\ &\ \frac{3}{4\pi}\cdot \frac{\Re[\,\rho^2\cdot
(-\overline{\beta}\partial_w\f(b,w)\cdot \overline{\beta'}\f(b',w)+\overline{\beta}\f(b,w)\cdot \overline{\beta'}\partial_w\f(b',w))\,]}{\overline{\beta'\beta}\f(b,b')}\,.
\end{align*}
Taking into account Definition~\ref{def:T-Dobrushin-expectation} and the convergence~(\ref{eq_TetaTieta_conv}), we conclude that
\[
\delta^{-2}\Edelta^{b^{\vphantom{\prime}}_\delta,b'_\delta}\big[\Trho(w_\delta)+\Tmrho(w_\delta)\big]\ \mathop{\to}_{\delta\to 0}\ -\tfrac{3}{\pi}\Re\big[\,\rho^2\Corr{T(w)}_\Omega^\dob\,\big]\,.
\]
Finally, note that~(\ref{eq_x_conv_Tdob-zero}) and~(\ref{eq_x_conv_Tdob-two}) imply
\[
\delta^{-2}\Big[ \Edelta^{b^{\vphantom{\prime}}_\delta,b'_\delta}\big[\Rrho(w_\delta)\big]-\Edelta^+\big[\Rrho(w_\delta)\big]\Big]=O(\delta)\quad\text{as}\ \ \delta\to 0
\]
due to Proposition~\ref{prop_via_fermion_tdob} and since~$\f(b,b')\ne 0$. Together with~(\ref{eq_Reta_conv}), this guarantees that the expectation of the term~$\Rrho(w_\delta)$ is of order~$\delta^3$ and hence does not contribute to the limit.
\end{proof}

\subsection{Proofs of Theorems~\ref{thm:Ising2} and~\ref{thm:Ising3}} \label{sub:convergence-theorem1}

These proofs go exactly along the same lines as the proof of Theorem~\ref{thm:Ising1} given in the previous section.

\begin{proof}[Proof of Theorem~\ref{thm:Ising2}]
Using Proposition~\ref{prop_via_ferm_tt}, we see that it is essentially enough to compute the limits of the following quantities:
\[
\delta^{-1}\!\cdot[\sOne\sTwo\Fdelta](\wrho_\delta,\wptau_\delta),\quad \delta^{-2}\!\cdot[\dOne\sTwo\Fdelta](\wrho_\delta,\wptau_\delta)\quad \text{and}\quad \delta^{-3}\!\cdot[\dOne\dTwo\Fdelta](\wrho_\delta,\wptau_\delta)
\]
as~$\delta\to 0$. As above, let~$\wrho_{\delta,\star}$ (and similarly~$\wptau_{\delta,\star}$) denote one of the two half-edges~$\weta_{\delta,\mathrm{up}}\,,\weta_{\delta,\mathrm{down}}$\,. Let~$\eta:=\eta(\rho)$ and~$\mu:=\eta(\tau)$, where the function~$\eta(\rho)$ is given by~(\ref{eq:eta=rho}).
Due to~(\ref{eq_F_real_via_complex}),
we have
\[
\Re[\,\overline{\mu}\Fdelta(\wrho_{\delta,\star}\,,z(\wptau_{\delta,\star}))] = \Fdelta(\wrho_{\delta,\star}\,,\wptau_{\delta,\star})  =  -\Fdelta(\wptau_{\delta,\star}\,,\wrho_{\delta,\star}) = -\Re[\,\overline{\eta}\Fdelta(\wptau_{\delta,\star}\,,z(\wrho_{\delta,\star}))].
\]
Therefore, it follows from Theorem~\ref{thm:convergence-fermions} that
\begin{equation}
\label{eq_x_conv_TT_ss}
\delta^{-1}\!\cdot\Fdelta(\wrho_{\delta,\star}\,,\wptau_{\delta,\star}\,)\ \mathop{\to}_{\delta\to 0}\ \tfrac{\sqrt{3}}{2\pi}\Re[\,\overline{\mu}\feta(w,w')\,] ~=~ -\tfrac{\sqrt{3}}{2\pi}\Re[\,\overline{\eta}\fmu(w',w)\,]\,.
\end{equation}
Moreover, 
Lemma~\ref{lemma:conv-derivatives} implies also the convergence of the discrete derivatives
\begin{align}
\notag 
\delta^{-2}\!\cdot[\dTwo\Fdelta](\wrho_{\delta,\star}\,,\wptau_\delta)\ &\mathop{\to}_{\delta\to 0}\ \tfrac{\sqrt{3}}{2\pi}\Re[\,\overline{\mu}\cdot (i\tau)\partial_{w'}\feta(w,w')\,] \\
&=~-\tfrac{\sqrt{3}}{2\pi}\Re[\,\overline{\mu}^3\partial_{w'}\feta(w,w')\,]\,,
\label{eq_x_conv_TT_sd}\\
\delta^{-2}\!\cdot[\dOne\Fdelta](\wrho_\delta,\wptau_{\delta,\star})\ &\mathop{\to}_{\delta\to 0}\ \,\tfrac{\sqrt{3}}{2\pi}\Re[\,\overline{\eta}^3\partial_{w}\fmu(w',w)\,]\,.
\label{eq_x_conv_TT_ds}
\end{align}

In order to handle the mixed discrete derivative~$[\dOne\dTwo\Fdelta]$, note that the last convergence (applied to both~$\wptau_{\delta,\star}$ and~$\wpmtau_{\delta,\star}$) also implies
\begin{align*}
\delta^{-2}\!\cdot[\dOne\Fdelta](\wrho_\delta,z(\wptau_{\delta,\star}))\ 
\mathop{\to}_{\delta\to 0}&\ \tfrac{\sqrt{3}}{2\pi} \left(\mu\Re[\,\overline{\eta}^3\partial_{w}\fmu(w',w)\,] \pm i\mu \Re[\,\overline{\eta}^3\partial_{w}f_{\Omega}^{[\pm i\mu]}(w',w)\,]\right) \\
 = &\ {\tfrac{\sqrt{3}}{4\pi}}\left(\overline{\eta}^3\partial_{w}\f(w',w)+\eta^3\overline{\partial_{w}\fstar(w',w)}\right).
\end{align*}
Applying~(\ref{eq_F_real_via_complex}) and Lemma~\ref{lemma:conv-derivatives} once again, we obtain the convergence
\begin{align}
\delta^{-3}\!\cdot[\dOne\dTwo\Fdelta](\wrho_\delta,\wptau_\delta)\ \mathop{\to}_{\delta\to 0}\ & \ \tfrac{\sqrt{3}}{4\pi}\Re\left[\overline{\mu}\cdot (i\tau)\partial_{w'}\big(\overline{\eta}^3\partial_{w}\f(w',w)+\eta^3\overline{\partial_{w}\fstar(w',w)}\big)\right] \notag\\
=\ & \tfrac{\sqrt{3}}{4\pi}\Re\left[\overline{\eta}^3\overline{\mu}^3\partial_{w}\partial_{w'}\f(w,w')+ \eta^3\overline{\mu}^3\overline{\partial}_{w}\partial_{w'}\fstar(w,w')\right].
\label{eq_x_conv_TT_dd}
\end{align}
Using Proposition~\ref{prop_via_ferm_tt}(i) and formulae~(\ref{eq_x_conv_TT_ss})--(\ref{eq_x_conv_TT_dd}), we are now able to compute the limit
\begin{align*}
\delta^{-4}\mathrm{Cov}_{\Omega_\delta}^+&\big[\Trho(w_\delta);\Ttau(w'_\delta)\big] \mathop{\to}_{\delta\to 0} \\
  ~\tfrac{9\,}{64\pi^2} \big[-\Re&[\,\overline{\eta}\,\overline{\mu}^3\partial_{w'}\f(w,w')+\eta\overline{\mu}^3\partial_{w'}\fstar(w,w')] \Re[\,\overline{\eta}^3\overline{\mu}\,\partial_{w}\f(w',w)+\overline{\eta}^3\mu\partial_{w}\fstar(w',w)]  \\
 + \Re&[\,\overline{\eta}\,\overline{\mu}\f(w',w)+\overline{\eta}\mu\fstar(w',w)] \Re[\,\overline{\eta}^3\overline{\mu}^3\partial_{w}\partial_{w'}\f(w,w')+ \eta^3\overline{\mu}^3\overline{\partial}_{w}\partial_{w'}\fstar(w,w')]\big],
\end{align*}
where~$\mathrm{Cov}_{\Omega_\delta}^+[X;Y]:=\Edelta^+[XY]-\Edelta^+[X]\Edelta^+[Y]$, and then 
arrive at the formula
\begin{align*}
\delta^{-4}\mathrm{Cov}_{\Omega_\delta}^+\big[\Trho(w_\delta)&+\Tmrho(w_\delta)\,;\,\Ttau(w'_\delta)+\Tmtau(w'_\delta)\big]\ \mathop{\to}_{\delta\to 0} \\
 \tfrac{9}{32\pi^2}\Re\big[\,&\overline{\eta}^4\overline{\mu}^4\big(\partial_{w'}\f(w,w')\partial_{w}\f(w,w') - \f(w,w')\partial_w\partial_{w'}\f(w,w')\big)\\
  - & \eta^4\overline{\mu}^4\big(\partial_{w'}\fstar(w,w')\overline{\partial}_{w}\fstar(w,w')- \fstar(w',w)\overline{\partial}_{w}\partial_{w'}\fstar(w,w')\big)\big].
\end{align*}

Recall that~$\overline{\eta}^4=-\rho^2$,~$\overline{\mu}^4=-\tau^2$ and ~$T(w_\delta)=-\frac{2}{3}\sum_{\rho\in\wp^+}\overline{\rho}^2[\Trho(w_\delta)+\Tmrho(w_\delta)+\Rrho(w_\delta)]$. As in the proof of Theorem~\ref{thm:Ising1}, it is not hard to see that all expectations containing the remainders~$\Rrho(w_\delta)$ and~$\Rtau(w'_\delta)$ have smaller orders of magnitude. Since
\[
\textstyle\frac{4}{9}\sum_{\rho,\tau\in\wp^+}\overline{\rho}^2\overline{\tau}^2\Re[\,\rho^2\tau^2A-\overline{\rho}^2\tau^2 B\,]~=~2A,
\]
we conclude that
\begin{align*}
\delta^{-4}\mathrm{Cov}_{\Omega_\delta}^+\big[T(w_\delta)\,;\,T(w'_\delta)\big]\ \mathop{\to}_{\delta\to 0}\ \tfrac{9\,}{16\pi^2}[\partial_{w'}\f(w,w')\partial_{w}\f(w,w')-\f(w,w')\partial_{w}\partial_{w'}\f(w,w')]\big]
\end{align*}
The right-hand side is equal to~$\tfrac{9}{\pi^2}[\CorrO{T(w)T(w')}-\CorrO{T(w)}\CorrO{T(w')}]$\,, see Definition~\ref{def:TT-and-Te-correlations}.
\end{proof}

\begin{proof}[Proof of Theorem~\ref{thm:Ising3} for the mixed correlation of~$T$ and~$\varepsilon$] The proof repeats the proof of Theorem~\ref{thm:Ising2}. Due to Corollary~\ref{cor:energy-density-convergence}, it is enough to check that
\begin{equation}
\label{eq:Te_conv}
\delta^{-3}\mathrm{Cov}^+_{\Omega_\delta}\big[T(w_\delta)\,;\,\varepsilon(a_\delta)\big]\ \mathop{\to}_{\delta\to 0}\ \tfrac{3\sqrt{3}}{\pi^2}\big[\CorrO{T(w)\varepsilon(a)}-\CorrO{T(w)}\CorrO{\varepsilon(a)}\big]\,.
\end{equation}
As above, one can see that~$\Rrho(w_\delta)$ does not contribute to the limit. Let~$\tau$ denote the direction of~$a_\delta$ and~$\mu:=\eta(\tau)$, note that we can assume~$\tau\in\wp^+$ since~$\varepsilon(a_\delta)=\varepsilon(\overline{a}_\delta)$. Using Proposition~\ref{prop_via_ferm_tt}(ii) together with~(\ref{eq_x_conv_TT_ss}) and~(\ref{eq_x_conv_TT_ds}) applied to~$a_\delta$,~$\overline{a}_\delta$ instead of~$\wptau_{\delta,\star}$, we get
\begin{align*}
\delta^{-3}\mathrm{Cov}^+_{\Omega_\delta}\big[& \Trho(w_\delta)\,;\,\varepsilon(a_\delta)\big]~\mathop{\to}_{\delta\to 0}\\
  &{\tfrac{3\sqrt{3}}{4\pi^2}}\big[\Re[\overline{\eta}\fmu(a,w)]\Re[\overline{\eta}^3\partial_w\f^{[i\mu]}(a,w)] - \Re[\overline{\eta}\f^{[i\mu]}(a,w)]\Re[\overline{\eta}^3\partial_w\fmu(a,w)]\big].
\end{align*}
Adding a similar formula for~$\Tmrho(w_\delta)$ and using~(\ref{eq_x_eta_eta3=eta4}) we obtain
\begin{align*}
\delta^{-3}\mathrm{Cov}^+_{\Omega_\delta}\big[& \Trho(w_\delta)+\Tmrho(w_\delta)\,;\,\varepsilon(a_\delta)\big]\\
~\mathop{\to}_{\delta\to 0}~& -\tfrac{3\sqrt{3}}{4\pi^2}\Re\big[\rho^2\cdot\big(\fmu(a,w)\partial_w\f^{[i\mu]}(a,w)-\f^{[i\mu]}(a,w)\partial_w\fmu(a,w)\big)\big]\\
=~ &-\tfrac{3\sqrt{3}}{8\pi^2}\Re\big[\rho^2\cdot i\big(\f(a,w)\partial_w\fstar(a,w)-\fstar(a,w)\partial_w\f(a,w)\big)\big]\\
=~ &-\tfrac{3\sqrt{3}}{\pi^2}\Re\big[\rho^2\cdot \big(\CorrO{T(w)\varepsilon(a)}-\CorrO{T(w)}\CorrO{\varepsilon(a)}\big)\big],
\end{align*}
see Definition~\ref{def:TT-and-Te-correlations}. The identity~$\frac{2}{3}\sum_{\rho\in\wp^+}\overline{\rho}^2\Re[\rho^2 A]= A$ concludes the proof of~(\ref{eq:Te_conv}).
\end{proof}

\begin{proof}[Proof of Theorem~\ref{thm:Ising3} for the mixed correlation of~$T$ and~$\sigma$]
This proof mimics the proof of the first part of Theorem~\ref{thm:Ising1} (convergence of expectations~$\delta^{-2}\Edelta^+[\Teta(w_\delta)]$) with the following modifications. First, one should use Proposition~\ref{prop_via_fermion_ts} instead of Proposition~\ref{prop_via_ferm_tt} to express these ratios in terms of the local values of spinor observables~$F_{[\Omega_\delta,u_\delta]}$. Second, one should use Theorem~\ref{thm:convergence-spinors} instead of Theorem~\ref{thm:convergence-fermions} to show the convergence of these discrete spinor observables to their continuous counterparts. Taking into account Definition~\ref{def:Tsigma-correlation}, the computations needed to derive the convergence of the ratios~$\delta^{-2}\Edelta^+[\Teta(w_\delta)\sigma(u_\delta)]/\Edelta^+[\sigma(u_\delta)]$ just repeat those in the proof of Theorem~\ref{thm:Ising1}.
\end{proof}

\section{Loop $O(n)$ model}
\label{sec:o-n}

We will prove Proposition~\ref{prop:relation-Te} by taking derivative in the Yang--Baxter equation for the loop~$O(n)$ model.
We first introduce some notation.

\begin{figure}
	\begin{center}
		\includegraphics[width=0.7\textwidth]{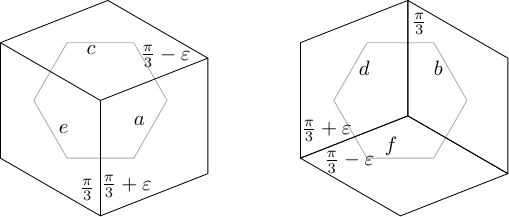}
		\caption{Graphs~$\OmegaE(a,c;\tfrac\pi{3}+\varepsilon)$ and~$\OmegaE(d,f;\tfrac\pi{3}+\varepsilon)$.}
		\label{fig:edges-transform}
	\end{center}
\end{figure}

Let~$G$ be an abstract planar tiling consisting of equilateral triangles and rhombuses: 
this is a planar graph whose inner faces are triangles or quadrangles; 
assign to each triangle angles~$(\pi/3,\pi/3,\pi/3)$; 
assign to each quadrangle angles~$(\varphi,\pi-\varphi,\varphi,\pi-\varphi)$ in a cyclic order (where~$\varphi$ might be different for different quadrangles).
The boundary of~$G$ is defined as the set of all edges belonging to only to one face.

Define~$\conf_G^\varnothing$ as the set of all loop configurations that in each triangle or rhombus look like on Fig.~\ref{figConfE} (top layer) and contain only loops.
A measure~$\mathrm{Loop}_{G,\sigma}$ is supported on~$\conf_G^\varnothing$ and is defined by
\[
	\mathrm{Loop}_{G,n,\sigma}^\varnothing(\omega) = \tfrac{1}{\cZ_G^\varnothing(\sigma)}n^{\#\mathrm{loops}(\omega)}\prod_{f\in \mathrm{Faces}(G)} \mathrm{w}_f(\omega),
\]
where the product is taken over all faces~$f$ of~$G$ and $\mathrm{w}_f(\omega)$ denotes the weight of a local configuration of~$\omega$ inside~$f$ as prescribed by Fig.~\ref{figConfE} and Eq.~\eqref{eq:weights}. Similarly, for~$e,e'$ two midpoints of different boundary edges of~$G$, a measure~$\mathrm{Loop}_{G,\sigma}^{e,e'}$ is supported on all configurations that consist of loops and a unique interface between~$e$ and~$e'$ and is defined as above.

As in Definition~\ref{def:invariance-to-transform}, we say that~$\mathrm{Loop}_{G,n,\sigma}^\b$ is invariant to some local transformation of the tiling~$G_P\mapsto G_{\tilde{P}}$ if, for any given connectivity pattern~$\xi$ on the boundary edges of~$P$,
\[
	\mathrm{Loop}_{G,n,\sigma}^\b(\cdot \, \vert \, \omega_{|P}\in\xi) =
	\mathrm{Loop}_{\tilde{G},n,\sigma}^\b(\cdot \, \vert \, \omega_{|P}\in\xi).
\]

\begin{proposition}\label{prop:inv-to-transform-o-n}
	The invariance relations (T1), (T2), (T3) of Proposition~\ref{prop:inv-to-transform} hold for the loop~$O(n)$ model with~$n=-2\cos [(1-\sigma)\frac{4\pi}{3}]$ and the local weights given by~\eqref{eq:weights}, for any~$\sigma\in \bbR$.
	Moreover, the transformations (T1) and (T2) preserve the partition function.
\end{proposition}

The Yang--Baxter equation was first discovered by Nienhuis~\cite{Nie90}.
The two remaining relations (T1) and (T3) can be proven in the same way as in Proposition~\ref{prop:inv-to-transform} ($n=1$) by considering all possible local configurations and we omit the details.

\begin{proof}[Proposition~\ref{prop:relation-Te}]
	Recall the tiling~$\OmegaE(a;\varphi)$ introduced in Section~\ref{sub:graphical}: it is obtained from the dual triangulation~$\OmegaE$ by replacing the replacing the triangles corresponding to the endpoints of~$a$ by a rhombus of angle~$\varphi$.
	We now introduce~$\OmegaE(a,c;\varphi)$, which is obtained from~$\OmegaE(a,c;\varphi)$ by replacing the triangles corresponding to the endpoints of~$c$ by a rhombus of angle~$\pi-\varphi$; see Fig.~\ref{fig:edges-transform}.
	As discussed in Section~\ref{sub:results-o-n},
	\[
		\Te^\b(a)-\Te^\b(c) = (\log \cZ_{\OmegaE(a;\tfrac\pi{3}+\varepsilon)}^{\b})'-(\log \cZ_{\OmegaE(c;\tfrac\pi{3}+\varepsilon)}^{\b})'=(\log \cZ_{\OmegaE(a,c;\tfrac\pi{3}+\varepsilon)}^{\b})',
	\]
	where the derivatives are evaluated at~$\varepsilon=0$. Define~$\OmegaE(d,f;\varphi)$ similarly and get
	\[
		\Te^\b(a)-\Te^\b(c) =(\log \cZ_{\OmegaE(d,f;\tfrac\pi{3}+\varepsilon)}^{\b})'.
	\]
	By the Yang--Baxter equation ((T2) in Proposition~\ref{prop:inv-to-transform-o-n}), $\cZ_{\OmegaE(a,c;\tfrac\pi{3}+\varepsilon)}^{\b} = \cZ_{\OmegaE(d,f;\tfrac\pi{3}+\varepsilon)}^{\b}$, whence~$\Te^\b(a)-\Te^\b(d) = \Te^\b(c)-\Te^\b(f)$.
	The other relation is analogous.
\end{proof}

\appendix

\section{Full-plane s-holomorphic observable on~$\Hex$.}
In this part of the appendix we review the results concerning discrete s-holomorphicity (strong holomorphicity)
of the fermionic observable, contruct the full-plane fermionic observable on the honeycomb lattice
and compute its singularity.
We start by formulating the result that the fermionic observable is s-holomorphic everywhere
except for the origin.
Such a function is unique up to a multiplicative constant.
It can be constructed by its projections which are equal to particular linear combinations
of the discrete Green's functions (see~\cite{KenWil15}).
Knowing the asymptotic of the Green's function, one can easily compute the asymptotic of the
fermionic observable.
This gives the constant~$\tfrac{\sqrt{3}}{2\pi}$ in Theorem~\ref{thm:convergence-fermions}.

The next proposition states that the fermionic observable~$\Fdelta(a,\cdot)$ is s-holomorphic everywhere
except for two pairs of edges near~$a$, where it has a defect~1.

\begin{proposition}
\label{prop:s-hol}
Let~$a$ be a half-edge inside~$\Omega_\delta$.
Given two adjacent edges~$e$ and~$e'$ of~$\Omega_\delta$,
denote by~$z_e$ and~$z_{e'}$ their midpoints.
Then~$\Fdelta(a,\cdot)$ satisfies the following properties:
\begin{itemize}
\item[--] If none of~$e$ and~$e'$ contains~$a$, then the s-holomorphicity condition~\eqref{eq:s-hol-rel}
holds for~$\Fdelta(a,z_e)$ and~$\Fdelta(a,z_{e'})$;
\item[--] If~$e$ contains~$a$ and~$e'$ has a common vertex with~$a$, then the s-holomorphicity condition~\eqref{eq:s-hol-rel}
holds for~$\Fdelta(a,z_e)+1$ and~$\Fdelta(a,z_{e'})$;
\item[--] If~$e$ contains~$a$ and~$e'$ does not have a common vertex with~$a$, then
the s-holomorphicity condition~\eqref{eq:s-hol-rel} holds for~$\Fdelta(a,z_e)$ and~$\Fdelta(a,z_{e'})$.
\end{itemize}
\end{proposition}

This proposition was proven in~\cite{CheSmi12} (Proposition~2.5) in a general case of the isoradial graphs,
and we do not give any details about the proof.
See also~\cite{Smi10} (Lemma~4.5) for the same result about the fermionic observable
in FK-Ising.

In the continuous case, notions of holomorphicity and harmonicity are very well related~---
the projections of a holomorphic functions are harmonic, and vice versa,
if a function has harmonic projections, then it is holomorphic.
The next lemma states that the same holds on a discrete level.
First we define the corners of the lattice.

\begin{center}
\begin{figure}
\tempskiped{\includegraphics[scale=1.1]{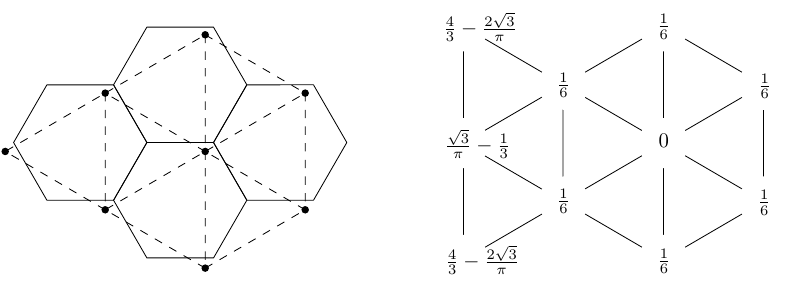}}
\caption{
{\em Left:} One type of corners of the hexagonal lattice of the mesh size~$\delta$
forms the triangular lattice of the mesh size~$\sqrt{3}\delta$.
A direction~$\eta\in \wp$ associated to the shown type of corners is~$\rho$.
{\em Right:} Local values of the Green's function (see~\cite{KenWil15}).
}
\label{fig:Green_function}
\end{figure}
\end{center}

\begin{definition}
By a {\em corner} of a lattice we mean a pair of adjacent edges.

For each corner~$c=(e,e')$ we associate the unique number~$\eta(c)\in \wp$,
such that~$\overline{\eta(c)}^2$ has the same argument as~$u-v$,
where~$v$ denotes the common vertex of
the edges~$e$ and~$e'$, a point~$u$ is the center of the face adjacent to both~$e$ and~$e'$.

According to their direction, the corners can be divided into~6 types.
Each of these types forms the triangular lattice of the mesh size~$\sqrt{3}\delta$ (see Fig.~\ref{fig:Green_function}).
\end{definition}

\begin{lemma}
\label{lem:harmonic_projections}
Let~$F$ be an s-holomorphic function on the midedges of~$\Omega_\delta$.
Then one can define six discretely harmonic functions~$F^\eta$, where~$\eta\in\wp$,
on the vertices of~$\Omega_\delta$, so that for any corner~$c=(e,e')$ the following relation holds:
$$
F^{\eta(c)}(v) = \mathrm{Proj}[F(z_e)\,;\eta(c)\R]=\mathrm{Proj}[F(z_{e'})\,;\eta(c)\R],
$$
where~$z_e$ and~$z_{e'}$ denote the midpoints of the edges~$e$ and~$e'$
and~$v$ denotes the common vertex of~$e$ and~$e'$.
\end{lemma}

Given an s-holomorphic function~$F$, we will call the funcitons~$F^\eta$ for~$\eta\in\wp$
the projections of~$F$ on the corresponding corners.

\begin{remark}
\label{rem:s-holomorphicity_corners}
Using this terminology, the s-holomorphicity is equivalent to the fact these
projections on the corners are well defined, i.e. the projections from the adjacent
edges on the corner formed by them are the same.
\end{remark}

Our next goal is to construct the so-called full-plane fermionic observable~$F_{\C}(a,\cdot)$,
i.e. the function on the edges of the whole hexagonal lattice~$\C_\delta$ which is s-holomorphic
everywhere except for the neighbourhood of~$a$ (in the sense of Proposition~\ref{prop:s-hol}).
We construct~$F_{\C}(a,\cdot)$ by means of its projections.
The projections of a holomorphic function are harmonic, so the projections of a
holomorphic function with a defect are harmonic functions with some defects.
The latter can be obtained as linear combinations of the Green's functions
(on the triangular lattice).

\begin{center}
\begin{figure}[ht]
\tempskiped{\includegraphics[scale=1.1]{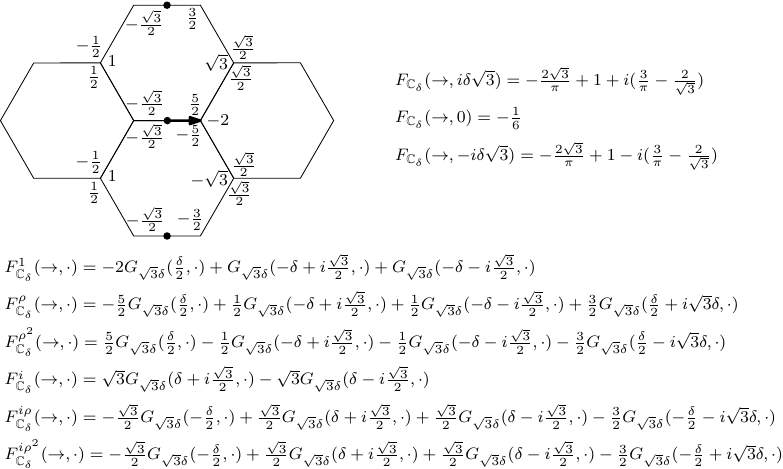}}
\caption{
Full-plane fermionic observable~$F_{\C_\delta}(\rightarrow, \cdot)$, where~$\rightarrow$ stands
for the right half-edge of the horizontal edge which has its midpoint at~0.
{\em Up-left:} Graphical representation of the singularities of the funcitons~$F_{\C_\delta}^{\eta}$
on the corners of the lattice.
The singularities for each type of corners are put in the corresponding corners of this type.
{\em Up-right:} Local values of~$F_{\C_\delta}(\rightarrow, \cdot)$.
{\em Bottom:} Explicit expressions of the functions~$F_{\C_\delta}^{\eta}$ in terms
of the Green's function.
The coefficient in front of the Green's function at a given point is prescribed by the
value of singularity of~$F_{\C_\delta}^{\eta}$ at this corner.
}
\label{fig:full_plane_projections}
\end{figure}
\end{center}

\begin{definition}
Let~$v_0$ be a vertex of the triangular lattice~$\T_\delta$.
The {\em Green's function}~$G_\delta(v_0,\cdot)$ is a (unique) function on vertices of~$\T_\delta$
with the following properties:
\begin{itemize}
\item $G_\delta(v_0,v_0) = 0$;
\item $\Delta G_\delta(v_0,v) = \delta_{v,v_0}$
\item $G_\delta(v_0,v) = O(\log |v_0 - v|)$,
\end{itemize}
where by~$\Delta$ we denote the discrete Laplacian, i.e.~$\Delta F(v) =\sum_{v'\sim v}{(F(v') - F(v))}$,
and~$\delta_{v,v_0}$ is equal to~1 if~$v=v_0$ and~0 otherwise.
\end{definition}

\begin{remark}
It is well known (\cite{Ken02,Spi01}) that the Green's function~$G_\delta(v_0,\cdot)$ exists and unique,
and that its asymptotic on the triangular lattice looks as follows:
\begin{align}
\label{eq:greens_function}
G_\delta(v_0,v) = \tfrac{1}{2\pi\sqrt{3}}\log \tfrac{|v_0 - v|}{\delta} + \tfrac{\gamma_{Euler}}{2\pi} + O(\tfrac{\delta}{|v_0 - v|}).
\end{align}
Note that in~\cite{Ken02} the discrete Laplacian is normalized differently.
Thus, the additional factor~$\frac{1}{\sqrt{3}}$ appears in our case.
\end{remark}

\begin{proposition}
\label{prop:full-plane_existence}
Let~$a$ be a half-edge of the hexagonal lattice.
There exists a unique function~$F_{\C_\delta}(a,\cdot)$ defined on all the edges of the hexagonal lattice
which is decreasing on~$\infty$ and which is s-holomorphic everywhere except for the neighbourhood of~$a$,
where it has a defect~1 (i.e. satisfies the relations from Proposition~\ref{prop:s-hol}).

Moreover, the projections of~$F_{\C_\delta}(a,\cdot)$ on the corners of the lattice can be expressed
in terms of the Green's function and the values of~$F_{\C_\delta}(a,\cdot)$ in the neighbourhood of~$a$
can be computed explicitly (see Fig.~\ref{fig:full_plane_projections}, where the case of a horizontal edge is treated).

\begin{proof}
Assume that~$a$ is the right half-edge of a horizontal edge contatining~0 (all other cases can be treated in the same way).
Take six funcitons~$F_{\C_\delta}^\eta(a,\,\cdot\,)$, where~$\eta\in\wp$,
as they are defined in Fig.~\ref{fig:full_plane_projections}.
In order to show the existence of such a function~$F_{\C_\delta}(a,\,\cdot\,)$,
it is enough to show that for any edge not containing~$a$ its four projections are coherent,
and that for the edge containing~$a$
the value defined by the projections~$F_{\C_\delta}^\rho(a,-\tfrac{\delta}{2})$ and~$F_{\C_\delta}^{\rho^2}(a,-\tfrac{\delta}{2})$,
and the value defined by the projections~$F_{\C_\delta}^{i\rho}(a,-\tfrac{\delta}{2})$ and~$F_{\C_\delta}^{i\rho^2}(a,-\tfrac{\delta}{2})$,
differ by~1.

We will investigate the case of horizontal edges to obtain a precise value of the singularity
at the origin.
Other edges can be studied in a similar manner.
Denote by~$\tilde{F}(z)$ a function on the midpoints of the horizontal edges which projections on the corresponding directions
are~$F_{\C_\delta}^\rho(a,z+\tfrac{\delta}{2})$ and~$F_{\C_\delta}^{\rho^2}(a,z+\tfrac{\delta}{2})$.
And denote by~$\tilde{\tilde{F}}(z)$ a function on the midpoints of the horizontal edges which projections on the corresponding directions
are~$F_{\C_\delta}^{i\rho}(a,z-\tfrac{\delta}{2})$ and~$F_{\C_\delta}^{i\rho^2}(a,z-\tfrac{\delta}{2})$.
Our goal is to show that~$\tilde{F}(z_e) = \tilde{\tilde{F}}(z_e)$ if~$z_e \ne 0$
and that~$\tilde{F}(0) = \tilde{\tilde{F}}(0) + 1$.

It is easy to see that
\begin{align}
\tilde{F}(z) &=&
\left[ F_{\C_\delta}^\rho(a,z+\tfrac{\delta}{2}) - F_{\C_\delta}^{\rho^2}(a,z+\tfrac{\delta}{2})\right]
&+\tfrac{i}{\sqrt{3}}\left[ F_{\C_\delta}^\rho(a,z+\tfrac{\delta}{2}) + F_{\C_\delta}^{\rho^2}(a,z+\tfrac{\delta}{2}) \right]\\
\tilde{\tilde{F}}(z) &=&
-\tfrac{1}{\sqrt{3}}\left[ F_{\C_\delta}^{i\rho}(a,z-\tfrac{\delta}{2}) + F_{\C_\delta}^{i\rho^2}(a,z-\tfrac{\delta}{2})\right]
&+\left[ F_{\C_\delta}^{i\rho}(a,z-\tfrac{\delta}{2}) - F_{\C_\delta}^{i\rho^2}(a,z-\tfrac{\delta}{2}) \right].
\end{align}
Now we plug in the linear combinations of the Green's functions from Fig.~\ref{fig:full_plane_projections}
instead of~$F_{\C_\delta}^\eta(a,\,\cdot\,)$.
Using the translation invariance of the Green's function, we replace everything by the values
of the Green's function~$G(\,\cdot\,) =G_{1}(0,\,\cdot\,)= G_{\sqrt{3}\delta}(0,\,\cdot\,\sqrt{3}\delta)$ with the source at~0:
\begin{align*}
\tilde{F}(z) &=
\left[ -5G(\tfrac{z}{\sqrt{3}\delta}) + G(\tfrac{z}{\sqrt{3}\delta}+\lambda - i) + G(\tfrac{z}{\sqrt{3}\delta}+\lambda) + \tfrac{3}{2}G(\tfrac{z}{\sqrt{3}\delta}-i)
+ \tfrac{3}{2}G(\tfrac{z}{\sqrt{3}\delta}+i) \right]\\
&-\tfrac{i\sqrt{3}}{2}\left[ G(\tfrac{z}{\sqrt{3}\delta}+i) - G(\tfrac{z}{\sqrt{3}\delta}-i)\right],\\
\tilde{\tilde{F}}(z) &=
\left[G(\tfrac{z}{\sqrt{3}\delta}) -  G(\tfrac{z}{\sqrt{3}\delta}-\lambda) -  G(\tfrac{z}{\sqrt{3}\delta}+i-\lambda)
+\tfrac{1}{2} G(\tfrac{z}{\sqrt{3}\delta}-i) +\tfrac{1}{2} G(\tfrac{z}{\sqrt{3}\delta}+i)  \right]\\
&-\tfrac{i\sqrt{3}}{2}\left[  G(\tfrac{z}{\sqrt{3}\delta}+i) - G(\tfrac{z}{\sqrt{3}\delta}-i) \right],
\end{align*}
where~$\lambda = \tfrac{\sqrt{3}}{2} + \tfrac{i}{2}$ denotes a step ``right-up'' on the triangular
lattice (Fig.~\ref{fig:Green_function}, divided by~$\sqrt{3}$).

Substracting one formula from the other, one gets
\begin{align}
\label{eq:s_hol_defect_horizontal}
\tilde{F}(z) - \tilde{\tilde{F}}(z) =
&-6G(\tfrac{z}{\sqrt{3}\delta}) + G(\tfrac{z}{\sqrt{3}\delta}+\lambda - i) + G(\tfrac{z}{\sqrt{3}\delta}+\lambda) \\
&+G(\tfrac{z}{\sqrt{3}\delta}-i) + G(\tfrac{z}{\sqrt{3}\delta}+i) + G(\tfrac{z}{\sqrt{3}\delta}-\lambda) +  G(\tfrac{z}{\sqrt{3}\delta}+i-\lambda).
\end{align}
Note that the points where the Green's function is evaluated in Eq.~\ref{eq:s_hol_defect_horizontal}
are~$\tfrac{z}{\sqrt{3}\delta}$ and its six neighbours.
By the definition, the Green's function is harmonic everywhere except~0,
where its laplacian is equal to~1.
Hence, one gets the desired relation on~$\tilde{F}(z)$ and~$\tilde{\tilde{F}}(z)$.
\end{proof}
\end{proposition}

\begin{corollary}
Given a discrete domain~$\Omega_\delta$ on the hexagonal lattice and a half-edge~$a$ in~$\Omega_\delta$,
the function~$\Fdelta^\sharp(a_\delta,\,\cdot\,) = \Fdelta(a_\delta,\,\cdot\,) - F_\C(a_\delta,\, \cdot \,)$
is s-holomorphic everywhere in~$\Omega_\delta$.

\begin{proof}
It follows from propositions~\ref{prop:s-hol} and~\ref{prop:full-plane_existence} that the s-holomorphicity
is satisfied on all pairs of edges not containing~$a$.
On the other hand both fermionic observables~$\Fdelta (a, \,\cdot\,)$ and~$F_{\C_\delta}$ have the same singularity at~$a$.
Thus, their difference~$\Fdelta^\sharp(a_\delta,\,\cdot\,)$ does not have any singularity at~$a$.
\end{proof}
\end{corollary}

\begin{proposition}
\label{prop:full-plane_convergence}
Given two distinct points~$a,z\in\C$ and~$\eta\in\wp$,
let~$a_\delta$~be the closest to~$a\in\Omega$ half-edge of~$\Hex_\delta$ oriented in the direction~$\eta^2$,
and~$z_{e_\delta}$ be the closet to~$z$ midedge of~$\Hex_\delta$.
Then, the following convergence results hold as~$\delta\to 0$:
\begin{align*}
\delta^{-1}\cdot F_{\C_\delta}(a_\delta,z_{e_\delta})&\ \rightrightarrows\ \tfrac{\sqrt{3}}{2\pi z}.
\end{align*}
Moreover, the convergence is uniform (with respect to~$(a,z)$) on compact subsets
of~$\C\setminus D$, where~$D:=\{(a,a),a\in\C\}$.

\begin{proof}
Assume that~$a_\delta$ is the right half-edge of a horizontal edge containing~0
and~$z_\delta$ is the midpoint of a horizontal edge
(all other cases can be treated in the same way).
While proving Proposition~\ref{prop:full-plane_existence}, we obtained the follwing expression of
the full-plane fermionic observable~$F_{\C_\delta}$
in terms of the values of the Green's function at different points:
\begin{align*}
F(a_\delta,z_\delta) &=
\left[ -5G(\tfrac{z}{\sqrt{3}\delta}) + G(\tfrac{z}{\sqrt{3}\delta}+\lambda - i) + G(\tfrac{z}{\sqrt{3}\delta}+\lambda) + \tfrac{3}{2}G(\tfrac{z}{\sqrt{3}\delta}-i)
+ \tfrac{3}{2}G(\tfrac{z}{\sqrt{3}\delta}+i) \right]\\
&-\tfrac{i\sqrt{3}}{2}\left[ G(\tfrac{z}{\sqrt{3}\delta}+i) - G(\tfrac{z}{\sqrt{3}\delta}-i)\right],
\end{align*}
where~$G(\,\cdot\,) =G_{1}(0,\,\cdot\,)= G_{\sqrt{3}\delta}(0,\,\cdot\,\sqrt{3}\delta)$
and~$\lambda = \tfrac{\sqrt{3}}{2} + \tfrac{i}{2}$.
Clearly, the right-hand side can be expressed in terms of the discrete
derivatives~$\nabla_\nu G(\tfrac{z}{\sqrt{3}\delta}) = G(\tfrac{z}{\sqrt{3}\delta} + \nu) - G(\tfrac{z}{\sqrt{3}\delta})$,
where~$\nu$ takes one of the values~$\pm\lambda$, $\pm i$, $\pm(\lambda - i)$
(which correspond to the directions of the edges):
\begin{align*}
F(a_\delta,z_\delta) &=
\nabla_{\lambda - i} G(\tfrac{z}{\sqrt{3}\delta})
+\nabla_{\lambda} G(\tfrac{z}{\sqrt{3}\delta})
+\left(\tfrac{3}{2} - \tfrac{i\sqrt{3}}{2}\right) \nabla_{i} G(\tfrac{z}{\sqrt{3}\delta})
+\left(\tfrac{3}{2} + \tfrac{i\sqrt{3}}{2}\right) \nabla_{-i} G(\tfrac{z}{\sqrt{3}\delta})\,.
\end{align*}
It is well-known that discrete derivatives of a convergent sequence of discretely harmonic functions
converge to the corresponding derivatives of the limitting function
(for instance, one can see Proposition~$3.1$ in~\cite{CheSmi12}).
Using the asymptotic~\eqref{eq:greens_function} of the Green's function,
we get that on compact subsets of~$\C\setminus \{0\}$
\begin{align*}
\tfrac{1}{\delta}\nabla_\nu G(\tfrac{z}{\sqrt{3}\delta}) &\rightrightarrows \tfrac{1}{2\pi}\Re \left[ \tfrac{\nu}{z}\right].
\end{align*}
Hence,
\begin{align*}
\tfrac{1}{\delta}F(a_\delta,z_\delta) &\rightrightarrows \tfrac{1}{2\pi} \Re \left[
\frac{(\lambda-i) + \lambda + \tfrac{3}{2}i -  \tfrac{3}{2}i }{z}
\right]
+ \tfrac{i}{2} \cdot \tfrac{\sqrt{3}}{2\pi} \Re \left[
\frac{ -i - i }{z}
\right] = \tfrac{\sqrt{3}}{2\pi z}\,.
\end{align*}
\end{proof}
\end{proposition}

\section{Star-segment transformation}
\label{sec:star-segment}

Here we give some details regarding (T2) in Proposition~\ref{prop:inv-to-transform}: invariance of the Ising model defined by~\eqref{eq:o-1-def} to the star-segment transformation when the weights are given by~\eqref{eq:weights}.

Take any~$\varphi \in (0,\pi/3)$ and define~$c(\varphi):= 1+ u(\pi-\varphi)u(2\pi/3+\varphi)x_c$, with the weights given by~\eqref{eq:weights}.
All boundary conditions are considered on Fig.~\ref{fig:pentagon} and below we give the corresponding relations on the weights:
\begin{align*}
	1\cdot c(\varphi) &\stackrel{1}{=} 1+ u(\pi-\varphi)u(2\pi/3+\varphi)x_c, \\
	u(2\pi/3-\varphi)x_c\cdot c(\varphi) &\stackrel{2}{=} u(\pi-\varphi)+ w(\varphi)u(2\pi/3+\varphi)x_c, \\
	v(\pi/3+\varphi)x_c\cdot c(\varphi) &\stackrel{3}{=} u(\varphi)u(2\pi/3+\varphi)x_c + v(\varphi)x_c, \\
	u(\pi/3+\varphi)x_c\cdot c(\varphi) &\stackrel{4}{=} u(\varphi)v(\pi/3-\varphi) + v(\varphi)u(\pi/3-\varphi)x_c, \\
	u(\pi/3+\varphi)\cdot c(\varphi) &\stackrel{5}{=} v(\varphi)u(2\pi/3+\varphi)x_c + u(\varphi)x_c, \\
	v(\pi/3+\varphi)\cdot c(\varphi) &\stackrel{6}{=} v(\varphi)v(\pi/3-\varphi) + u(\varphi)u(\pi/3-\varphi)x_c, \\
	v(\pi/3+\varphi)x_c\cdot c(\varphi) &\stackrel{7}{=} u(\pi-\varphi)u(2\pi/3+\varphi) + w(\varphi)w(\pi/3-\varphi), \\
	u(\pi/3+\varphi)x_c\cdot c(\varphi) &\stackrel{8}{=} u(\pi-\varphi)v(\pi/3-\varphi)x_c + w(\varphi)u(\pi/3-\varphi)x_c, \\
	w(\pi/3+\varphi)x_c\cdot c(\varphi) &\stackrel{9}{=} u(\pi-\varphi)u(\pi/3-\varphi)x_c + w(\varphi)v(\pi/3-\varphi)x_c.
\end{align*}

It is an exercise in trigonometry computations to check that~\eqref{eq:weights} implies these relations.

\begin{figure}
	\begin{center}
		\includegraphics[scale=1.1]{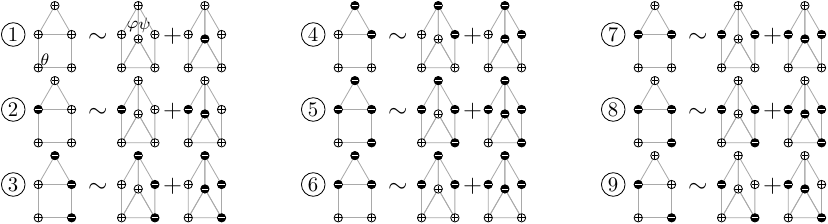}
		\caption{Local configurations before and after the transformation~(T2) under all possible boundary conditions.
		The angles of the rhombuses are~$\varphi$, $\psi:=\pi/3 - \varphi$, $\theta:=\pi/3+\varphi$, the triangles are equilateral.
		}
		\label{fig:pentagon}
	\end{center}
\end{figure}

\bibliographystyle{amsalpha}
\bibliography{biblicomplete}

\end{document}